%% file: main.tex
\documentclass[sigconf, 9pt]{acmart}

\usepackage{filecontents}

\usepackage{graphicx} 
\usepackage{balance}

\usepackage{amssymb}
\usepackage{amsmath}
\usepackage{multirow, rotating, wasysym, url}
\usepackage{array}
\usepackage{booktabs}
\usepackage{subfigure}
\usepackage{hyperref} 
\usepackage{color}
\usepackage[utf8]{inputenc}
\usepackage[english]{babel}
\usepackage{comment}
\usepackage[normalem]{ulem}
\usepackage{enumitem}
\usepackage{float}
\usepackage{url}
\usepackage{epstopdf}
\usepackage[ruled,linesnumbered]{algorithm2e}
\usepackage{algorithmicx}  
\usepackage{pifont}
\usepackage{bbm}
 \usepackage{tikz}
    \newcommand*{\circled}[1]{\lower.7ex\hbox{\tikz\draw (0pt, 0pt)%
    circle (.5em) node {\makebox[1em][c]{\small #1}};}}
\graphicspath{{./DrawingFigures/},{./EvaluateFigures/},{./figure/}} 
	\newcommand{\presec}{\vspace{-0.2cm}}
	\newcommand{\postsec}{\vspace{-0.1cm}}

	\newcommand{\prefig}{\vspace{-0in}}
	
	\newcommand{\postfig}{\vspace{-0.4cm}}

		\input{InputFiles/new-commands}

		\newcommand{\bbb}{\noindent\textbf}
	
	\newtheorem{theorem}{Theorem}[section]

	\definecolor{greener}{RGB}{0,166,0}
	\definecolor{reder}{RGB}{255,0,0}
	\definecolor{bluer}{RGB}{0,0,255}
	
	\definecolor{greener}{RGB}{0,166,0}

	\definecolor{ljz}{RGB}{0,0,255}

	\definecolor{lzr}{RGB}{255,0,0}

	\definecolor{wyh}{RGB}{0,200,0}

    \definecolor{ykc}{RGB}{0,0,255}

   	\definecolor{amethyst}{rgb}{0.6, 0.4, 0.8}

    \newcommand{\hstate}{healthy network state}
    \newcommand{\istate}{ill network state}

	\newcommand{\systemname}{ChameleMon}
	\newcommand{\sketchname}{FermatSketch}
	
\begin{document}

\thanks{\textsuperscript\textasteriskcentered Kaicheng Yang, Yuhan Wu, and Ruijie Miao contribute equally to this paper. Tong Yang (yangtongemail@gmail.com) is the corresponding author.}

\author{{Kaicheng Yang$^{\dag}$}, {Yuhan Wu$^{\dag}$},
\ 
{Ruijie Miao$^{\dag}$},
\ 
{Tong Yang$^{\dag}$},
\ 
{Zirui Liu$^{\dag}$},
\
{Zicang Xu$^{\dag}$}
\\
{Rui Qiu$^{\dag}$},
\ 
{Yikai Zhao$^{\dag}$},
\ 
{Hanglong Lv$^{\dag}$},
\ 
{Zhigang Ji$^{\P}$},
\ 
{Gaogang Xie$^{\S}$}}
\affiliation{{{$^{\dag}$National Key Laboratory for Multimedia Information Processing, School of Computer Science, Peking University}}
\\
{{$^{\P}$Huawei Technologies Co., Ltd.}}
\ \ \ \
{{$^{\S}$CNIC CAS, UCAS}}\country{}}



\acmYear{2023}\copyrightyear{2023}
\setcopyright{acmlicensed}
\acmConference[ACM SIGCOMM '23]{ACM SIGCOMM 2023 Conference}{September 10--14, 2023}{New York, NY, USA}
\acmBooktitle{ACM SIGCOMM 2023 Conference (ACM SIGCOMM '23), September 10--14, 2023, New York, NY, USA}
\acmPrice{15.00}
\acmDOI{10.1145/3603269.3604850}
\acmISBN{979-8-4007-0236-5/23/09}

\settopmatter{printacmref=true}
\renewcommand{\shortauthors}{Kaicheng Yang et al.}
    \title{\textbf{ChameleMon: Shifting Measurement Attention as Network State Changes}}
    \input{NSDI2023/1Abstract} 
\maketitle
\input{SIGCOMM23/newintro}
\input{NSDI2022_9/2_Overview}

\input{NSDI2022_9/3_Dataplane_design}
\input{NSDI2022_9/4_Controlplane_design}

\input{NSDI2022_9/5_Experiments}

\input{NSDI2022_9/6_Related}
\input{NSDI2023/10Conclusion}   
\section*{acknowledgement}
    We would like to thank the anonymous reviewers for their valuable suggestions. This work is supported by National Key R\&D Program of China (No. 2022YFB2901504), and National Natural Science Foundation of China (NSFC) (No. U20A20179).

\clearpage
    {
        \bibliographystyle{unsrt}
	    \bibliography{NSDI2022_9/reference2.bib}
    }
	
    \vfill\eject

\clearpage
\appendix
    \input{NSDI2022_9/Appendix/appendix}

\input{NSDI2022_9/Appendix/setup_exp}

    \input{NSDI2023/8Implementation}

    \input{NSDI2022_9/Appendix/Related_work}

 \clearpage
\end{document}

%% file: InputFiles/new-commands.tex
\mathchardef\Gamma="0100 \mathchardef\Delta="0101
\mathchardef\Theta="0102 \mathchardef\Lambda="0103
\mathchardef\Xi="0104 \mathchardef\Pi="0105
\mathchardef\Sigma="0106 \mathchardef\Upsilon="0107
\mathchardef\Phi="0108 \mathchardef\Psi="0109
\mathchardef\Omega="010A

\newcommand{\outline}[1]{}

\usepackage{xspace}
\usepackage{url}
\usepackage{graphicx}
\usepackage{latexsym}
\usepackage{amssymb}
\usepackage{amsfonts}
\usepackage{psfrag}
\usepackage{wrapfig}
\usepackage{comment}
\usepackage{alltt}
\usepackage{color}

\newcommand{\ie}{\emph{i.e.}\xspace}
\newcommand{\eg}{\emph{e.g.}\xspace}
\newcommand{\etc}{\emph{etc.}\xspace}

%% file: NSDI2023/1Abstract.tex
\presec
\begin{abstract}
\postsec
\textsuperscript\textasteriskcentered Network measurement is critical to many network applications.
There are mainly two kinds of flow-level measurement tasks: 1) packet accumulation tasks and 2) packet loss tasks.
In practice, the two kinds of tasks are often required at the same time, but existing works seldom handle both. 
In this paper, we design \systemname{} to support the two kinds of tasks simultaneously.
The key design of \systemname{} is to shift measurement attention as network state changes, through two dimensions of dynamics:
1) dynamically allocating memory between the two kinds of tasks; 2) dynamically monitoring the flows of importance.
To realize the key design, we propose a key technique, leveraging Fermat's little theorem to devise a flexible data structure, namely \sketchname{}.
\sketchname{} is dividable, additive, and subtractive, supporting the two kinds of tasks.
We have implemented a \systemname{} prototype on a testbed with a Fat-tree topology.
We conduct extensive experiments and the results show \systemname{} supports the two kinds of tasks with low memory/bandwidth overhead, and more importantly, it can automatically shift measurement attention as network state changes.

\end{abstract}

%% file: SIGCOMM23/newintro.tex
\presec 
\section{Introduction} 

\postsec

Network measurement provides critical statistics for various network operations, such as traffic engineering \cite{benson2011microte,feldmann2001deriving}, congestion control \cite{li2019hpcc}, network accounting \cite{cusketch}, anomaly detection \cite{zhang2013adaptive,mai2006sampled,estan2004building,duffield2003estimating}, and failure troubleshooting \cite{handigol2014know,netbouncer2019}.
In earlier years, sampling-based solutions \cite{everflow2015,netflow2004,sflow2001,csamp2008} were widely accepted thanks to their simplicity and ease of use.
Recently, sketch-based solutions \cite{elastic2018,nitrosketch2019,sketchlearn2018} have attracted much more attention than sampling-based ones, as they are designed to approach the ultimate goal of network measurement \cite{univmon2016,beaucoup2020,zhang2021cocosketch}: to support more tasks and achieve higher accuracy with less memory.
The emerging programmable switches that can process packets at terabit line rate further make sketches practical for production deployment, and designing novel sketches for flow-level measurement capabilities on programmable switches has become a hot topic \cite{univmon2016,beaucoup2020,zhang2021cocosketch}.

There are mainly two kinds of flow-level measurement tasks.
The first kind is \textit{packet accumulation tasks} that focus on flow sizes at certain network nodes, including flow size estimation \cite{cmsketch}, heavy-hitter detection \cite{sivaraman2017heavy}, entropy estimation \cite{gu2005detecting}, \etc.
The second kind is \textit{packet loss tasks} that focus on changes of flow sizes between network nodes, among which the most representative one is packet loss detection \cite{lossradar2016}.
However, the two kinds of tasks are seldom considered and supported simultaneously in one solution.
One reason behind is that the two kinds of tasks require very different flow-level statistics.

However, in practice, the two kinds of tasks are often required at the same time, and there are only limited resources for measurement in programmable switches (\eg, O(10MB) SRAM and limited accesses to the SRAM).
Therefore, the \textit{first requirement} for a practical measurement system is versatile to support the two kinds of tasks with high accuracy using limited resources, where limited resources refer to sub-linear space complexity.

Based on the first requirement, the \textit{second requirement} is to pay attention to different kinds of tasks for different network states.
When the network state is healthy and there are only few packet losses in the network, the system should pay more attention (\eg, allocate more memory) to packet accumulation tasks.
When the network state is ill and there are lots of packet losses in the network, the system should pay more attention to packet loss tasks to help diagnose network faults, especially for those flows which experience a great number of packet losses.

In summary, a practical measurement system should meet the following requirements: \textbf{[R$_{1.1}$]} \textit{versatility requirement}: supporting both packet loss tasks and packet accumulation tasks simultaneously; \textbf{[R$_{1.2}$]} \textit{efficiency requirement}: achieving high accuracy with sub-linear space complexity; \textbf{[R$_{2}$]} \textit{attention requirement:} paying attention to different kinds of tasks for different network states.

Existing solutions can be mainly classified into three categories according to supported measurement tasks:
\begin{enumerate}
[leftmargin=*,parsep=0pt,itemsep=0pt,topsep=2pt,partopsep=2pt]

    \item \textit{Solutions for packet loss tasks:}
    Typical solutions including LossRadar \cite{lossradar2016} based on Invertible Bloom filter \cite{eppstein2011s}, Netseer \cite{netseer2020} and Dapper \cite{dapper} based on the advanced features of programmable switches, and more.
    These solutions are often carefully designed to only obtain the exact difference set of flows/packets to minimize measurement overhead, while packet accumulation tasks require approximate sizes of all flows or simply large flows.
    Therefore, these solutions can hardly be extended to packet accumulation tasks and fail to meet \textbf{[R$_{1.1}$]}. 

    \item \textit{Solutions for packet accumulation tasks:}
    These solutions are usually based on sketches, including CM sketch \cite{cmsketch}, UnivMon \cite{univmon2016}, ElasticSketch \cite{elastic2018}, HashPipe \cite{sivaraman2017heavy}, and more.
    To efficiently maintain approximate flow sizes, these solutions choose to embrace hash collisions and select the estimation with least collisions to minimize error.
    For these solutions, due to their inherent error caused by hash collisions, it is difficult to obtain the exact difference set of flows/packets.
    Therefore, these solutions can hardly be extended to packet loss tasks and fail to meet \textbf{[R$_{1.1}$]}.

    \item \textit{Solutions for both kinds of tasks:}
    These solutions support both kinds of tasks by recording exact IDs and sizes of all flows, including FlowRadar \cite{flowradar2016}, OmniMon \cite{omnimon2020}, Counter Braids \cite{lu2008counter}, Marple \cite{marple} and more.
    However, recording exact IDs and sizes of all flows requires at least memory/bandwidth overhead linear with the number of flows.
    Therefore, these solutions fail to meet \textbf{[R$_{1.2}$]}.

\end{enumerate}

In summary, the first two categories of solutions cannot meet \textbf{[R$_{1.1}$]} due to their limited measurement capabilities, and the third category of solutions cannot meet \textbf{[R$_{1.2}$]} due to their linear space complexities.
A naive solution meeting both \textbf{[R$_{1.1}$]} and \textbf{[R$_{1.2}$]} is to combine the first two categories of solutions: choosing one solution in the corresponding category for each kind of tasks.
However, such a combination fails to achieve \textbf{[R$_{2}$]} on programmable switches. 
The reason behind is that the data structures and operations of different categories of solutions usually differ significantly.
For example, LossRadar \cite{lossradar2016} records the IDs and existences of packets using XOR operation and addition, while ElasticSketch \cite{elastic2018} records the IDs and sizes of flows using comparison, substitution, and addition.
Therefore, solutions in different categories can only utilize their resources allocated at compile time, which prohibits flexible allocation of memory resources between packet loss tasks and packet accumulation tasks.
Therefore, the naive solution cannot pay attention to different kinds of tasks for different network states.
%
%




%
%

In this paper, we design \systemname{}, which meets all the above requirements simultaneously.
\systemname{} can support almost all packet loss tasks and packet accumulation tasks.
Compared to the state-of-the-art solutions, for packet loss tasks, \systemname{} reduces the memory overhead from proportional to the number of all flows (FlowRadar) or lost packets (LossRadar), to proportional to the number of flows experiencing packet losses, which we call \textit{victim flows}; 
for packet accumulation tasks, \systemname{} achieves at least comparable accuracy.
Our \systemname{} has a key design and a key technique as follows.

The key design of \systemname{} is to shift measurement attention as network state changes, which is just like the process of the chameleons changing their skin coloration, through two dimensions of dynamics:
1) dynamically allocating memory between the two kinds of tasks; 2) dynamically monitoring the flows of importance.
First, \systemname{} monitors the network state and allocates memory between the two kinds of tasks accordingly.
When the network state is healthy and only a few packet losses occur in the network, \systemname{} pays most attention to and allocates most of the memory for packet accumulation tasks.
As the network state degrades and packet losses increase, \systemname{} gradually shifts measurement attention to and allocates more and more memory for packet loss tasks to assist in diagnosing network faults.
Second, \systemname{} ranks the flows according to their importance, and selects those of most importance to monitor.
When the network state is ill and there are too many victim flows, \systemname{} selects those flows experiencing many packet losses (called \textit{heavy-losses}, HLs for short) to monitor, instead of monitoring all victim flows.
Overall, when the network state continuously degrades from the healthy state to the ill state, \systemname{} runs as follows.
1) As the number of victim flows increases, \systemname{} leverages the first dimension of dynamic: gradually shifting measurement attention to and allocating more and more memory for packet loss tasks;
2) When the victim flows are too many to monitor, \systemname{} leverages the second dimension of dynamic: focusing measurement attention on HLs while monitoring a small portion of other packet losses (named \textit{light-losses}, LLs for short) through sampling.

To realize the key design, \systemname{} incorporates a key technique, leveraging Fermat's little theorem\footnote{Fermat's little theorem states that if $p$ is a prime, then for any integer $a$ that is indivisible by $p$, we have $a^{p-1} \equiv 1 \mod{p}$.} to devise a flexible data structure, namely \textit{\sketchname{}}.
The data structure of \sketchname{} is made of many same units.
\sketchname{} is dividable, additive, and subtractive, supporting packet loss detection and heavy-hitter (HH for short) detection simultaneously.
By dividing \sketchname{} into three parts to detect HLs, LLs, and HHs, \systemname{} can flexibly move the division points to shift attention and allocate memory between the two kinds of tasks as network state changes.
For each incoming packet, We further use a flow classifier (TowerSketch \cite{yang2021sketchint}) to determine which of the three parts to insert.
For packet loss detection, owing to Fermat's little theorem, \sketchname{} only requires memory proportional to the number of victim flows. 
Differently, the state-of-the-art solutions require memory proportional to the number of all flows (FlowRadar) or lost packets (LossRadar).

Thanks to the visibility to per-flow size provided by Towersketch, \systemname{} can support five other widely-studied \cite{elastic2018,univmon2016,song2020fcm,yang2021sketchint} packet accumulation tasks, including flow size estimation, heavy-change detection, flow size distribution estimation, entropy estimation, and cardinality estimation.
We have fully implemented a \systemname{} prototype on a testbed with a Fat-tree topology composed of 10 Tofino switches and 8 end-hosts.
We conduct extensive experiments and the results show that \systemname{} supports both kinds of tasks with low memory/bandwidth overhead, and more importantly, it can \textit{automatically shift measurement attention as network state changes at run-time without recompilation}.
We have released all related source codes at Github \cite{opensource}.

\bbb{Ethics:} This work does not raise any ethical issue.

%% file: NSDI2022_9/2_Overview.tex
 \begin{figure*}[t]
    \setlength{\abovecaptionskip}{0.1cm}
    \setlength{\belowcaptionskip}{-0.35cm}
    \centering  
\includegraphics[width=0.9\linewidth]{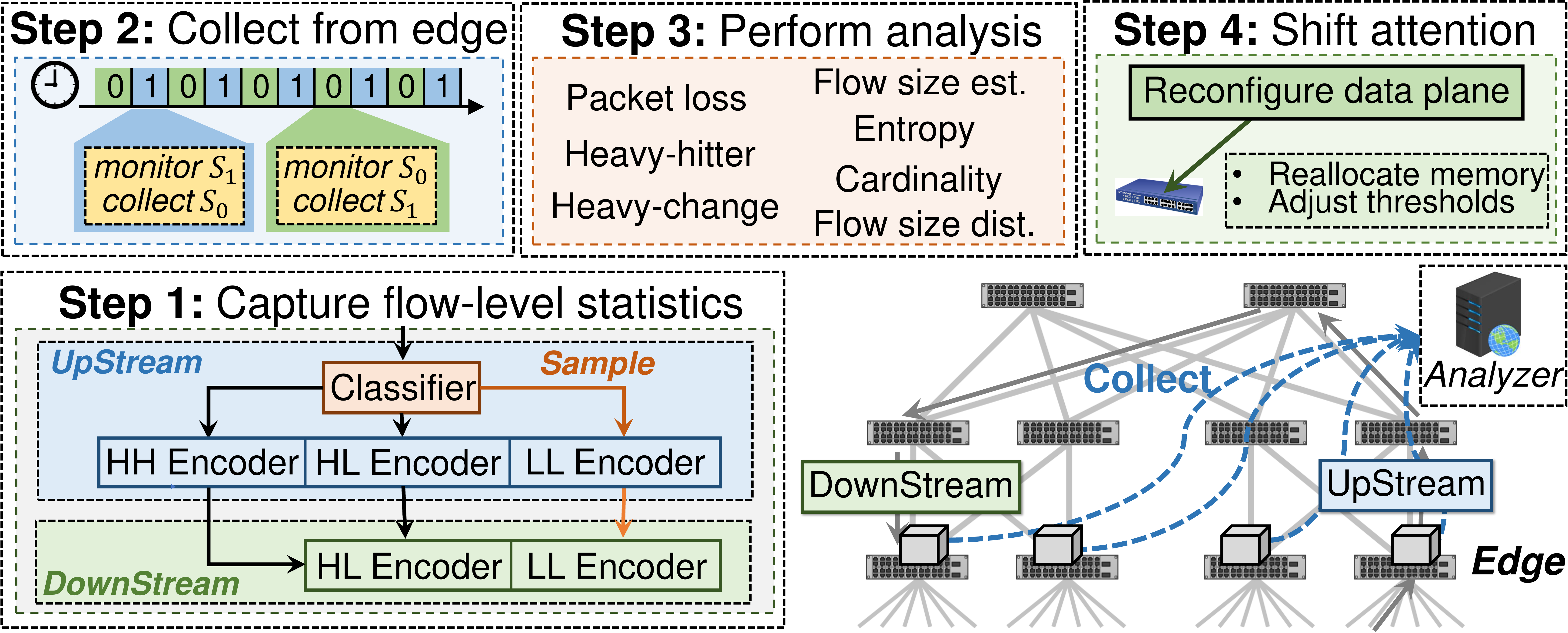}
    \caption{Overview of \systemname{}.}
    \label{fig:overview}
\end{figure*}

 \presec \section{Overview of \systemname{}} \postsec

\systemname{} monitors the network in four steps (Figure \ref{fig:overview}).  

\bbb{1) Capturing flow-level statistics on edge switches:}
To capture desired flow-level statistics, \systemname{} deploys three sketches on the data plane of each edge switch, including a flow classifier (TowerSketch), an upstream flow encoder (our \sketchname{}), and a downstream flow encoder (our \sketchname{}).
To detect HHs, HLs, and LLs, the upstream and downstream flow encoders are divided into multiple parts:
1) the upstream flow encoder is divided into an upstream HH encoder, an upstream HL encoder, and an upstream LL encoder;
2) the downstream flow encoder is divided into a downstream HL encoder and a downstream LL encoder.
For every packet with flow ID $f$ entering the network, according to the size of flow $f$, the flow classifier classifies flow $f$ into one of three hierarchies: 1) HH candidate, 2) HL candidate, or 3) LL candidate.
The LL candidate is further classified into sampled LL candidate or non-sampled LL candidate through sampling.
Based on the hierarchy of flow $f$, the packet is then inserted into the corresponding part of the upstream flow encoder and downstream flow encoder when it enters and exits the network, respectively.

\bbb{2) Collecting sketches from edge switches:}
A central controller periodically collects sketches from each edge switch to support persistent measurement.
To avoid colliding with packet insertion when collecting sketches, each edge switch divides the timeline into consecutive fixed-length time intervals (called \textit{epochs}), and copies a group of sketches for rotation.
Every time an epoch ends, the central controller collects the group of sketches monitoring this epoch, and the other group of sketches starts to monitor the current epoch.
%
%

%
%

\bbb{3) Performing network-wide analysis:}
Every epoch, the central controller performs network-wide analysis of the collected sketches to support seven measurement tasks.
By analyzing the upstream and downstream flow encoders, the central controller can support packet loss detection.
By analyzing the flow classifier and the upstream HH encoder, the central controller can support heavy-hitter detection and five other packet accumulation tasks.
%

%
%

\bbb{4) Shifting measurement attention as network state changes:}
Every epoch, the central controller monitors the real-time network state by analyzing the collected sketches.
Then, the central controller reconfigures the data plane of edge switches at run-time according to the real-time network state, shifting measurement attention through two dimensions of dynamics.
In the first dimension, the central controller dynamically allocates memory between packet loss tasks and packet accumulation tasks by reallocating the memory of the upstream and downstream encoders between their different parts.
In the second dimension, the central controller dynamically selects the most important flows (HH/HL/sampled LL candidates) to monitor by adjusting the thresholds for flow classification and the sample rate for sampling LL candidates.

%% file: NSDI2022_9/3_Dataplane_design.tex


\newcommand{\classifiername}{The flow classifier}

\presec \section{\systemname{} Data Plane} \postsec

The \systemname{} data plane consists of the flow classifier, the upstream flow encoder, and the downstream flow encoder deployed on each edge switch.
In this section, we detail the design of the \systemname{} data plane.
First, we propose the key technique of \systemname{}, namely \textit{\sketchname{}}.
%
Second, we detail each component of the \systemname{} data plane in sequence. 
%

\presec \subsection{The \sketchname{} Algorithm} \postsec
\label{sec:sketch}

\bbb{Rationale:} 
Our primary goal is to detect packet losses with low memory overhead.
Existing solutions focus on either per-packet loss (LossRadar \cite{lossradar2016}) or all-flow visibility (FlowRadar \cite{flowradar2016}), incurring unacceptable memory overhead.
To reduce overhead, we hope to aggregate all the lost packets of the same flow to detect per-flow packet losses.
It is very challenging because existing solutions commonly use XOR operation for high memory efficiency and hardware-friendliness, but simply using XOR operation to aggregate flow IDs of lost packets causes every two lost packets of the same flow to cancel each other out.
While invertible Bloom lookup table (IBLT) \cite{goodrich2011invertibleIBLT} can overcome this challenge as IBLT uses addition to aggregate flow IDs, such design requires computation over large numbers, and thus complicates the implementation of IBLT on programmable switches.
To address this challenge while maintaining hardware-friendliness, we devise \sketchname{}, which uses modular addition to aggregate flow IDs and leverages Fermat's little theorem to extract flow IDs.

\bbb{Data structure (Figure \ref{algopic:insertion}):}
\sketchname{} has $d$ equal-sized bucket arrays $\mathcal{B}_1, \cdots, \mathcal{B}_d$, each of which consists of $m$ buckets.
Each bucket array $\mathcal{B}_i$ is associated with a pairwise-independent hash function $h_i(\cdot)$ that maps each incoming packet into one bucket (called mapped bucket) in it. 
Each bucket $\mathcal{B}_i[j]$ contains two fields: 1) a \textit{count field} $\mathcal{B}_i^{c}[j]$ recording the number of packets mapped into the bucket; 2) an \textit{IDsum field} $\mathcal{B}_i^{ID}[j]$ recording the result of the sum of flow IDs of packets mapped into the bucket modulo a prime $p$.
%
%
At initialization, we set all fields of all buckets in \sketchname{} to zero, and $p$ to a prime that must be larger than any available flow ID $f$ and the size of any flow, so as to make use of Fermat's little theorem.
%

\bbb{Encoding/Insertion operation (Figure \ref{algopic:insertion}):} 
To encode an incoming packet with flow ID $f$, we first calculate the $d$ hash functions to locate $d$ mapped buckets: $\mathcal{B}_1[h_1(f)], \mathcal{B}_2[h_2(f)], \cdots, \mathcal{B}_d[h_d(f)]$. 
For each mapped bucket $\mathcal{B}_i[h_i(f)]$, we update it as follows.
First, we increment its count field $\mathcal{B}_i^{c}[h_i(f)]$ by one. 
Second, we update its IDsum field through \textit{modular addition}: $\mathcal{B}_i^{ID}[f)] \gets (( \mathcal{B}_i^{ID}[h_i(f)] $ $ + f ) \mod{p})$. 
The pseudo-code of encoding operation is shown in Algorithm \ref{alg:insert} in Appendix \ref{sec:pseudo}. 

\begin{figure}[t!]
    \setlength{\abovecaptionskip}{0.15cm}
    \setlength{\belowcaptionskip}{-0.4cm}
    \centering  
\includegraphics[width=0.9\linewidth]{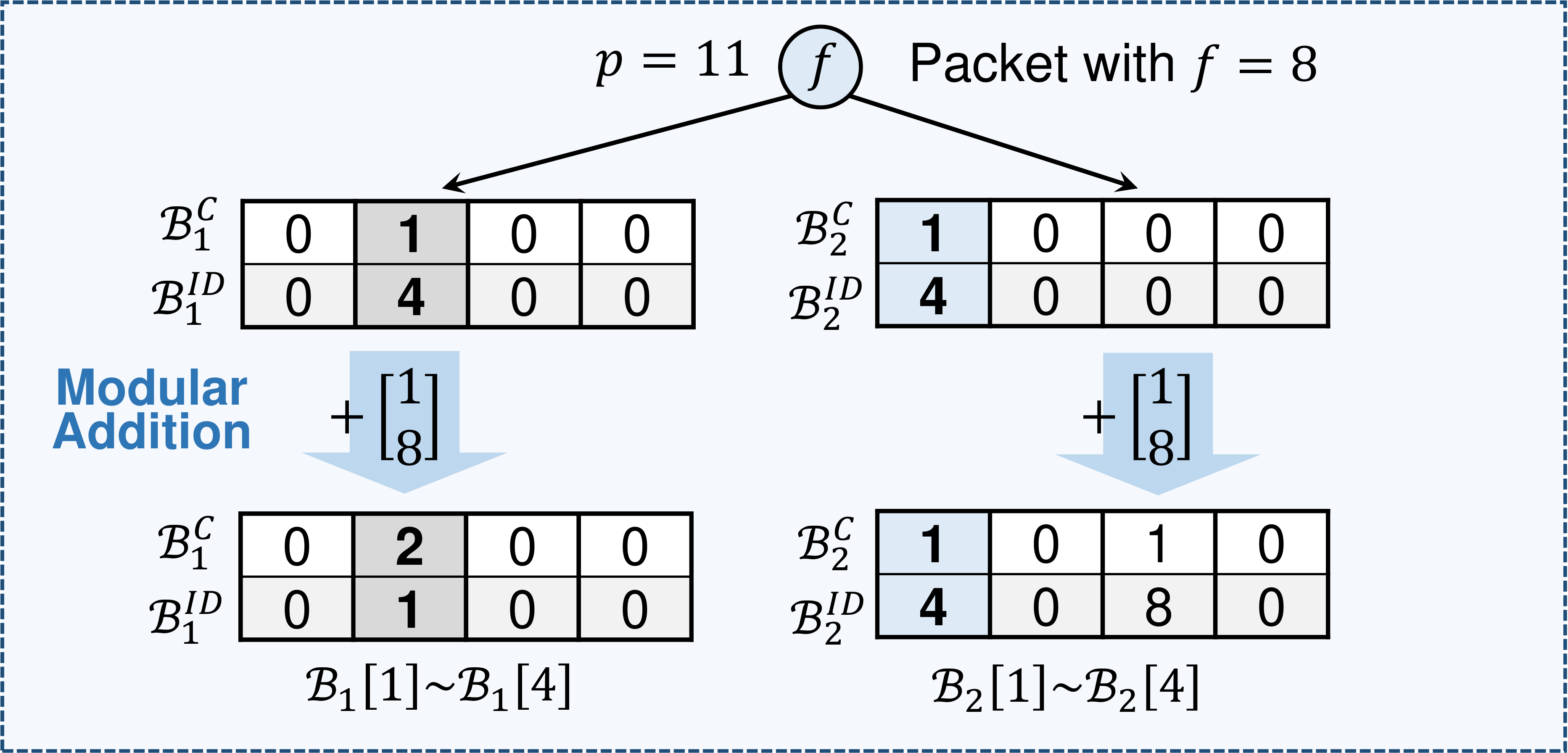}
    \caption{An example of encoding/insertion.}
    \label{algopic:insertion}
\end{figure}



\bbb{Decoding operation:}
The decoding operation, which can extract exact flow IDs and flow sizes from \sketchname{}, has two important suboperations: 1) \textit{pure bucket verification} that verifies whether a bucket only records packets of a single flow (pure bucket); 2) \textit{single flow extraction} that extracts and deletes a single flow and its size from all its mapped buckets.
%
%
Next, we propose the workflow of decoding operation and detail the two suboperations.  
The pseudo-code of decoding operation is shown in Algorithm \ref{alg:decode} in Appendix \ref{sec:pseudo}.

\noindent\textit{$\bullet$ Decoding workflow (Figure \ref{algopic:decoding}):}
decoding proceeds as follows.

\noindent\circled{1} Traverse \sketchname{} and push all non-zero buckets to decoding queue.

\noindent\circled{2} Pop a bucket from queue.

\noindent\circled{3} For the popped bucket, we perform pure bucket verification to verify whether it is a pure bucket.
If not, we simply ignore the bucket and go back to step \noindent\circled{2}.

\noindent\circled{4} If so, we perform single flow extraction to extract and delete a single flow and its size from the pure bucket as well as the other mapped buckets of the single flow.

\noindent\circled{5} We insert the extracted single flow and its size into a hash table, namely \textit{Flowset}, which is used to record all the extracted flows and their sizes.
We regard all flows recorded in Flowset as the flows previously encoded into \sketchname{}.

\noindent\circled{6}Except the pure bucket, we push the other mapped non-zero buckets of the extracted flow into queue.
%

\noindent\circled{7}
Check whether the queue is empty.
If so, the decoding stops.
Otherwise, go back to step \noindent\circled{2}.
After stopping, if there are still non-zero buckets in \sketchname{}, the decoding is considered as failed. Otherwise, the decoding is considered as successful.

\noindent\textit{$\bullet$ Pure bucket verification:}
The pure bucket verification reports whether one given bucket is pure (\ie, only records a single flow), but it may misjudge a non-pure bucket as a pure one with a small probability $\frac{1}{m}$.
Suppose a bucket ${B}_i[j]$ only records a single flow $f^{\prime}$, it should satisfy that $(\mathcal{B}_i^{c}[j] \times f^{\prime}) \mod p = \mathcal{B}_i^{ID}[j]$.
Leveraging Fermat's little theorem, we can get that $f^{\prime} = (\mathcal{B}_i^{ID}[j] \times (\mathcal{B}_i^{c}[j])^{p-2}) \mod p$.
Considering that bucket ${B}_i[j]$ should be one of the $d$ mapped buckets of flow $f^{\prime}$, to verify whether ${B}_i[j]$ is a pure bucket, we propose a verification method namely \textit{rehashing verification}.
First, we calculate the $i^{th}$ hash function $h_i(\cdot)$ to locate the $i^{th}$ mapped bucket of $f^{\prime}$, \ie{}, we calculate $h_i(f^{\prime})$.
Then we check whether $h_i(f^{\prime})$ is equal to $j$. 
If so, we consider $\mathcal{B}_i[j]$ as a pure bucket recording flow $f^{\prime}$ with size ${B}_i^{c}[j]$. 
Note that the false positive rate of pure bucket verification, \ie, the probability of misjudging a non-pure bucket as a pure one, is $\frac{1}{m}$, which is calculated as follows.
For any non-pure bucket, we can calculate its flow ID, which should be considered as a random value.
The probability that a random ID is hashed to the same bucket is $\frac{1}{m}$.
We further discuss that such false positives can be automatically eliminated during decoding ($\S$~\ref{sec:dicussion}), and prove they have little impact on decoding success rate through mathematical analysis (Theorem ~\ref{math:sketch1}).

\noindent\textit{$\bullet$ Single flow extraction:}
To extract/delete flow $f^{\prime}$ from $\mathcal{B}_i[j]$ as well as its other mapped buckets, first, we locate its other $(d-1)$ mapped buckets.
Second, for each mapped bucket ${B}_{i^{\prime}}[h_{i^{\prime}}(f^{\prime})]$, we decrement its count field ${B}_{i^{\prime}}^{c}[h_{i^{\prime}}(f^{\prime})]$ by ${B}_i^{c}[j]$, and update its IDsum field to $(({B}_{i^{\prime}}^{ID}[h_{i^{\prime}}(f^{\prime})] - {B}_i^{ID}[j]) \mod p)$ through \textit{modular subtraction}.

\begin{figure}[t!]
    \setlength{\abovecaptionskip}{0.15cm}
    \setlength{\belowcaptionskip}{-0.8cm}
    \centering  
    \includegraphics[width=0.9\linewidth]{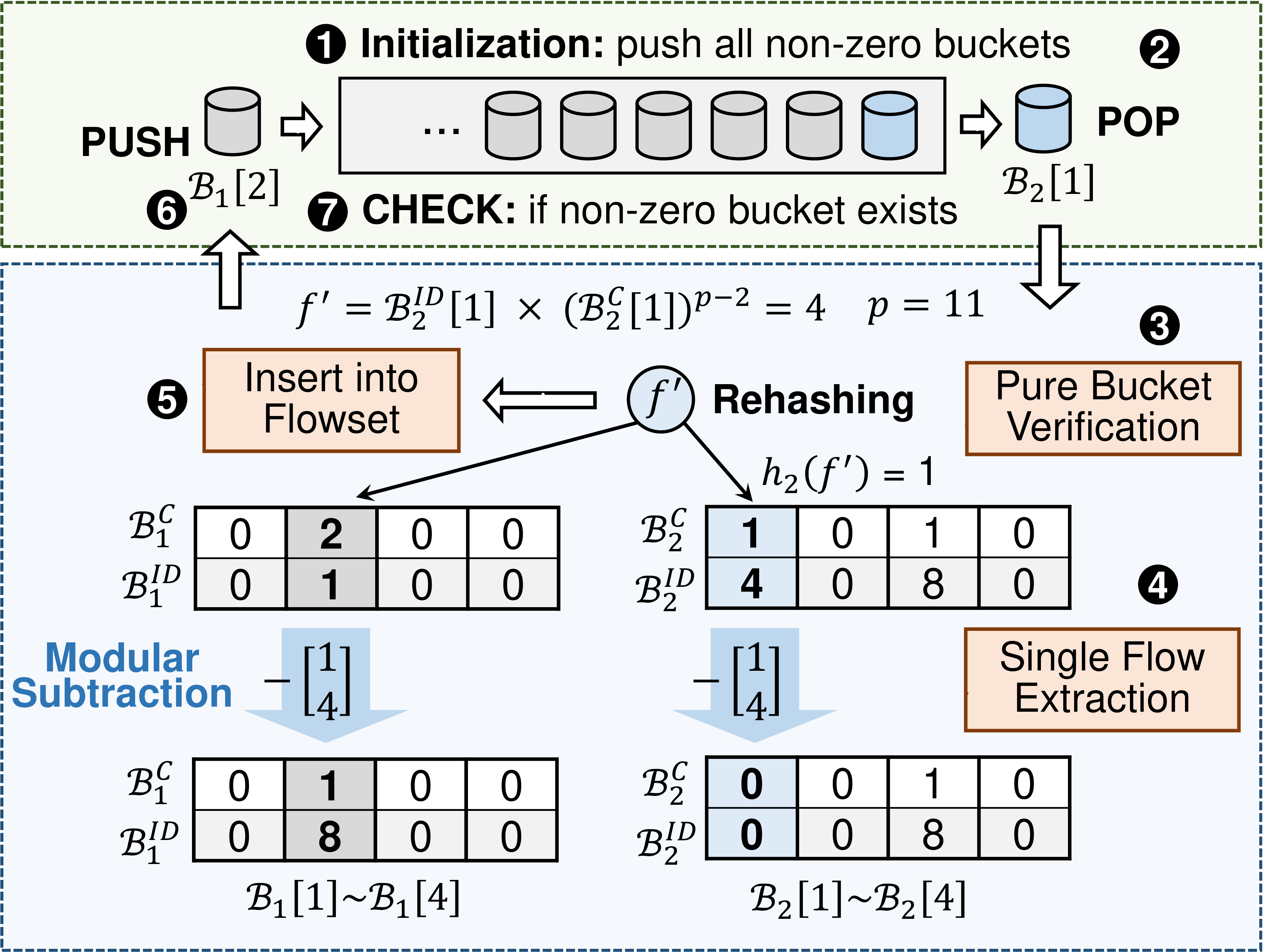}
    \caption{An example of decoding.}
    \label{algopic:decoding}
\end{figure}

\bbb{Addition/Subtraction operations:}
Adding/Subtracting \sketchname{} $FS_1$ to/from \sketchname{} $FS_2$.
$FS_1$ and $FS_2$ must use the same parameters including hash functions, number of arrays, number of buckets, and primes. 
For each bucket of $FS_2$, we update it as follows.
First, we locate the bucket of $FS_1$ that is in the same position as it.
Second, we add/subtract the count field of the located bucket of $FS_1$ to/from its count field.
Third, we modular add/subtract the IDsum field of the located bucket of $FS_1$ to/from its IDsum field.

%

\bbb{Complexities:}
We discuss the complexities of \sketchname{} in detail in Appendix \ref{sec:dicussion}.
Its space complexity is $\Theta(M)$ and time complexity of decoding is $O(md^2)$, where $M$ refers to the number of flows encoded into \sketchname{}.
When $d=3$ and $m=\frac{1.23 M}{d}$, \sketchname{} achieves the highest memory efficiency.
That is to say, on average at least 1.23 buckets can record a flow and its size.
We use Theorem~\ref{math:sketch1} to show that when $m d > c_d M+\epsilon$ and $M$ is not too small, the decoding only fails with an extremely small probability $O(\frac{1}{M^{d-2}})$, where $c_d$ refers to the minimum \textit{average number of buckets} required to record a flow and its size for a $d$-array \sketchname{}.
\begin{theorem}
\label{math:sketch1}
Suppose $m d > c_d M+\epsilon$ and $M= \Omega (d^{4d}log^d(md))$.
the decoding of \sketchname{} fails with probability $O(\frac{1}{M^{d-2}})$, where both $\epsilon$ and $c_d$ are small constants.
$$c_d=\left(sup\left\{\alpha \Big| \alpha \in (0,1), \forall x\in(0,1), 1-e^{-d\alpha x^{d-1} } \right\}\right)^{-1}$$
For example, $c_3=1.23, c_4=1.30, c_5=1.43.$
\end{theorem}

The detailed proof is presented in Appendix \ref{app:math}.

\bbb{Packet loss detection:}
To support packet loss detection, we can deploy a group of \sketchname{}es on edge switches to encode the packets entering the network, and another group of \sketchname{}es to encode the packets exiting the network. 
Thanks to the additivity and subtractivity of \sketchname{}, for each group, we add up the \sketchname{}es in it to obtain a cumulative \sketchname{} encoding all the packets entering or exiting the network.
Then, we subtract the cumulative \sketchname{} encoding all the packets exiting the network from the other one, and the \sketchname{} after subtraction just encodes all the victim flows in the network.
Therefore, this \sketchname{} just requires memory proportional to the number of victim flows for successful decoding.
In other words, \sketchname{} can support packet loss detection with memory overhead proportional to the number of victim flows.

\bbb{[Optional] fingerprint verification:}
To reduce the false positive rate of pure bucket verification, we propose an extra verification method, namely \textit{fingerprint verification}.
Please refer to Appendix \ref{app:fv} for details.

\presec \subsection{Data Plane Components} \postsec
As shown in Figure \ref{fig:overview}, every packet entering the network undergoes the three components of the \systemname{} data plane in sequence: 1) the flow classifier, 2) the upstream flow encoder, and 3) the downstream flow encoder.
%

\subsubsection{Flow Classifier}
\label{sec:fc}
~\\
\bbb{Rationale:}
To detect HHs, HLs, and LLs, \systemname{} deploys the flow classifier in the ingress of each edge switch, so as to classify flows into different hierarchies.
While it is easy to select HHs to monitor according to flow sizes, it is not easy to select HLs to monitor because we can hardly predict how many packets a flow will lose.
Our observation is that for each flow, the number of its lost packets cannot exceed its size.
Therefore, the sizes of HLs should have a minimum value.
\systemname{} selects flows whose sizes exceed this value to monitor, so as to approximate the monitoring of HLs.
In summary, the flow classifier classifies flows purely according to flow sizes.
We choose TowerSketch \cite{yang2021sketchint}, a simple, accurate, and hardware-friendly sketch, as the flow classifier.

\bbb{Data Structure:}
The flow classifier consists of $l$ equal-sized arrays.
The $i^{th}$ array $\mathcal{A}_i$ consists of $w_i$ $\delta_i$-bit counters, where $w_i \times \delta_i$ is a constant and $\delta_{i-1} < \delta_i$.
Also, array $\mathcal{A}_i$ is associated with a pairwise-independent hash function $s_i(\cdot)$.
For each $\delta_i$-bit counter, its maximum value $2^{\delta_i}-1$ is used to represent the state that it is overflowed, and thus be regarded as $+\infty$.


\bbb{Insertion:}
To insert a packet with flow ID $f$, we first calculate the $l$ hash functions to locate $l$ counters: $\mathcal{A}_1[s_1(f)],\mathcal{A}_2[s_2(f)],\cdots$, $\mathcal{A}_l[s_l(f)]$.
We call these counters \textit{the $l$ mapped counters}.
Then, for each of the $l$ mapped counters, we increment it by one unless it is overflowed.
%
%

\bbb{Online query:}
To query the size of flow $f$ online, we simply report the minimum value among the $l$ mapped counters.

\bbb{Packet processing:} 
For a packet with flow ID $f$ entering the network, the flow classifier processes it as follows.
First, we insert it into the flow classifier and query the size of flow $f$.
Then, with the queried flow size, we classify flow $f$ into the corresponding hierarchy according to two thresholds $T_h$ and $T_l$, where $T_h$ is used for selecting HH candidates, and $T_l$ is used for selecting HL candidates.
In general, it satisfies that $T_l <= T_h$.
%
%
If the flow size is larger than or equal to $T_h$, flow $f$ is classified as a HH candidate.
If the flow size is less than $T_l$, flow $f$ is classified as a LL candidate.
If the flow size is between $T_l$ and $T_h$, flow $f$ is classified as a HL candidate.
%
%
%
The LL candidate is further classified into sampled LL candidate or non-sample LL candidate through sampling.
%

 \subsubsection{Upstream Flow Encoder} 
\label{sec:ufe}
~\\
\bbb{Rationale:}
To support packet loss detection, \systemname{} deploys the upstream flow encoder in the ingress of each edge switch just after the flow classifier, so as to encode the packets entering the network.
%
%
Therefore, the upstream flow encoder should contain two \sketchname{}es to encode HL candidates and sampled LL candidates individually.
Here, for better monitoring of the network state, \systemname{} monitors a portion of LLs to maintain an overview of all victim flows. 
Besides, to support heavy-hitter detection, the upstream flow encoder should contain a \sketchname{} to encode HH candidates.
In summary, the upstream flow encoder should consist of three \sketchname{}es.
%
%

\bbb{Data structure:}
The upstream flow encoder is a $d$-array \sketchname{} divided into three $d$-array \sketchname{}es: 1) an upstream HH encoder for encoding HH candidates; 2) an upstream HL encoder for encoding HL candidates; 3) an upstream LL encoder for encoding sampled LL candidates.
We denote the number of buckets per array of the upstream flow encoder, HH encoder, HL encoder, and LL encoder by $m_{uf}$, $m_{hh}$, $m_{hl}$, and $m_{ll}$, respectively.
Obviously, it satisfies that $m_{uf} = m_{hh} + m_{hl} + m_{ll}$.

\bbb{Packet processing:}
For a packet with flow ID $f$ entering the network, the upstream flow encoder processes it by encoding the packet into one of the encoders corresponding to the hierarchy of flow $f$ unless flow $f$ is a non-sampled LL candidate.
Here, the hierarchy of flow $f$ can be directly obtained because the upstream flow encoder and the flow classifier are deployed on the same edge switch.
%
%

\subsubsection{Downstream Flow Encoder} 
\label{sec:dfe}
~\\
\bbb{Rationale:}
To support packet loss detection, \systemname{} deploys the downstream flow encoder in the egress of each edge switch, so as to encode the packets exiting the network.
As the downstream flow encoder is not responsible for heavy-hitter detection, it should consist of two \sketchname{}es to encode HL candidates and sampled LL candidates.
%
%

\bbb{Data structure:}
The downstream flow encoder is a $d$-array \sketchname{} divided into two $d$-array \sketchname{}es: 1) a downstream HL encoder; 2) a downstream LL encoder.
To support packet loss detection, the number of buckets per array of the downstream HL encoder and LL encoder must also be $m_{hl}$ and $m_{ll}$, respectively, so as to support addition and subtraction operations with the corresponding upstream encoder.
We denote the number of buckets per array of the downstream flow encoder by $m_{df}$.
In general, it satisfies that $m_{df} < m_{uf}$, and therefore satisfies that $m_{df} \geqslant m_{hl} + m_{ll}$.
%
%

\bbb{Packet processing:}
For a packet with flow ID $f$ exiting the network, the downstream flow encoder processes it by encoding the packet into one of the encoders corresponding to the hierarchy of flow $f$ unless flow $f$ is a non-sampled LL candidate.
Here, packets of HH candidates are also encoded into the downstream HL encoder.
Different from the upstream flow encoder, the downstream flow encoder cannot directly obtain the flow hierarchy from the flow classifier, as a flow could enter and exit the network at different edge switches.
To address this issue, first, considering that there are four flow hierarchies, we can use $\lceil \log(4) \rceil = 2$ bits in the original packet header to transmit this information.
For example, for IPv4 protocol, we can use the unused bits in the type of service (ToS) field.
If there are not enough unused bits, second, we can transmit the flow hierarchy in an INT-like \cite{kim2015band} manner.

%% file: NSDI2022_9/4_Controlplane_design.tex
\presec \section{\systemname{} Control Plane} \postsec
The \systemname{} control plane consists of a central controller, as well as the control plane of each edge switch.
In this section, we detail the design of the \systemname{} control plane.
We begin by laying out how the \systemname{} control plane collects sketches from the \systemname{} data plane, then introduce how to support seven measurement tasks with the collected sketches, and finally propose how to shift measurement attention as network state changes.

\presec \subsection{Collection from Data Plane} \postsec
The central controller needs to periodically collect sketches, \ie, the flow classifier, the upstream flow encoder, and the downstream flow encoder, from the \systemname{} data plane, so as to support persistent measurement.
However, the collection cannot be completed in an instant, and thus inevitably collide with packet insertion if there is only a group of sketches.
Specifically, if the central controller wants to collect sketches at time $t$, it will inevitably collect some counters inserted by packets after $t$, which could result in decoding failure of \sketchname{}.
%
%
%
To address this issue, \systemname{} takes two steps: 1) \textit{timeline division} and 2) \textit{clock synchronization}.
Next, we just briefly cover the two steps.
We detail the two steps in Appendix~\ref{appendix:collection}, where we further analyze \textit{the appropriate time} for the central controller to collect sketches.

\bbb{Timeline division:}
Each edge switch periodically flips a 1-bit timestamp to divide the timeline into fixed-length time intervals (called epochs) with interleaved 0/1 timestamp, and copies a group of sketches for rotation.
Each group of sketches corresponds to a distinct timestamp value, and monitors the epochs with that timestamp value.

\bbb{Clock synchronization:}
The central controller also maintains a 1-bit periodically flipping timestamp, and periodically synchronizes its clock with the control plane of each edge switch, so as to make opportunities for collection.

Every time the locally maintained 1-bit timestamp flips, an epoch ends, the central controller starts to collect the group of sketches monitoring this epoch, and the other group of sketches starts to monitor the current epoch.

%

\presec \subsection{Measurement Tasks} \postsec
\label{sec:mtasks}
With the collected sketches, the central controller can support packet loss detection and six packet accumulation tasks.

\bbb{Packet loss detection:}
reporting each victim flow and the number of its lost packets.
The central controller can support packet loss detection by analyzing the upstream and downstream flow encoders collected from each edge switch.
First, for each edge switch, we decode the upstream HH encoder to obtain the HH Flowset, and then reinsert each flow with its size in the HH Flowset into the upstream HL encoder.
Second, we add up the upstream/downstream HL/LL encoder of each edge switch through addition operation to obtain the cumulative upstream/downstream HL/LL encoder.
Third, we subtract the cumulative downstream HL/LL encoder from the cumulative upstream HL/LL encoder to obtain the \textit{delta HL/LL encoder}.
Fourth, we decode the delta HL/LL encoder to obtain the HL/LL Flowset.
Finally, we report the flows in the HL Flowset as HLs, and the flows in the LL Flowset but not in the HL Flowset as LLs.
For each of these flows, its estimated number of lost packets is the sum of its size in the HL Flowset and the LL Flowset.

For each edge switch, the central controller can support the following six widely-studied \cite{elastic2018,univmon2016,song2020fcm,yang2021sketchint} packet accumulation tasks by analyzing the flow classifier and upstream HH encoder collected from it.
Then, by synthesizing the results of each edge switch, the central controller can easily support these tasks in a network-wide manner.
We detail these six tasks from the perspective of an edge switch.

\bbb{Heavy-hitter detection:}
reporting flows whose sizes exceed $\Delta_h$. 
First, we decode the upstream HH encoder to obtain the HH Flowset, which records flows with ID $f_i$ and size $q_i$.
For any flow $f_j$ in the HH Flowset, if its estimated flow size $\mathcal{T}_h+ q_j$ is larger than $\Delta_h$, we report it as a HH.
Note that $T_h$ is the threshold used for selecting HH candidates.

\bbb{Flow size estimation:} 
reporting flow size of flow $f_j$.
%
%
Similarly, we obtain the HH Flowset.
If flow $f_j$ is in the HH Flowset, we report its flow size as $\mathcal{T}_h+ q_j$.
Otherwise, we report its flow size as query result from the flow classifier. 

\bbb{Heavy-change detection:}
reporting flows whose sizes change beyond $\Delta_c$ in two adjacent epochs. 
Similarly, we obtain the HH Flowset.
For any flow $f_j$ in the HH Flowset of either epoch, we estimate its flow size in the two epochs.
If the difference between the two estimated flow sizes is larger than $\Delta_c$, we report flow $f_j$ as a heavy-change.

\bbb{Cardinality estimation:}
reporting number of flows. 
We apply linear-counting algorithm \cite{whang1990linear} to the counter array with most counters in the flow classifier for estimation.

\bbb{Flow size distribution estimation:}
reporting distribution of flow sizes.
We apply MRAC algorithm \cite{kumar2004data} to each counter array in the flow classifier.
Array $\mathcal{A}_i$ provides distribution of flow size in range $[2^{\delta_{i - 1}} - 1, 2^{\delta_i} - 1)$. 
The remaining distribution of flow size in range $[2^{\delta_{i}} - 1, +\infty)$ is estimated from the flows larger than $2^{\delta_{i}} - 2$ in the HH Flowset.
%

\bbb{Entropy estimation:}
reporting entropy of flow sizes. 
Based on the estimated flow size distribution, we can easily compute the entropy as follows: $-\sum \left ( n_i \cdot \frac{i}{N}\log{\frac{i}{N}} \right )$, where $n_i$ is the number of flows of size $i$, and $N = \sum ( i \cdot n_i)$.


%


\presec
\subsection{Shifting Measurement Attention}
\postsec

A practical measurement system should pay attention to different kinds of tasks for different network states.
When there are only rare packet losses in network, the system should pay more attention to and allocate more memory for packet accumulation tasks.
In contrast, when there are lots of packet losses in network, the system should pay more attention to and allocate more memory for packet loss detection to help diagnose network faults.

Aiming at this target, \systemname{} decides to shift measurement attention as network changes at run-time without recompilation.
Every time all the sketches monitoring the previous epoch are collected, \systemname{} takes two phases to shift measurement attention.
First, the central controller monitors the real-time network state, including the number and flow size distribution of flows and victim flows, by analyzing the collected sketches.
Second, the central controller reconfigures the \systemname{} data plane according to the real-time network state while \textit{maintaining high memory utilization.}
The central controller not only reallocates memory of the upstream and downstream encoders between their different parts, but also adjusts the thresholds for flow classification and the sample rate for sampling LL candidates.
To avoid interference with the monitoring of the current epoch, the reconfiguration will not function immediately, but in the next epoch.

For \systemname{}, the network state could be clearly classified into two levels: 1) \hstate{} that \systemname{} can allocate sufficient memory to monitor all victim flows; 2) \istate{} that \systemname{} cannot allocate sufficient memory to monitor all victim flows, and thus must select HLs to monitor.
For each level of network state, \systemname{} behaves almost the same in shifting measurement attention, and we detail how it behaves in this section.
%

\presec
\subsubsection{Healthy Network State}

\label{sec:healthy}
~\\
Suppose the previously monitored network state is healthy, and now the central controller starts to shift measurement attention.
Currently, the LL encoders are not allocated any memory as \systemname{} can monitor all victim flows, and $T_l$ must be $1$ as no flows should be classified into LL candidates.
The memory allocation between the upstream HH encoder and the upstream HL encoder is flexible.


\bbb{Monitoring real-time network state:}
The monitoring proceeds as follows.
First, for each edge switch, the central controller estimates the number of flows and flow size distribution\footnote{The estimation of flow size distribution using the MRAC algorithm typically takes several seconds to perform multiple iterations. Therefore, we recommend either 1) monitoring the flow size distribution over time intervals longer than epoch or 2) reducing the number of iterations to support more real-time monitoring.} as described above ($\S$~\ref{sec:mtasks}).
Second, for each edge switch, the central controller obtains the number of HH candidates by decoding the upstream HH encoder.
After all decoding stops, if the decoding of any upstream HH encoder fails, the central controller stops the monitoring as the decoding of the delta HL encoder requires reinserting the decoded HH candidates into the upstream HL encoders.
Third, the central controller obtains the number of HLs (equals to victim flows for healthy network state) by decoding the delta HL encoder as described above ($\S$~\ref{sec:mtasks}).
If the decoding fails, the central controller estimates the number of HLs by applying linear-counting algorithm to any bucket array of the delta HL encoder.

\bbb{Reconfiguring \systemname{} data plane:}
The core idea of reconfiguration is to first ensure the successful decoding of \sketchname{}es for supporting packet loss detection and heavy-hitter detection, while maintaining high memory utilization.
The reconfiguration proceeds as follows.

\bbb{\textit{Step 1:}} We focus on the successful decoding of the upstream HH encoders as they are decoded first.
For each edge switch, if the decoding of the upstream HH encoder fails, the central controller turns up $T_h$ according to the number of flows and flow size distribution, controlling the expected load factor\footnote{Load factor refers to the ratio of the number of recorded flows to the number of buckets of \sketchname{}. \sketchname{} achieves the highest memory efficiency when $d$ is set to $3$, that on average 1.23 buckets can record a flow and its size. Thus, the maximum load factor of \sketchname{} is around $81.3\% = \frac{1}{1.23}$.} of the upstream HH encoder at around 70\%\footnote{Here, we decide not to pursue the maximum load factor for two reasons: 1) the potential increase of HH candidates in the current epoch and 2) the inevitable estimation error in linear-counting.}, so as to maintain high memory utilization.
After turning up $T_h$, the central controller stops the reconfiguration as the decoding of the delta HL encoder cannot proceed.

\bbb{\textit{Step 2:}} We focus on the successful decoding and high memory utilization of the delta HL encoder.
If the decoding of the delta HL encoder fails, according to the estimated number of HLs, the central controller estimates the required memory for 70\% load factor.
If the maximum memory that the HL encoders can be allocated to, \ie, all the memory of the downstream flow encoder, cannot cover the required memory, the healthy network state transitions to the ill network state.
In this case, the central controller 1) reallocates the memory inside the upstream and downstream flow encoders as the fixed allocation described in the ill network state ($\S$~\ref{sec:ill}), 2) sets $T_l$ to $T_h$, and 3) adjusts the sample rate for 70\% load factor of the delta LL encoder assuming that each HL will be a LL.
Otherwise, the central controller just expands the HL encoders to the required memory.
If the decoding of the delta HL encoder succeeds and its load factor is lower than 60\%, the central controller tries to compress the HL encoders to approach 70\% load factor for high memory utilization.
Here, we \textit{reserve the minimum memory} for the HL encoders to handle the potential small burst of victim flows.

\bbb{\textit{Step 3:}} After all the memory reallocation, we focus on the successful decoding and high memory utilization of the upstream HH encoders.
For each edge switch, with the number of HH candidates and the memory of the upstream HH encoder, the central controller further estimates the expected load factor of the upstream HH encoder.
if the expected load factor of the upstream HH encoder is lower than 60\% or larger than 70\%, the central controller turns down or up $T_h$ to approach 70\% load factor.

\presec
\subsubsection{Ill Network State}

\label{sec:ill}
~\\
Suppose the previously monitored network state is ill, and now the central controller starts to shift measurement attention.
Currently, all the HH, HL and LL encoders are allocated fixed memory, and $T_l$ must be larger than $1$ to select HL candidates.
Specifically, the upstream HH encoder is compressed to the minimum memory, which is the memory difference between the upstream flow encoder and the downstream flow encoder.

\bbb{Monitoring real-time network state:}
The monitoring proceeds in a similar way to that of the healthy network state.
In addition, the central controller obtains the number of LLs by decoding the delta LL encoder as described above ($\S$~\ref{sec:mtasks}).
If the decoding fails, the central controller estimates the number of LLs by applying linear-counting algorithm to the delta LL encoder, and then stops the monitoring.
If both decoding of the delta HL and LL encoders succeeds, the central controller estimates the number and flow size distribution of victim flows as follows.
First, the central controller samples the HLs with the same sampling method and rate as LLs.
Second, the central controller merges sampled HLs and sampled LLs to obtain sampled victim flows.
Third, the central controller estimates the flow size distribution of victim flows through querying the flow size of each sampled victim flow, and the number of victim flows through dividing the number of sampled victim flows by sample rate.
If the decoding of the delta HL encoder fails, the central controller regards the estimated flow size distribution of sampled LLs, which is also estimated by querying flow sizes, as the flow size distribution of victim flows.

\bbb{Reconfiguring \systemname{} data plane:}
The core idea of reconfiguration is the same as that of the healthy network state.
The reconfiguration proceeds as follows.

\bbb{\textit{Step 1:}} We focus on the successful decoding of the upstream HH encoders, and the reconfiguration proceeds the same as the first step of the healthy network state.
In addition, we focus on the successful decoding of the delta LL encoder.
If the decoding of the delta LL encoder fails, according to the estimated number of LLs, the central controller adjusts the sample rate to make the delta LL encoder approach 70\% load factor, and then stops the reconfiguration.

\bbb{\textit{Step 2:}} We focus on the successful decoding of the delta HL encoder.
If the decoding of the delta HL encoder fails, according to the estimated flow size distribution of victim flows, assuming that each victim flow larger than $T_l$ will be a HL, the central controller turns up $T_l$ to make the delta HL encoder approach 70\% load factor.

\bbb{\textit{Step 3:}} we focus on the high memory utilization of the HL and LL encoders.
If both the decoding of the delta HL and LL encoders succeeds, according to the estimated number of victim flows, the central controller estimates the required memory for monitoring all the victim flows with 70\% load factor.
If the downstream flow encoder can cover the required memory, the ill network state transitions to the healthy network state.
In this case, the central controller 1) eliminates the LL encoders, 2) allocates the required memory (at least the reserved minimum memory) to the HL encoders, and 3) sets $T_l$ to 1.
If the downstream flow encoder cannot cover the required memory, and the load factor of the delta HL encoder or the delta LL encoder is lower than 60\%, the central controller turns up $T_l$ or the sample rate according to the estimated flow size distribution of victim flows or the estimated number of LLs, respectively, so as to approach 70\% load factor.

\bbb{\textit{Step 4:}}
After all the memory reallocation, we focus on the successful decoding and high memory utilization of the upstream HH encoders, and the reconfiguration proceeds the same as the third step of the healthy network state.

%% file: NSDI2022_9/5_Experiments.tex
\presec
\section{Evaluation}
\postsec
We conduct various experiments on CPU platform and our testbed, and focus on the following five key questions. 
%

    \bbb{How much memory/time can \systemname{} save in packet loss detection? (Figure~\ref{fig:sketch:setting1}-\ref{fig:sketch:setting3})}
    We implement \sketchname{} and its competitors in C++, and use CAIDA dataset \cite{caida} to evaluate their memory and time overhead for packet loss detection on CPU platform.
    Results show that \sketchname{} can save memory in all cases and time in most cases.

    \bbb{How accurately can \systemname{} support six packet accumulation tasks? (Figure~\ref{fig:accuracy} in Appendix \ref{eval:2})}
    We implement the combination of TowerSketch and \sketchname{} and its competitors in C++, and use CAIDA dataset to evaluate their accuracy for these six tasks on CPU platform.
    Results show that the combination can achieve at least comparable accuracy in all six tasks.
    
    \bbb{Can \systemname{} automatically shift measurement attention? (Figure~\ref{fig:testbed:flownum}-\ref{fig:testbed:lossrate})}
    We generate workloads according to widely used traffic distributions (\eg, DCTCP \cite{dctcp}) for evaluation.
    We use the above workloads to evaluate \systemname{} by generating different network states on our testbed.
    Results show that \systemname{} can always automatically shift measurement attention as network state changes at run-time, and maintains high memory utilization in most cases.
    
    \bbb{How fast can \systemname{} shift measurement attention? (Figure~\ref{fig:testbed:shift-process})}
    We use the above workloads to evaluate \systemname{} over a large time window, in which the network state changes 8 times.
    Results show that \systemname{} can shift measurement attention within at most 3 epochs. 
    
    \bbb{How fast can \systemname{} monitor the network? (Figure~\ref{fig:testbed:response-time}-\ref{fig:testbed:cdf} in Appendix \ref{app:timeoverhead})}
    We use the above workloads to evaluate various factors that can affect the epoch length.
    Results show that \systemname{} can monitor the network every 50ms on our testbed, using only one CPU core and consuming only 320Mbps bandwidth. 
    We believe \systemname{} can easily scale to monitor a much larger network in a faster manner.


%


\presec
\subsection{Evaluation on Packet Loss Detection}
\label{eval:1}
\postsec
\bbb{Dataset:}
We use the anonymized IP traces collected in 2018 from CAIDA \cite{caida} as dataset, and use the 32-bit source IP address as the flow ID.
We use the first 100K flows containing 5.3M packets for evaluation.





\bbb{Setup:}
We set up a simulation with a simple topology consisting of only a link on CPU platform.
We compare \sketchname{} with FlowRadar \cite{flowradar2016} and LossRadar \cite{lossradar2016}.
For \sketchname{}, we set its count field and ID field to 32bits, and the number of hash functions to 3.
For FlowRadar, we allocates 10\% memory to the flow filter and 90\% memory to the counting table.
For the flow filter, which is actually a Bloom filter \cite{bloom1970space}, we sets its number of hash functions to $10$.
For the counting table, we set its FlowXOR field, FlowCount field, and PacketCount field to 32bits, and its number of hash functions to $3$.
For LossRadar, we set its count field to 32bits, xorSum field to 48bits, and number of hash functions to $3$.
Here, the xorSum field of LossRadar encodes a 32-bit flow ID as well as a 16-bit packet-specific information that represents the order of a packet in a flow.
For each solution, we deploy it upstream and downstream of the link to encode the packets entering and exiting the link.
%

%
%
%
%

\bbb{Memory/Time overhead\footnote{The memory overhead refers to the minimum memory required to achieve 99.9\% decoding success rate, and the time overhead refers to the corresponding decoding time with the minimum memory.
} \textit{vs.} number of victim flows (Figure \ref{fig:sketch:setting1}):}
Experimental results show that the memory/time overhead of \sketchname{} is proportional to the number of victim flows.
We let the largest 10K flows pass through the link, among which a part of flows are victim flows.
The packet loss rate of victim flows is set to 1\%.
%
%
As the number of victim flows increases, the memory/time overhead of FlowRadar remains unchanged, while that of \sketchname{} increases almost linearly.
We find when the number of victim flows exceeds 6000, the decoding time of \sketchname{} exceeds that of FlowRadar.
This is because the decoding operation of \sketchname{} is more complex than FlowRadar.
Compared to FlowRadar/LossRadar, \sketchname{} saves up to $15.9$/$23.2$ times memory and up to $3.0$/$4.6$ times decoding time.  

\bbb{Memory/Time overhead \textit{vs.} packet loss rate (Figure \ref{fig:sketch:setting2}):}
Experimental results show that the memory/time overhead of \sketchname{} is independent of the number of lost packets.
We let the largest 10K flows pass through the link, among which the largest $100$ flows are victim flows.
As the packet loss rate of victim flows increases, the memory/time overhead of \sketchname{} and FlowRadar remains unchanged, while that of LossRadar increases linearly.
Compared to FlowRadar/LossRadar, \sketchname{} saves up to $276.1$/$6411.2$ times memory and up to $64.5$/$1585.6$ times decoding time.

\bbb{Memory/Time overhead \textit{vs.} number of flows (Figure \ref{fig:sketch:setting3}):}
Experimental results show that the memory/time overhead of \sketchname{} is independent of the number of flows.
We let a certain number of flows pass through the link, among which the largest $100$ flows are victim flows.
The packet loss rate of victim flows is set to 1\%.
As the number of flows increases, the memory/time overhead of \sketchname{} and LossRadar remains unchanged, while that of FlowRadar increases linearly.
Compared to FlowRadar/LossRadar, \sketchname{} saves up to $1535.0$/$128.8$ times memory and up to $821.3$/$23.7$ times decoding time. 

\presec
\subsection{Evaluation on Testbed}
\postsec
\label{eval:testbed}
\bbb{Testbed setup:}
We have fully implemented a \systemname{} prototype on a testbed with a Fat-tree topology composed of 10 Tofino switches and 8 servers, with 1400 lines of P4 \cite{bosshart2014p4} code and 2400 lines of C/C++ code.
Each server has 48 2.1GHz CPU cores, 256 GB RAM, and a 40Gb Mellanox Connectx-3 Pro NIC.
Switches and servers are interconnected with 40Gb links.
We deploy the \systemname{} data plane on all four ToR/edge switches.
An additional server linked with a certain edge switch works as the central controller.
For implementation details of the \systemname{} data plane and control plane, please refer to Appendix \ref{sec:protoimplementaion}.

\begin{figure}[t!]
\setlength{\subfigcapskip}{-0.5cm}
\setlength{\abovecaptionskip}{-0.1cm}
\setlength{\belowcaptionskip}{-0.7cm}
\centering
    \subfigure[Memory overhead.]{
        \begin{minipage}[b]{0.21\textwidth}
            \includegraphics[width=\textwidth]{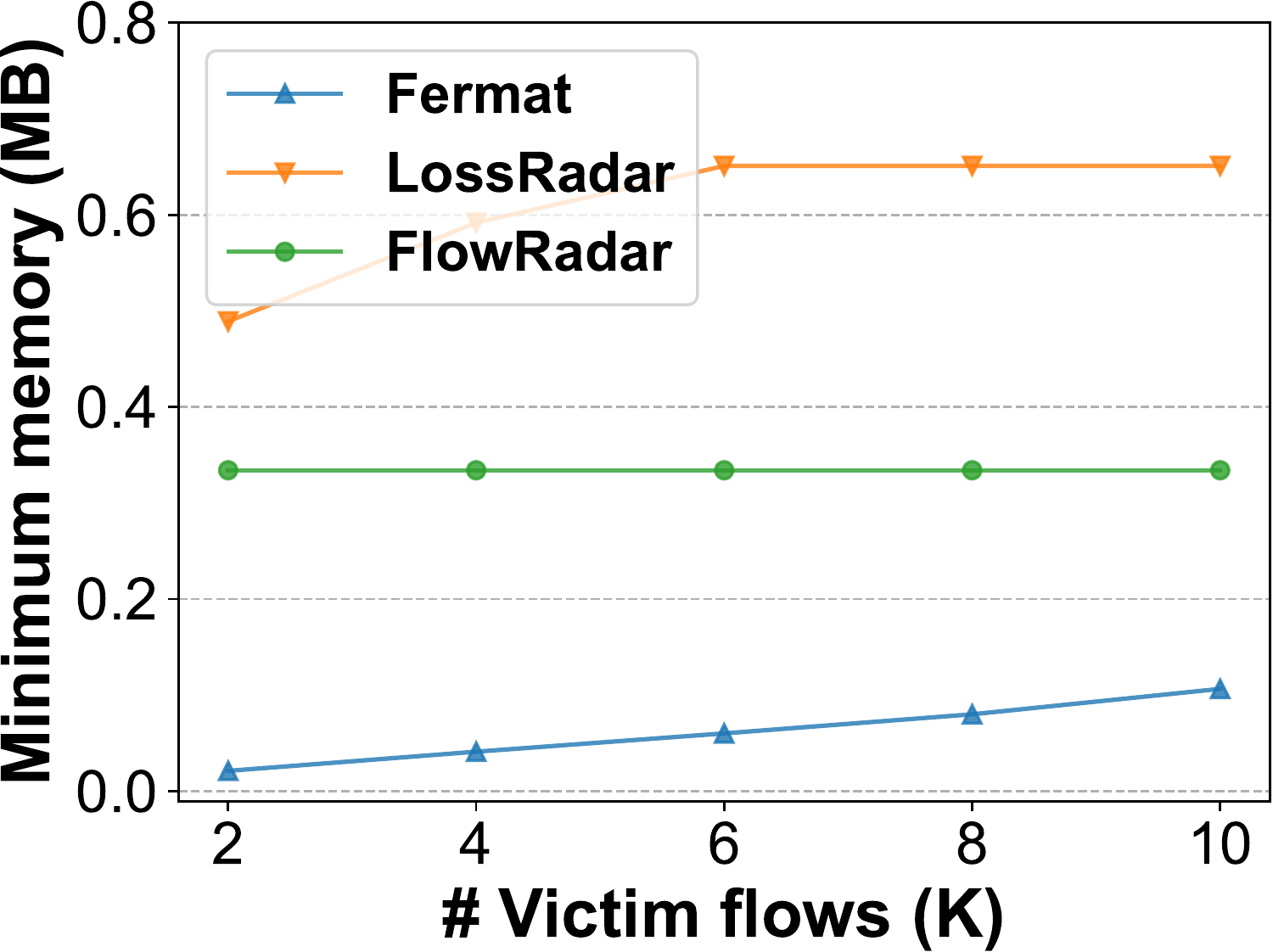}
            \label{fig:sketch:setting3:mem}
        \end{minipage}
    }
    \subfigure[Time overhead.]{
        \begin{minipage}[b]{0.21\textwidth}
            \includegraphics[width=\textwidth]{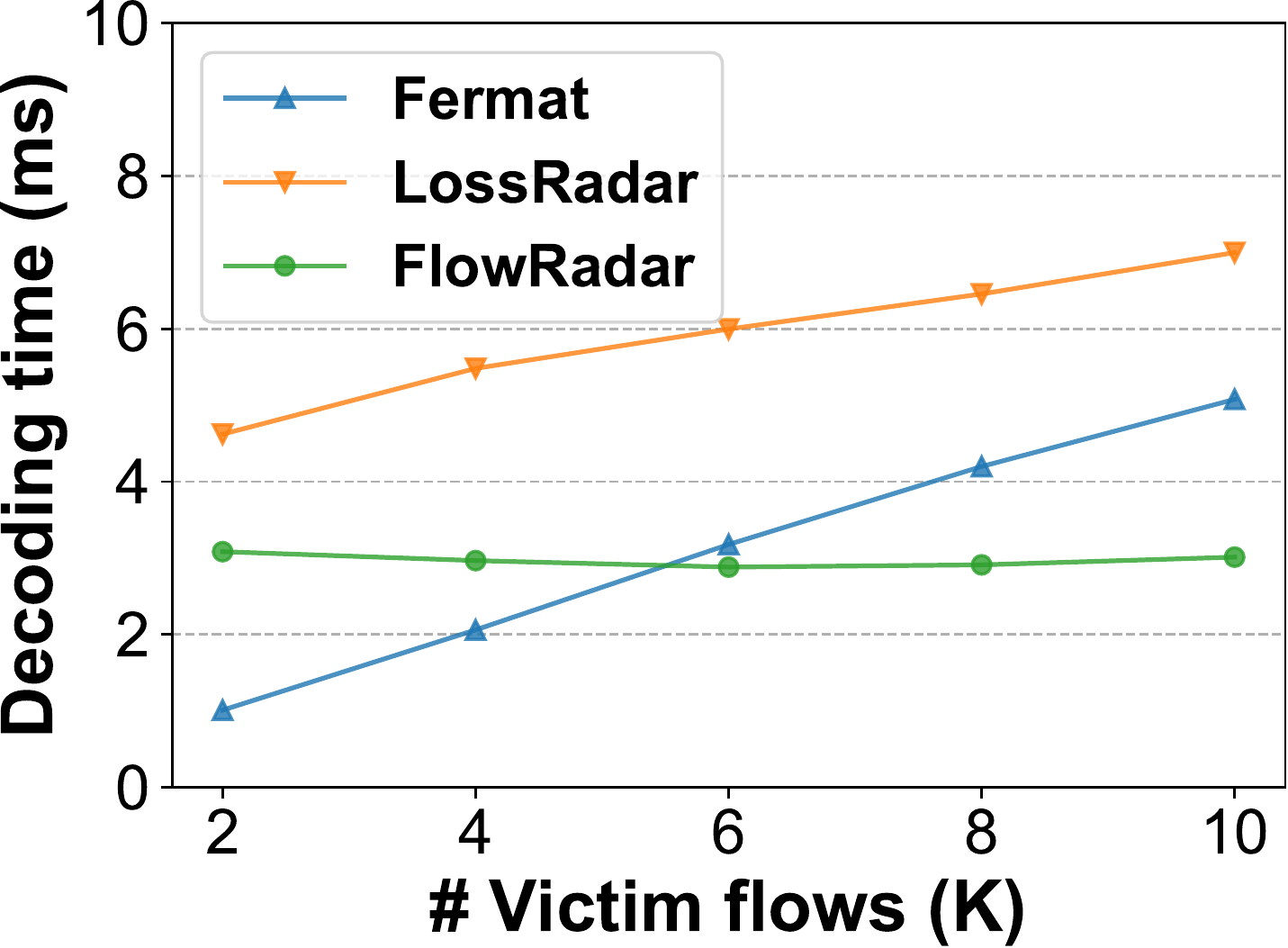}
            \label{fig:sketch:setting3:time}
        \end{minipage}
    }
\caption{Memory/Time overhead \textit{vs.} \# victim flows.}
\label{fig:sketch:setting1}
\end{figure}

\begin{figure}[t!]
\setlength{\subfigcapskip}{-0.5cm}
\setlength{\abovecaptionskip}{-0.1cm}
\setlength{\belowcaptionskip}{-0.7cm}
\centering
    \subfigure[Memory overhead.]{
        \begin{minipage}[b]{0.21\textwidth}
            \includegraphics[width=\textwidth]{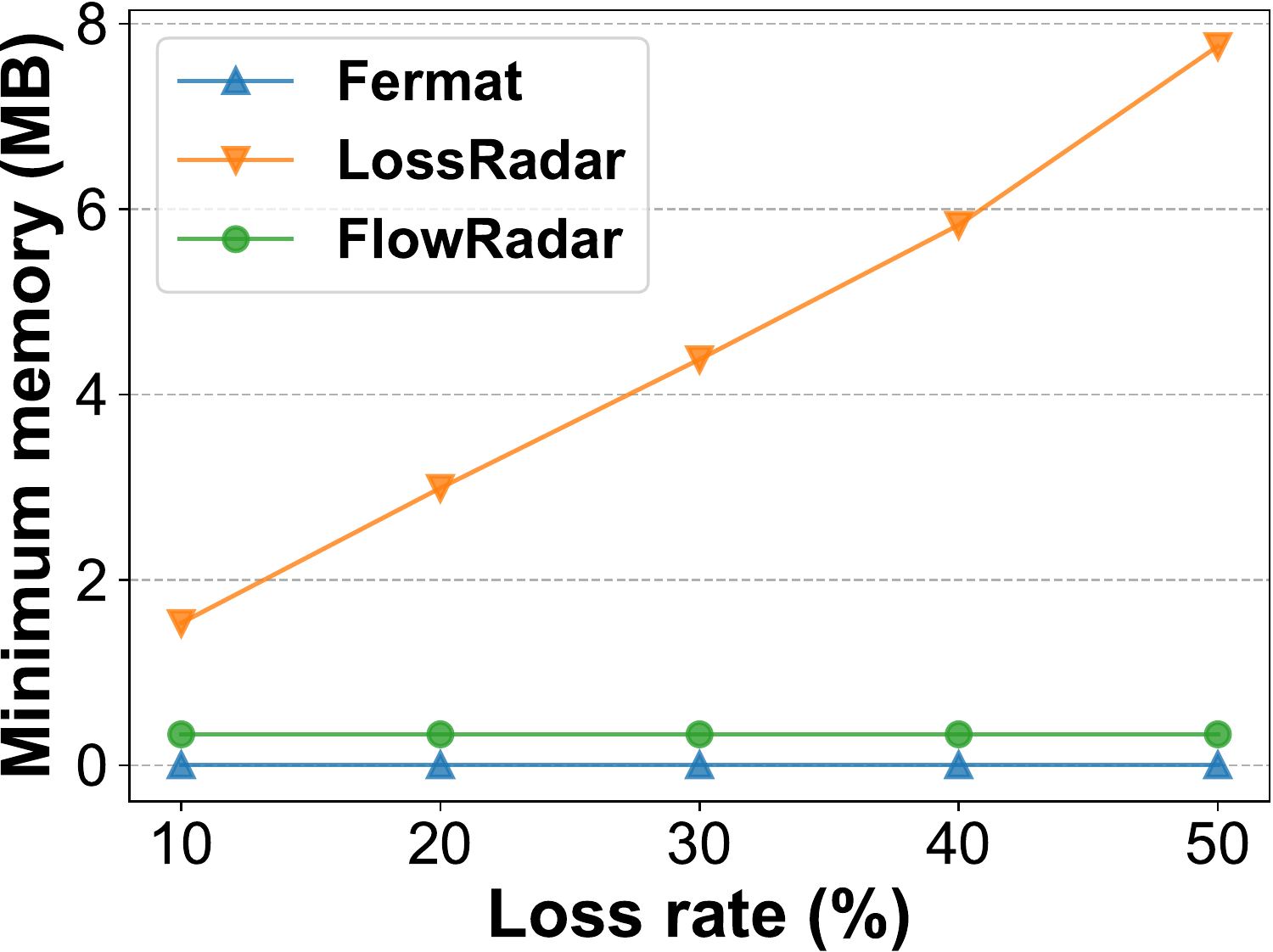}
            \label{fig:sketch:setting1:mem}
        \end{minipage}
    }
    \subfigure[Time overhead.]{
        \begin{minipage}[b]{0.21\textwidth}
            \includegraphics[width=\textwidth]{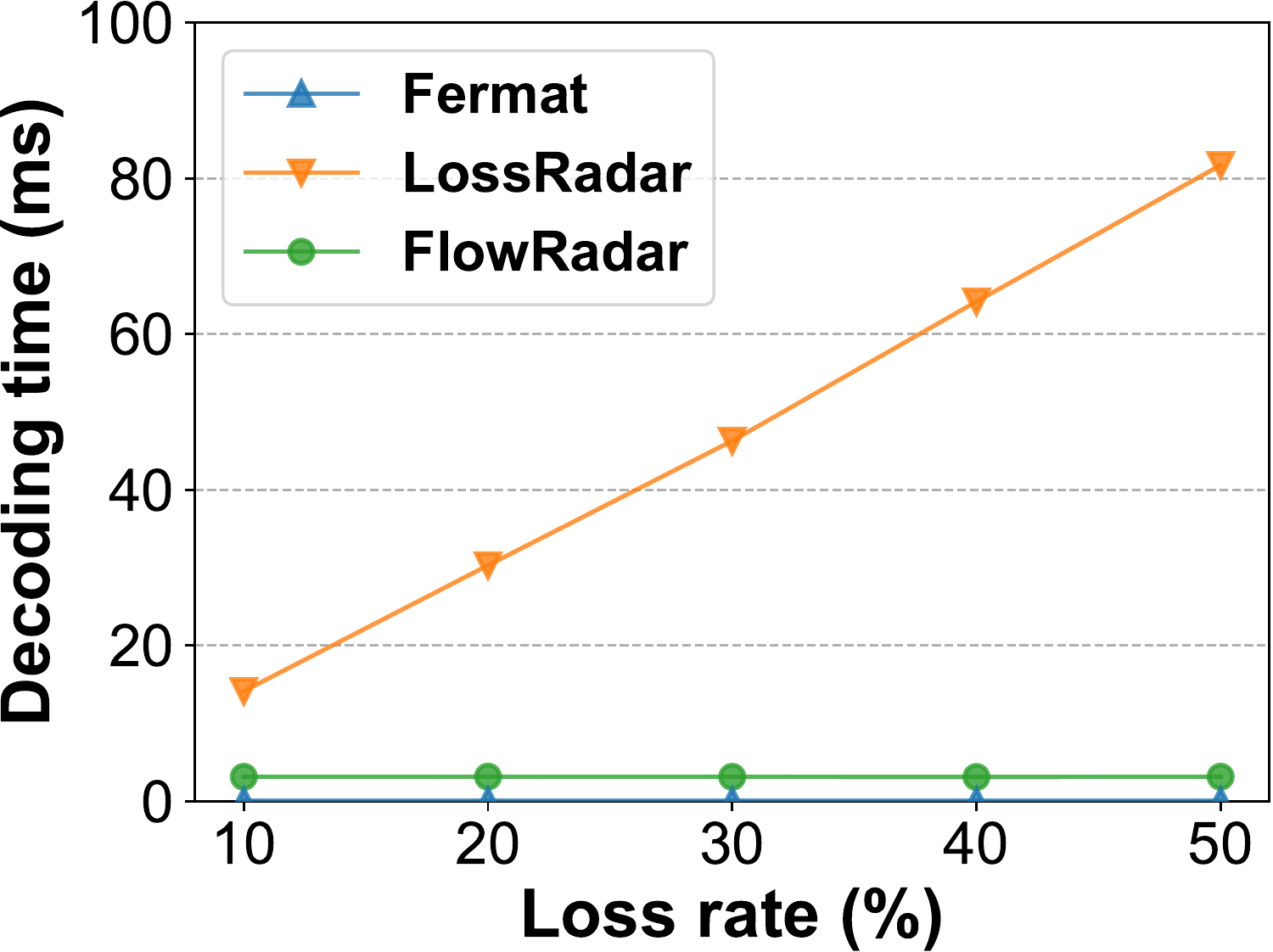}
            \label{fig:sketch:setting1:time}
        \end{minipage}
    }
\caption{Memory/Time overhead \textit{vs.} packet loss rate.}
\label{fig:sketch:setting2}
\end{figure}

\bbb{Workloads:}
We generate workloads consisting of UDP flows according to four widely used distribution: DCTCP \cite{dctcp}, HADOOP \cite{roy2015inside}, VL2 \cite{greenberg2009vl2} and CACHE \cite{atikoglu2012workload}.
We regard the distribution as known input to \systemname{}.
%
\textit{We use the 104-bit 5-tuple as the flow ID.} 
For each flow, We choose its source and destination IP address uniformly, and therefore each server sends and receives almost the same number of flows.
The packet sender and receiver are integrated into a program written in DPDK \cite{dpdk}. 
To manually control packet losses, we let switches proactively drop packets whose ECN fields are set to 1.
In this way, we can flexibly specify any flow as a victim flow and control its packet loss rate.
%
%
To avoid packet losses due to congestion, we set the size of every packet to 64 bytes regardless of its original size, so as to significantly reduce the traffic load in the network and eliminate congestion.
Such operation does not change the number of packets of each flow, and thus has no impact on the behavior of \systemname{}.

\bbb{Parameter settings:}
We set the epoch length to 50ms by default\footnote{For some workloads that cannot run out in 50ms, we extend the epoch length appropriately.}.
For the flow classifier, we set it to consist of an 8-bit counter array and a 16-bit counter array.
We set the number of 8-bit counters $w_1$ to $32768$ and the number of 16-bit counters $w_2$ to $16384$.
For the upstream flow encoder and downstream flow encoder, we set them to consist of 3 bucket arrays for the highest memory efficiency.
We set the number of buckets per array of the upstream flow encoder $m_{uf}$ to 4096, and that of the downstream flow encoder $m_{df}$ to 3072.
For the healthy network state, we fix the minimum memory reserved for HL encoders to a 3-array \sketchname{} with 512 buckets per array.
For the ill network state, we fix the upstream HH, HL, LL encoders to a 3-array \sketchname{} with 1024, 2560, and 512 buckets per array, respectively.
Please refer to Table \ref{tab:overhead} in Appendix \ref{sec:dataimpl} for resource usage on Tofino switches.

\begin{figure}[t!]
\setlength{\subfigcapskip}{-0.5cm}
\setlength{\abovecaptionskip}{-0.1cm}
\setlength{\belowcaptionskip}{-0.7cm}
\centering
    \subfigure[Memory overhead.]{
        \begin{minipage}[b]{0.21\textwidth}
            \includegraphics[width=\textwidth]{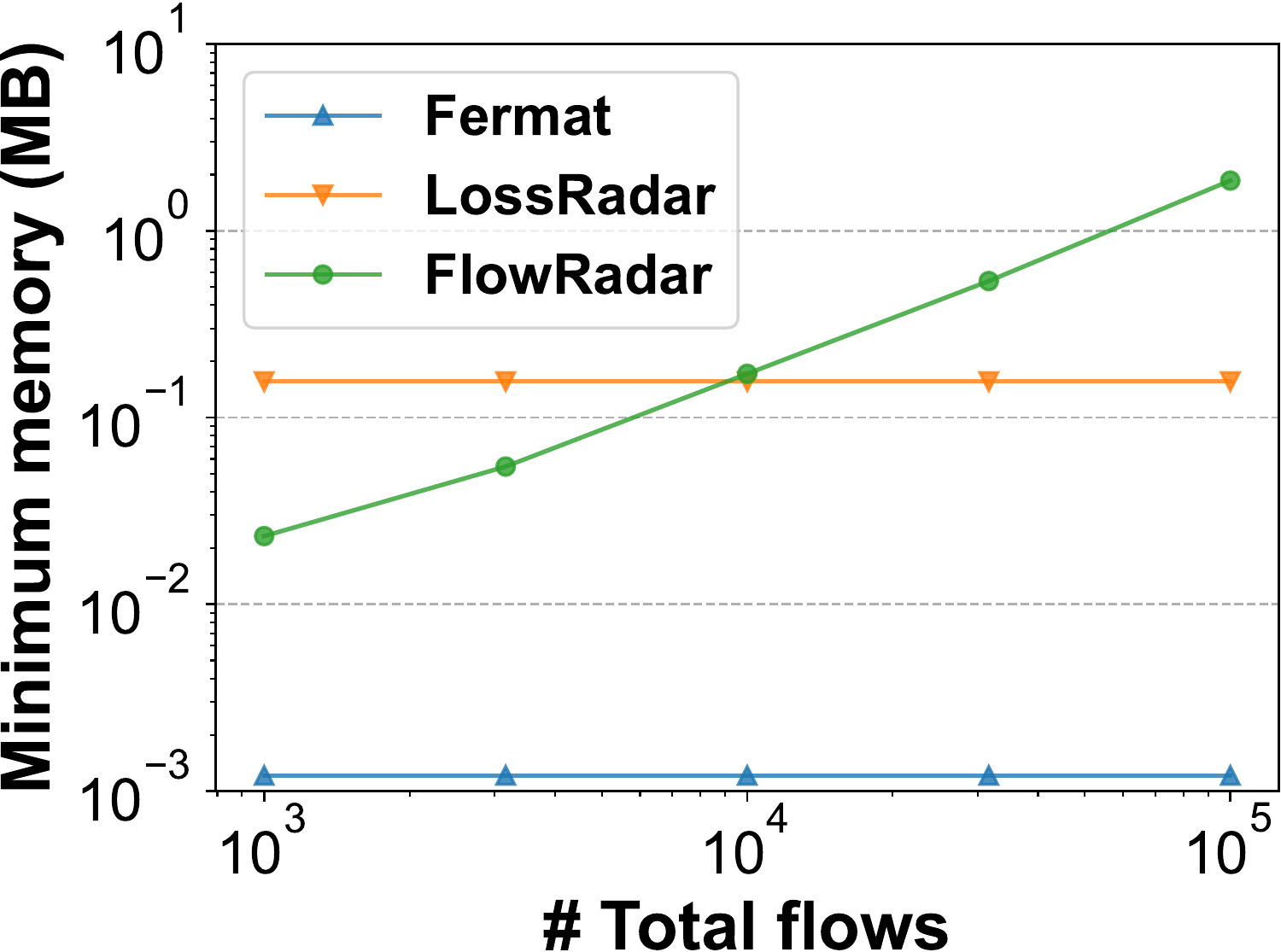}
            \label{fig:sketch:setting2:mem}
        \end{minipage}
    }
    \subfigure[Time overhead.]{
        \begin{minipage}[b]{0.21\textwidth}
            \includegraphics[width=\textwidth]{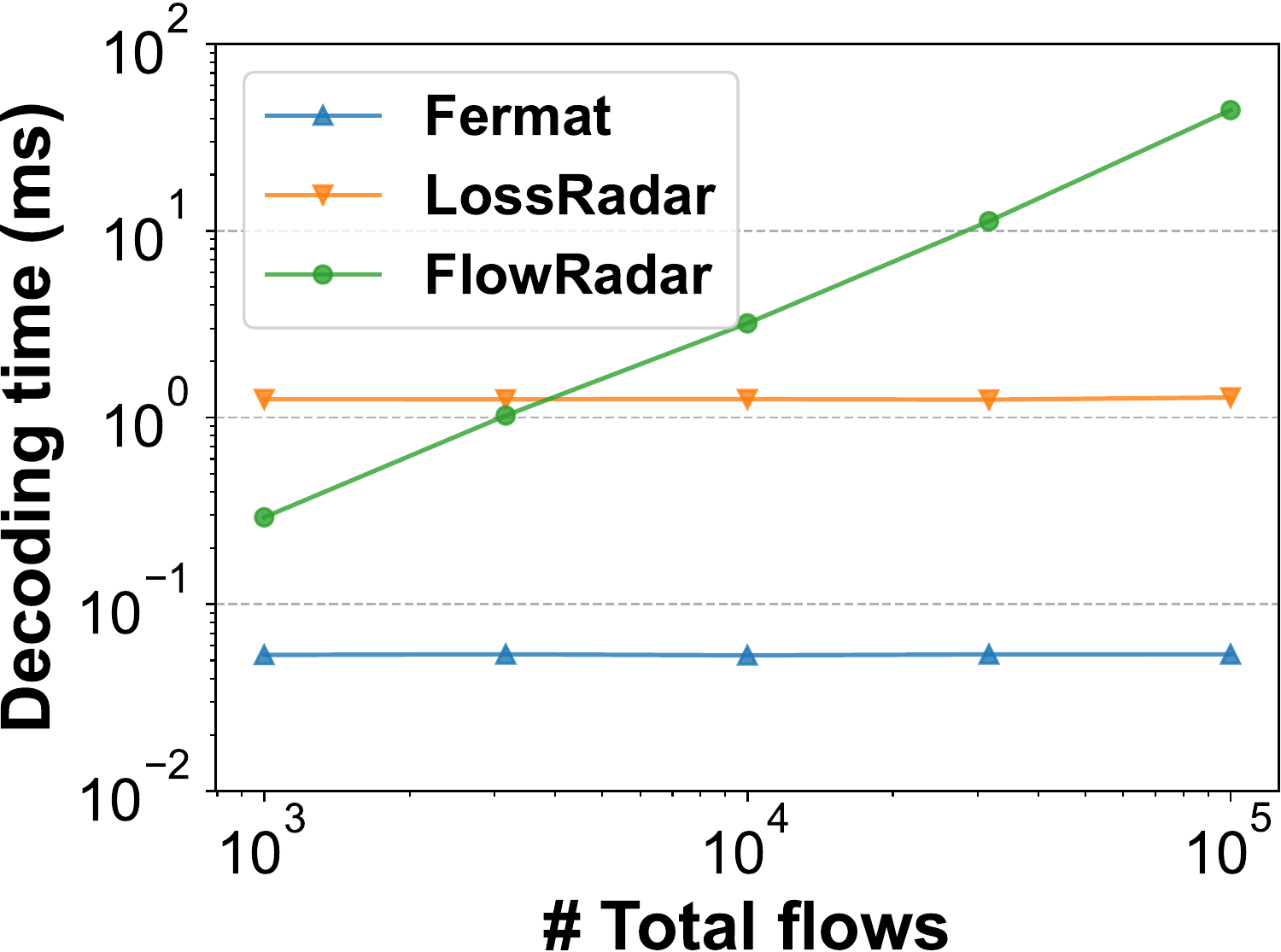}
            \label{fig:sketch:setting2:time}
        \end{minipage}
    }
\caption{Memory/Time overhead \textit{vs.} \# flows.}
\label{fig:sketch:setting3}
\end{figure}

\begin{figure*}[ht!]
\setlength{\subfigcapskip}{-0.1cm}
\setlength{\abovecaptionskip}{-0.1cm}
\setlength{\belowcaptionskip}{-0.4cm}
    \centering
    \subfigure[Memory division.]{
    \includegraphics[width=0.22\textwidth]{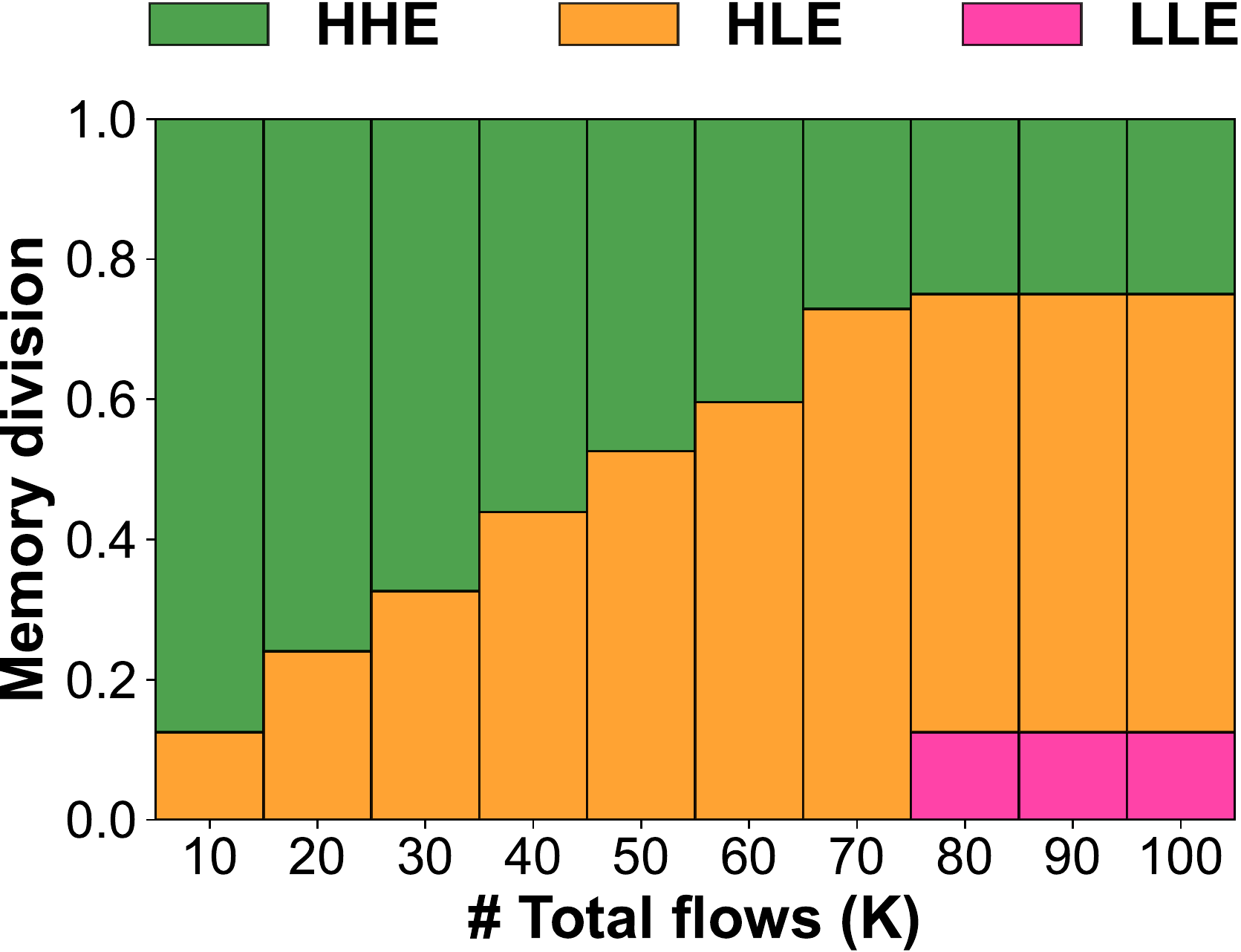}
    \label{fig:testbed:flownum:mem}
    }
    \subfigure[Number of decoded flows.]{
    \includegraphics[width=0.21\textwidth]{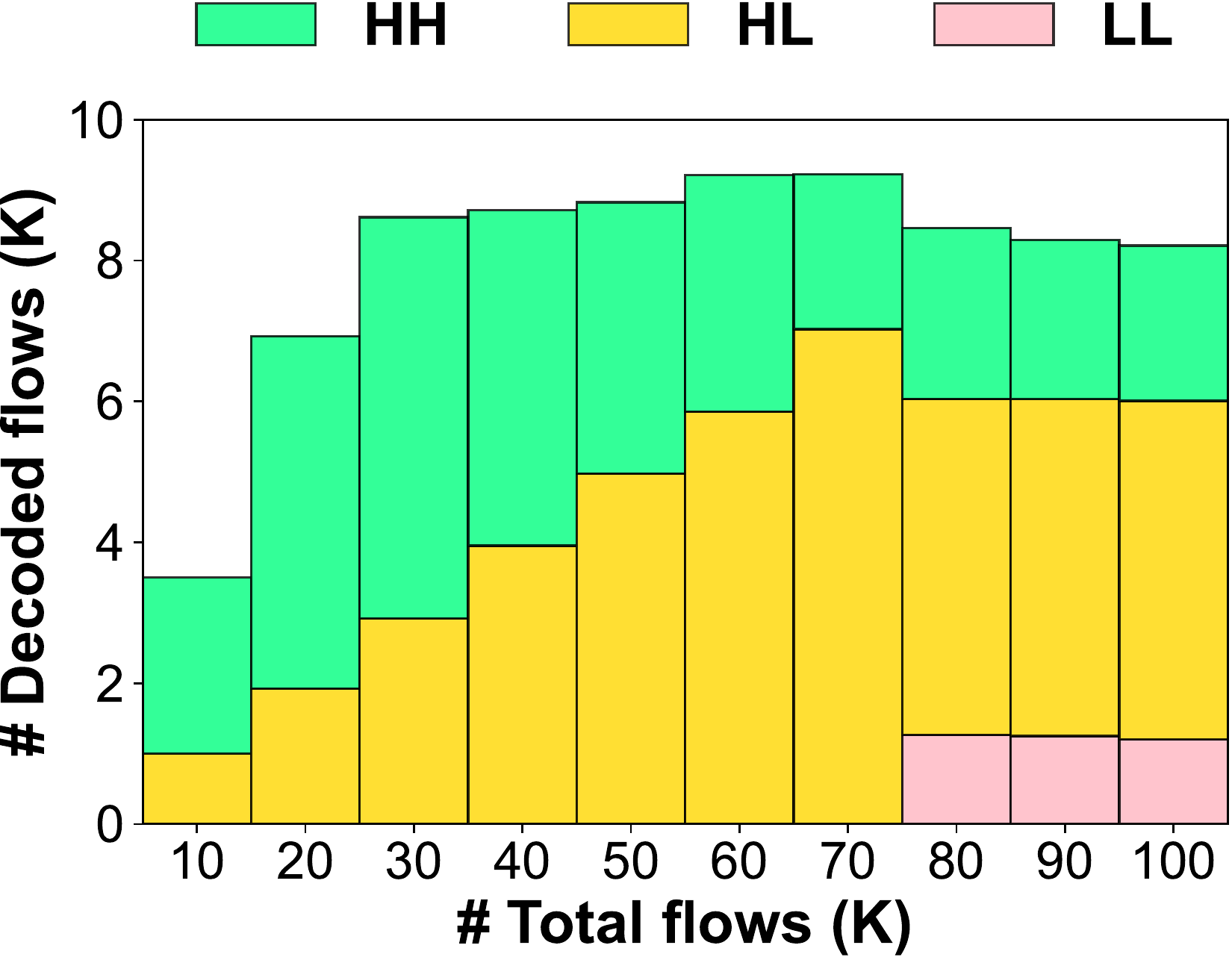}
    \label{fig:testbed:flownum:num}
    }
    \subfigure[Threshold.]{ \includegraphics[width=0.22\textwidth]{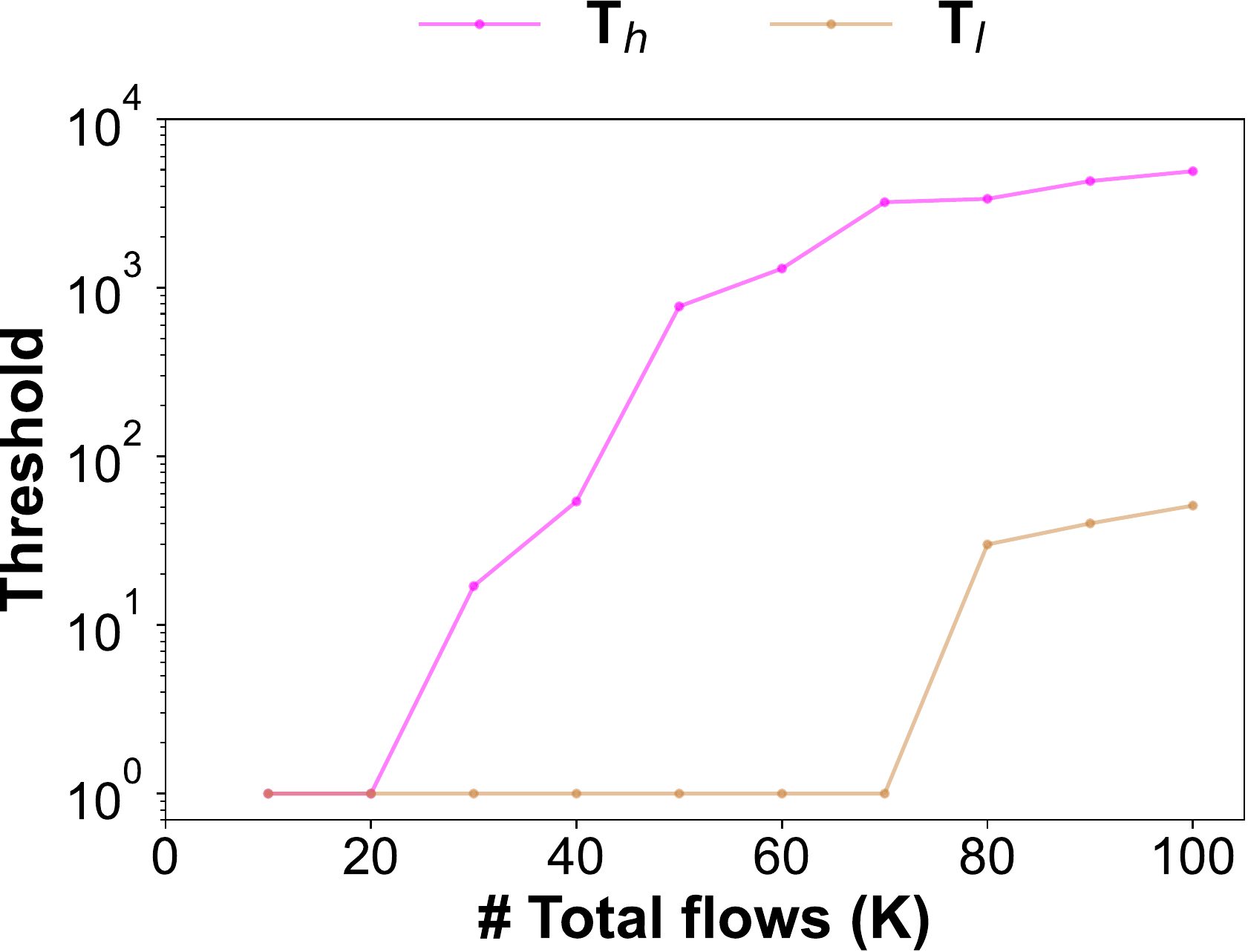}
    \label{fig:testbed:flownum:thresh}
    }
    \subfigure[Sample rate.]{ \includegraphics[width=0.22\textwidth]{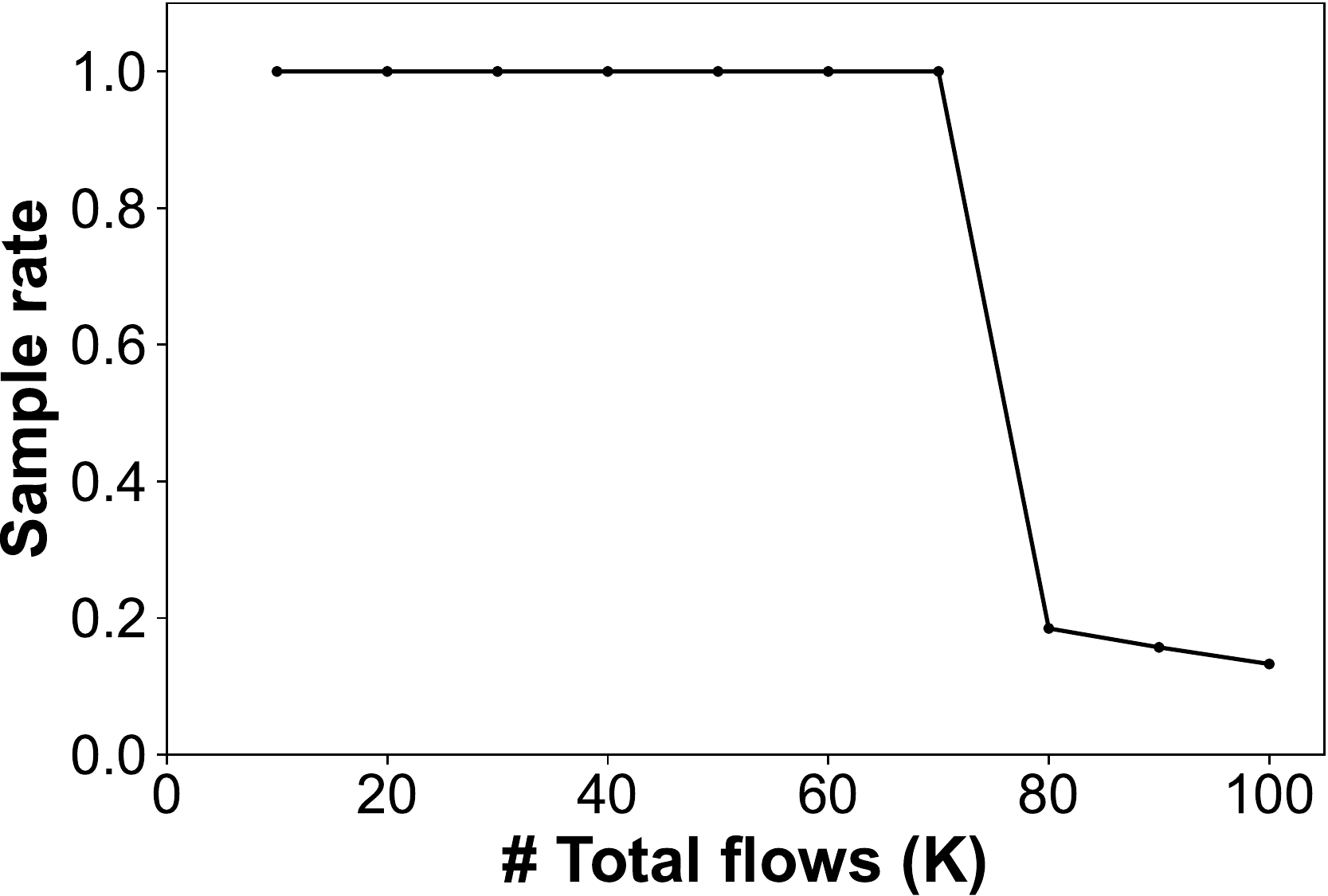}
    \label{fig:testbed:flownum:sample}
    }
    \caption{Measurement attention \textit{vs.} number of flows. \footnotesize{Figure \ref{fig:testbed:flownum:mem} depicts the memory division of HH encoder (HHE), HL encoder (HLE), and LL encoder (LLE) inside the upstream flow encoder. Figure \ref{fig:testbed:flownum:num} depicts the number of HH candidates of an edge switch, the number of HLs in the network, and the number of sampled LLs in the network.}}
\label{fig:testbed:flownum}
\end{figure*}

\begin{figure*}[ht!]
\setlength{\subfigcapskip}{-0.1cm}
\setlength{\abovecaptionskip}{-0.1cm}
\setlength{\belowcaptionskip}{-0.4cm}
    \centering
    \subfigure[Memory division.]{
    \includegraphics[width=0.22\textwidth]{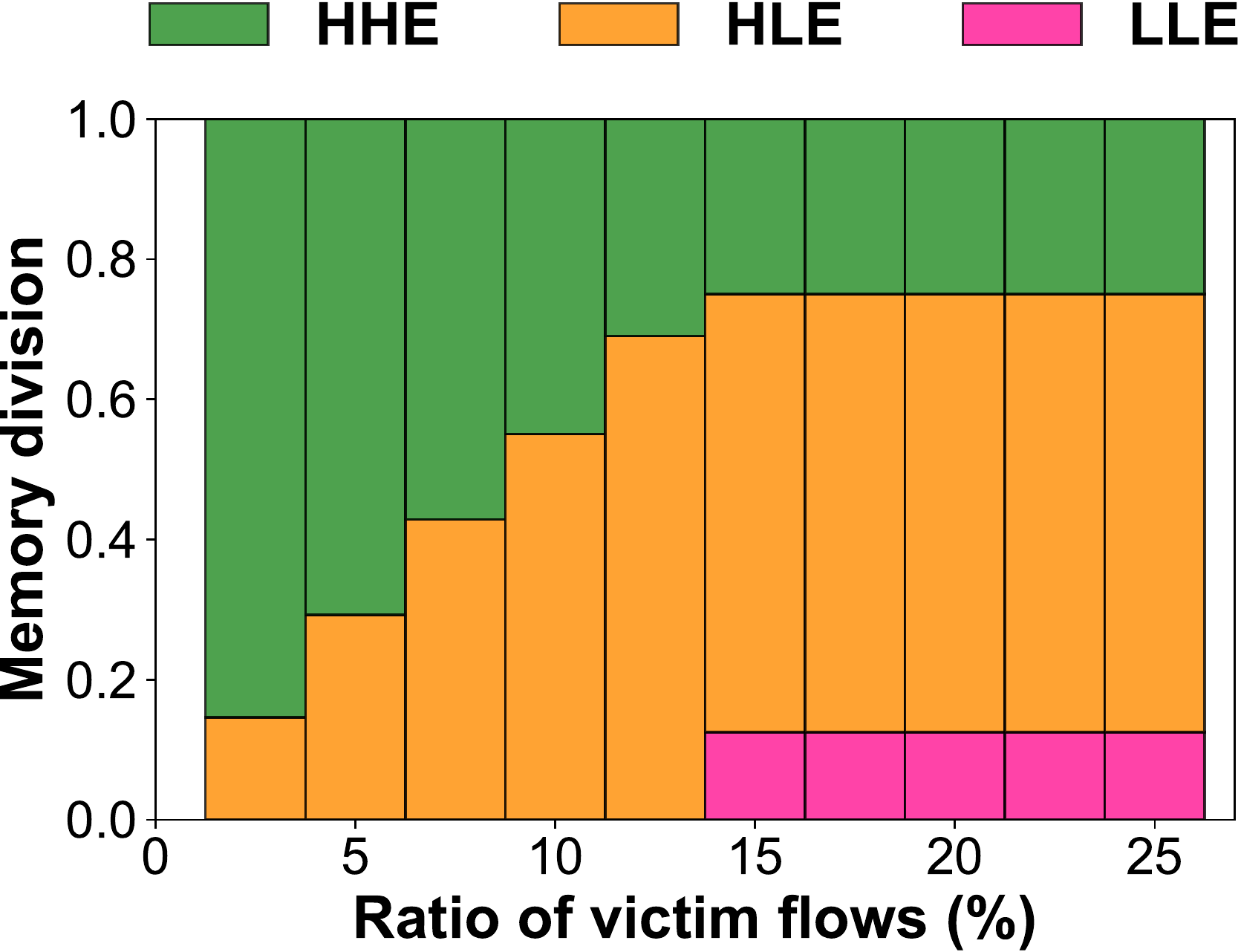}
    \label{fig:testbed:lossrate:mem}
    }
    \subfigure[Number of decoded flows.]{
    \includegraphics[width=0.21\textwidth]{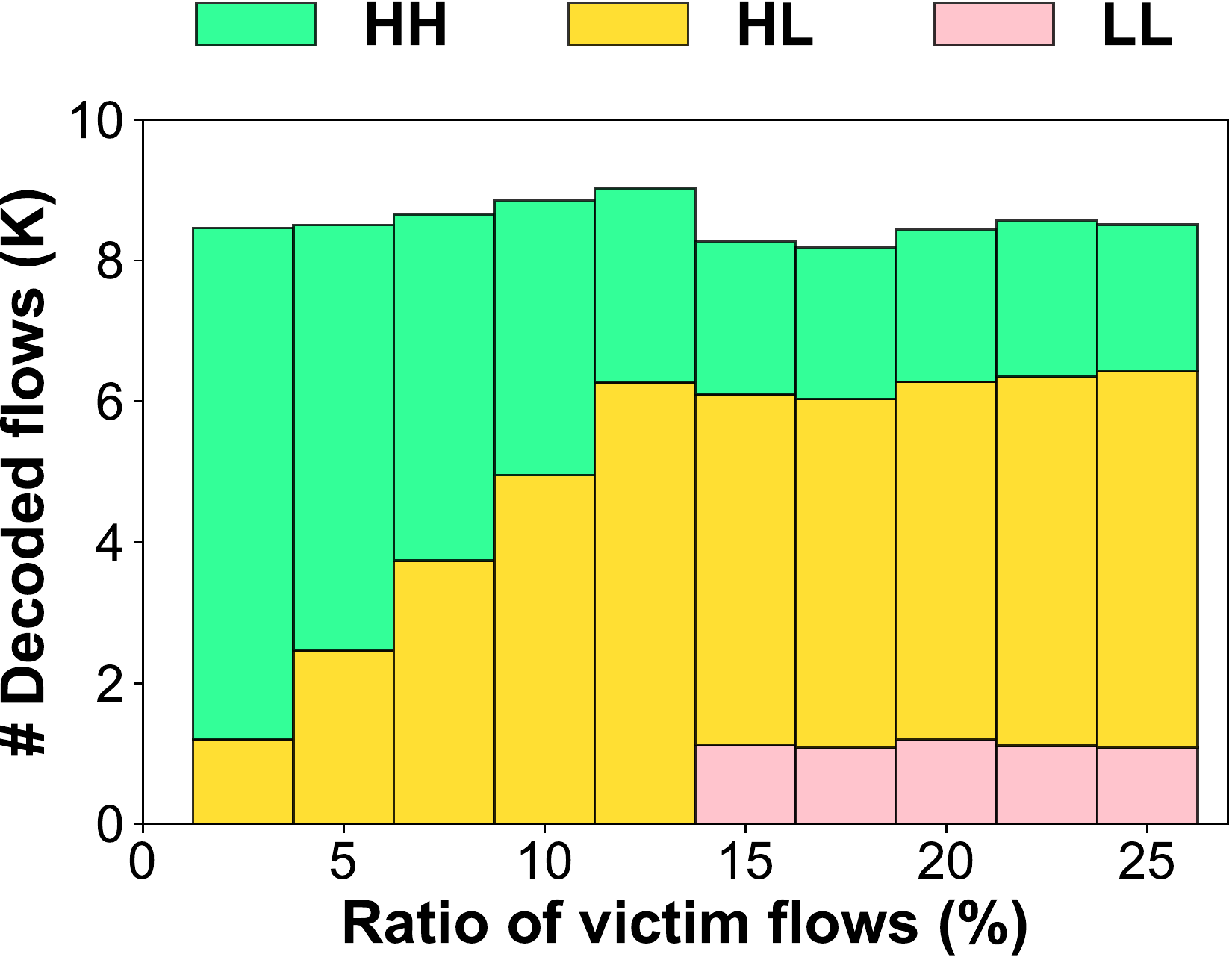}
    \label{fig:testbed:lossrate:num}
    }
    \subfigure[Threshold.]{ \includegraphics[width=0.22\textwidth]{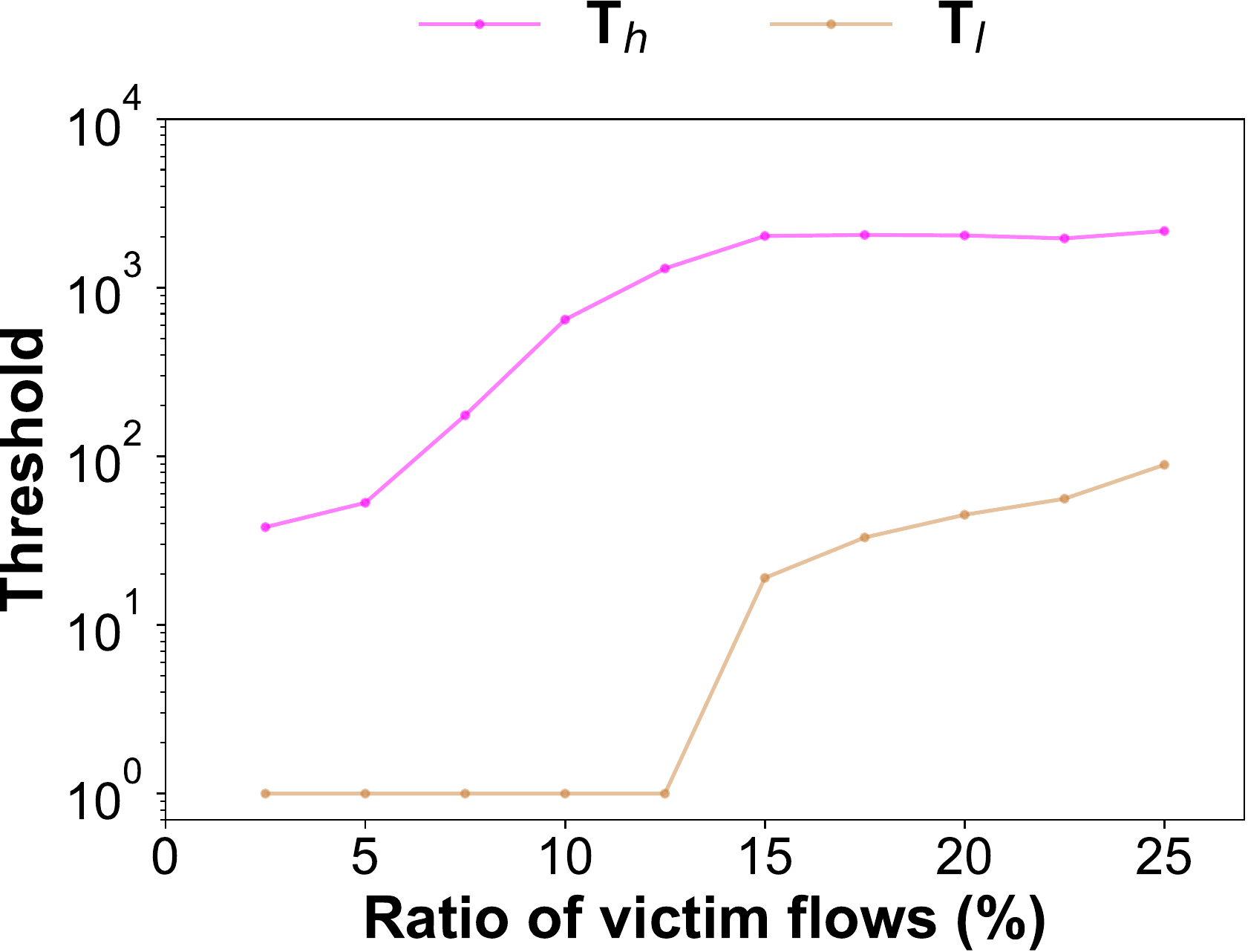}
    \label{fig:testbed:lossrate:thresh}
    }
    \subfigure[Sample rate.]{ \includegraphics[width=0.22\textwidth]{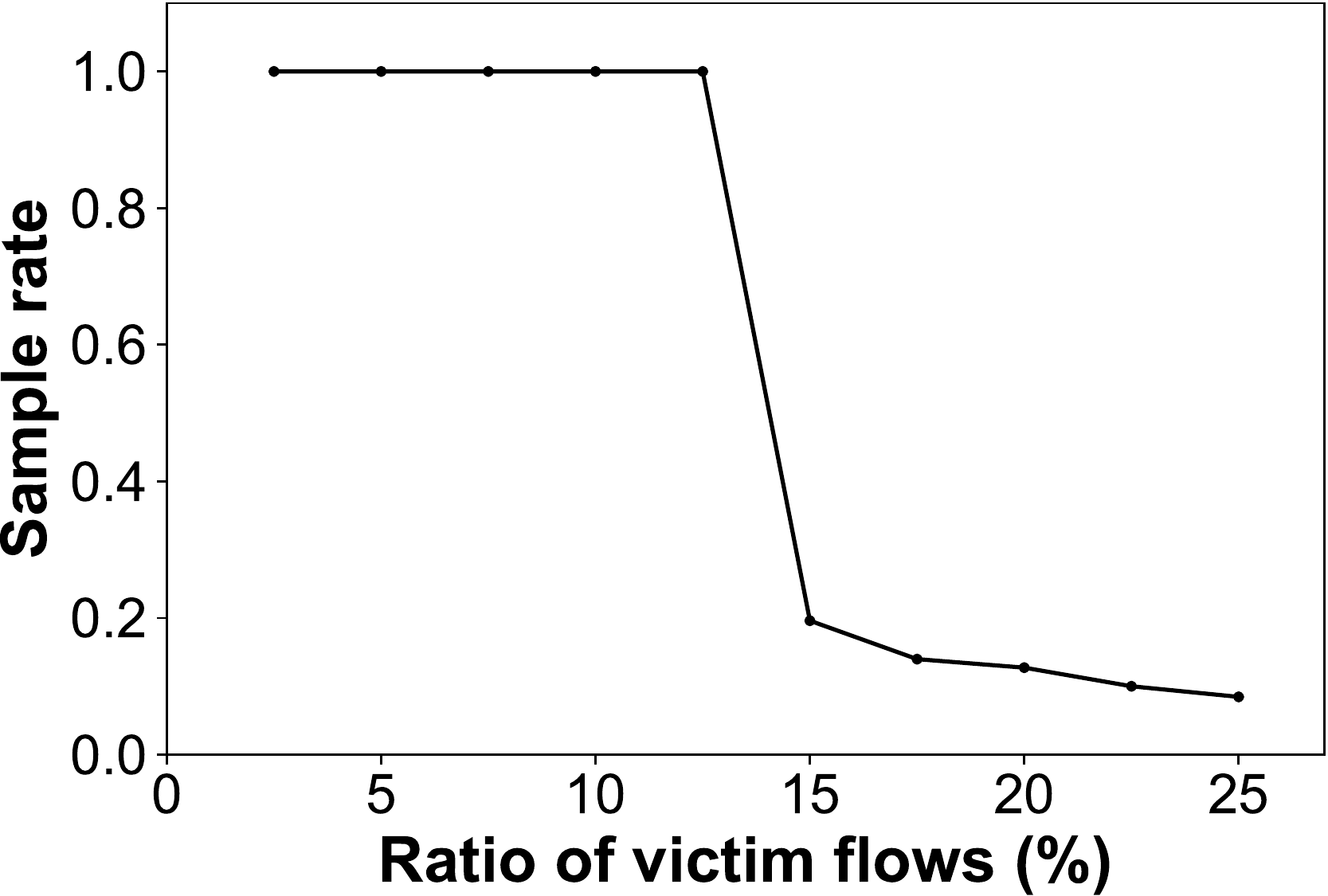}
    \label{fig:testbed:lossrate:sample}
    }
\caption{Measurement attention \textit{vs.} ratio of victim flows.}
\label{fig:testbed:lossrate}
\end{figure*}

First, on DCTCP workload, we evaluate whether \systemname{} can automatically shift measurement attention for different network states\footnote{For each data point of Figure \ref{fig:testbed:flownum}-\ref{fig:testbed:lossrate}, we collect it after \systemname{} successfully shifts measurement attention and the configuration of the \systemname{} data plane is stable.}.
For experimental results on the other three workloads, please refer to Appendix \ref{sec:workload}.

\bbb{Measurement attention \textit{vs.} number of flows (Figure~\ref{fig:testbed:flownum}):}
Experimental results show that \systemname{} can automatically shift measurement attention to packet loss detection while maintaining high memory utilization, as the number of flows increases and the number of victim flows increases.
%
We vary the number of flows in the network from 10K to 100K, and fix the ratio of victim flows to $10\%$.
At first, the network state is healthy.
As the number of flows increases from $10$K to $20$K, \systemname{} can record all flows and victim flows, and therefore sets both $T_h$ and $T_l$ to $1$.
As the number of flows increases from $30$K to $70$K, \systemname{} records all victim flows by allocating more and more memory to HL encoders.
However, \systemname{} cannot record all flows, and thus increases $T_h$ to decrease the number of HH candidates.
As the number of flows increases from $80$K to $100$K, \systemname{} cannot record all victim flows, and thus the network state transitions to the ill network state.
\systemname{} allocates fixed memory to LL encoders, increases $T_l$, and decreases the sample rate, so as to control the number of HLs and sampled LLs.
Meanwhile, \systemname{} keeps increasing $T_h$ to control the number of HH candidates.
Throughout the experiment, \systemname{} maintains high memory utilization.
The sum of decoded flows (Figure \ref{fig:testbed:flownum:num}) always exceeds 8K unless \systemname{} can record all flows and victim flows, representing a load factor larger than $65\%$ given that the upstream flow encoder has 12288 buckets.
It is acceptable considering that the target load factor of \systemname{} is 70\% and maximum load factor is $\frac{1}{1.23} = 81.3\%$.

\bbb{Measurement attention \textit{vs.} ratio of victim flows (Figure~\ref{fig:testbed:lossrate}):}
Experimental results show that \systemname{} can automatically shift measurement attention to packet loss detection while maintaining high memory utilization, as the ratio of victim flows increases and the number of victim flows increases.
We fix the number of flows to $50$K, and vary the ratio of victim flows from $2.5\%$ to $25\%$.
At first, the network state is healthy.
%
As the ratio of victim flows increases from $2.5\%$ to $12.5\%$, \systemname{} records all victim flows by allocating more and more memory to HL encoders, and increases $T_h$ to decrease the number of HH candidates.
As the ratio of victim flows increases from $15\%$ to $25\%$, \systemname{} cannot record all victim flows, and thus the network state transitions to the ill network state.
\systemname{} allocates fixed memory to LL encoders, increases $T_l$, and decreases the sample rate, so as to control the number of HLs and sampled LLs.
Meanwhile, because the memory of upstream HH encoder and the number of flows remain unchanged, $T_h$ also remains unchanged.
%
Throughout the experiment, \systemname{} maintains high memory utilization.
The sum of decoded flows (Figure \ref{fig:testbed:lossrate:num}) always exceeds 8K, representing a load factor larger than $65\%$.
It is acceptable considering that the target load factor of \systemname{} is 70\% and maximum load factor is 81.3\%.

Second, on DCTCP workload, we evaluate how fast can \systemname{} shift measurement attention over a large time window, in which the network state changes 8 times.

\begin{figure}[t]
\setlength{\subfigcapskip}{-0.5cm}
\setlength{\abovecaptionskip}{0.05cm}
\setlength{\belowcaptionskip}{-0.55cm}
    \centering  
\includegraphics[width=0.85\linewidth]{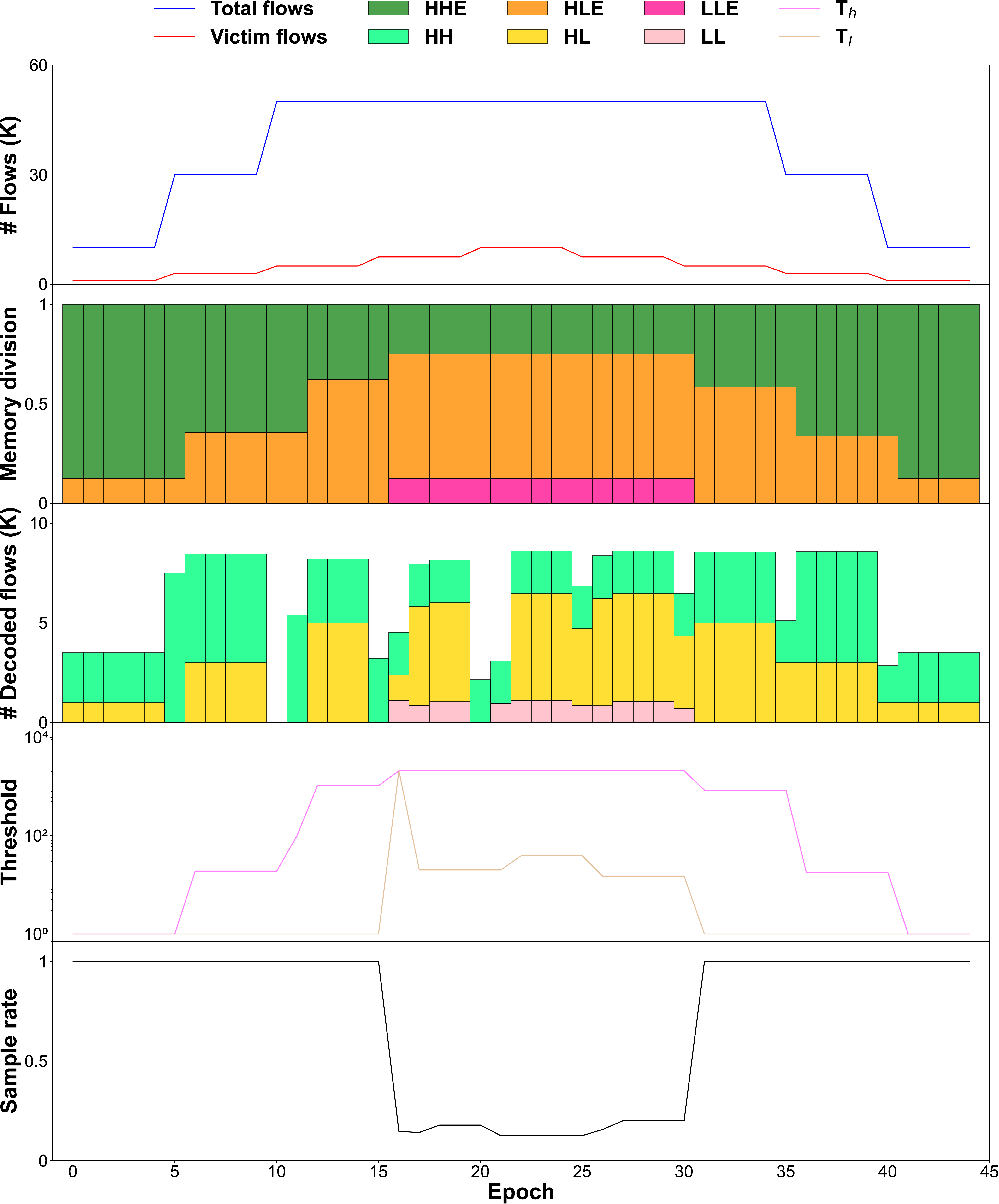}
    \caption{Measurement attention \textit{vs.} epoch.}
    \label{fig:testbed:shift-process}
\end{figure}

\bbb{Measurement attention \textit{vs.} epoch (Figure~\ref{fig:testbed:shift-process}):}
Experimental results show that \systemname{} can shift measurement attention within at most 3 epochs.
Figure~\ref{fig:testbed:shift-process} plots the shift of measurement attention in a large time window consisting of 45 consecutive epochs.
We change the network state (either the number of flows or the victim flow ratio) every 5 epochs, and the detailed settings are shown in the top sub-figure.
Overall, the network state first degrades from the healthy network state to the ill network state, and then improves back to the healthy network state.
For the eight changes, \systemname{} shifts measurement attention within 1 (6->7), 2 (11->13), 3 (16->19), 2 (21->23), 2 (26->28), 1 (31->32), 1 (36->37), and 1 (41->42) epochs, respectively. 


To evaluate how fast can \systemname{} monitor the network, we evaluate various factors that could affect the setting of epoch length: 1) the time and bandwidth required to collect sketches from edge switches, 2) the time required to respond to different network states, and 3) the time required to reconfigure the \systemname{} data plane.
We detail the corresponding experimental results in Appendix \ref{app:timeoverhead}.
In summary, we find that \systemname{} can monitor the network every 50ms with only one CPU core while consuming only 0.8\% bandwidth of a 40Gb NIC.
Thus, we believe \systemname{} can easily scale to monitor a much larger network with a shorter epoch length, requiring only one server as the central controller.

%% file: NSDI2022_9/6_Related.tex
\presec
\section{Related Work} \postsec
\label{sec:relatedworks}

First, we discuss prior art for packet loss tasks and packet accumulation tasks.
Then, we discuss prior art for resource management on switches.
We further discuss other prior art relevant to network measurement in Appendix \ref{app:related}.

\bbb{Prior art for packet loss tasks:}
They can be classified into two kinds of solutions.
The first kind is algorithm-oriented solutions, including LossRadar \cite{lossradar2016} based on Invertible Bloom filter (IBF) \cite{eppstein2011s}.
LossRadar can pinpoint the location of every lost packet and infer the root causes of packet losses by deploying IBF to monitor every link in the network.
The second kind is system-oriented solutions, including Netseer \cite{netseer2020}, PacketScope \cite{packetscope2020} and Dapper \cite{dapper} that are based on programmable switches.
%
Among them, Netseer utilizes the programmable data plane to detect both intra-switch and inter-switch packet losses, associate packet losses with direct causes, and batch lost packets to further reduce bandwidth overhead.
%
%
Both the above kinds of solutions are designed for only obtaining the exact difference set of flows/packets.
Therefore, they fail to meet versatility requirement as they can hardly be extended to packet accumulation tasks that require approximate flow sizes of all flows or simply large flows.

\bbb{Prior art for packet accumulation tasks:}
They can be classified into two kinds of solutions.
The first kind is sketches designed for specific packet accumulation task, including HashPipe \cite{sivaraman2017heavy}, R-HHH \cite{ben2017constant}, and more \cite{ben2016heavy,ben2018efficient,yang2019heavykeeper,schweller2004reversible}.
Among them, HashPipe designs a multi-stage data structure and kicks out small flows through comparison.
The second kind is sketches that support many packet accumulation tasks, including UnivMon \cite{univmon2016}, ElasticSketch \cite{elastic2018}, CocoSketch \cite{zhang2021cocosketch}, and more \cite{nitrosketch2019,sketchlearn2018,huang2017sketchvisor,beaucoup2020,cmsketch,cusketch,countsketch}.
%
Among them, CocoSketch proposes a key technique, namely stochastic variance minimization technique, to provide unbiased estimation for arbitrary partial key.
Both the above kinds of solutions choose to embrace hash collisions and provide approximate flow sizes for higher memory efficiency.
Therefore, they fail to meet versatility requirement as they can hardly be extended to packet loss tasks that require exact difference set of flows/packets.

\bbb{Prior art for both kinds of tasks:}
These solutions record the IDs and sizes of all flows in a zero-error manner.
Typical solutions include FlowRadar \cite{flowradar2016}, OmniMon \cite{omnimon2020}, Counter Braids \cite{lu2008counter}, Marple \cite{marple}, and more \cite{huang2021toward}.
Among them, FlowRadar encodes the IDs and sizes of all flows into a variant of IBLT \cite{goodrich2011invertibleIBLT} in switches, and then executes well-designed decoding schemes to retrieve exact flow IDs and sizes.
Marple designs a query language for a wide range of network measurement tasks, which relies on the programmable key-value store in switch hardware.
Marple requires an additional backing store to handle evicted flows.
These solutions fail to meet efficiency requirement as they record the exact IDs and sizes of all flows, incurring memory/bandwidth overhead linear with the number of flows.
Besides, \textit{INT-based solutions} that carry desired statistics in packet headers can potentially support both tasks given packet-level visibility.
Typical solutions include INT \cite{kim2015band}, PINT \cite{ben2020pint}, NetSight \cite{handigol2014know}, and more \cite{zhao2021lightguardian, sheng2021deltaint,yang2021sketchint,am-pm,sonchack2018scaling}.
However, INT-based solutions suffer from granularity-cost trade-off, and thus fail to meet either versatility requirement or efficiency requirement.

\bbb{Prior art for resource management:}
Due to the limited resources in hardware such as programmable switches, many solutions focus on resource management in measurement.
%
%
Some solutions \cite{p4all, opensketch2013, agarwal2022heterosketch} target at compile-time resource management.
Among them, HeteroSketch \cite{agarwal2022heterosketch} optimizes network-wide measurement by automatically optimizing the placement of sketches on heterogeneous devices.
These solutions differ from \systemname{} as \systemname{} executes memory reallocation at run-time.
Other solutions focus on run-time resource management \cite{misa2022dynamic, zheng2022flymon, sonata2018, marple,dream,moshref2015scream}.
Among them, FlyMon \cite{zheng2022flymon} achieves run-time reconfiguration of measurement tasks and resources.
However, these solutions does not focus on the resource management between packet loss tasks and packet accumulation tasks, which is our main focus.

%% file: NSDI2023/10Conclusion.tex
\vspace{-0.15cm}

\section{Conclusion}

In this paper, we present \systemname{}, which can automatically shift measurement attention as network state changes at run-time.
To achieve this, \systemname{} designs \sketchname{} to supports both packet loss tasks and packet accumulation tasks simultaneously.
We have fully implemented a \systemname{} prototype on a testbed consisting of 10 Tofino switches and 8 end-hosts. 
Experimental results on our testbed verify that 1) \systemname{} can achieve high accuracy in packet loss detection and six packet accumulation tasks; 2) \systemname{} can monitor the network every 50ms and shift measurement attention within at most 3 epochs as network state changes.
%

%

%% file: NSDI2022_9/Appendix/appendix.tex
\presec
\section{The \sketchname{} Algorithm}
\postsec
\label{sec:appendix:sketch}

\subsection{Pseudo-Code}
\label{sec:pseudo}
\begin{algorithm}[h!]
\renewcommand\baselinestretch{0.9}\selectfont
    \caption{Encoding/Insertion operation of \sketchname{}}
 \label{alg:insert}
 \KwIn{Flow ID $f$}
 \For {$i \in [1, d]$}
 {
     $j=h_i(f)$;\\
     $\mathcal{B}^{ID}_i[j]=(\mathcal{B}_i^{ID}[j]+f)$ mod $p$;\\
     $\mathcal{B}_i^c[j]++$;
 }
\end{algorithm}

\begin{algorithm}[h!]
\renewcommand\baselinestretch{0.9}\selectfont
    \caption{Decoding operation of \sketchname{}}
\label{alg:decode}
\SetKwProg{Fn}{Function}{:}{}
\SetKwFunction{FIsPure}{IsPure}
\SetKwFunction{FDelete}{Delete}
\SetKwFunction{FDecode}{Decode}
\Fn{\FIsPure{$i, j$}}
{
    $f=(\mathcal{B}_i^{ID}[j]\times(\mathcal{B}^c_i[j])^{(p - 2)})$ mod $p$;\\
    \KwRet $j==h_{i}(f)$;\\
}
\SetKwProg{Fn}{Function}{:}{}
\Fn{\FDelete{$\mathcal{B}_{i'}[j'], \mathcal{B}_i[j]$}}
{
    $\mathcal{B}_{i'}^{c}[j'] = \mathcal{B}_{i'}^{c}[j'] - \mathcal{B}_i^{c}[j]$\;
    $\mathcal{B}_{i'}^{ID}[j'] = (\mathcal{B}_{i'}^{ID}[j']-\mathcal{B}_i^{ID}[j])$ mod $p$\;
}
\SetKwProg{Fn}{Function}{:}{}
\Fn{\FDecode{}}{
    $Queue$ is an empty queue\;
    $Flowset$ is an empty map\;
    \For{$i\in[1, d], j\in[1, w]$}{
        \If{$\mathcal{B}_i^{c}[j]!=0$}{
            $Queue$.\rm push$(\mathcal{B}_i[j])$\;
        }
    }
    \While{\rm!$Queue.$\rm empty()}{
    $\mathcal{B}_{i}[j]=Queue$.front();\\
    $Queue$.pop()\;
    \If{\FIsPure{$i, j$}}{
        $f'=(\mathcal{B}_i^{ID}[j]\times(\mathcal{B}^c_i[j])^{(p - 2)})$ mod $p$\;
        $Flowset[f']=Flowset[f']+\mathcal{B}^c_{i}[j]$\;
        \For{$i'\in [1,d]$}{
            \FDelete($\mathcal{B}_{i'}[h_{i'}(f')],\mathcal{B}_{i}[j]$);\\
            \If{$\mathcal{B}_{i'}^{c}[h_{i'}(f')]!=0$}{
                $Queue$.push($\mathcal{B}_{i'}[h_{i'}(f')]$);
            }
        }
        }
    }
    \Return{Flowset}
}
\end{algorithm}

\subsection{Discussion of FermatSketch}
\label{sec:dicussion}
\bbb{Eliminating false positives during decoding:}
Due to hash collisions, the rehashing verification will inevitably misjudge some impure buckets as pure buckets with false positive rate $\frac{1}{m}$.
Such misjudgement will lead to extraction of flows that are not inserted into \sketchname{}, and finally could hinder the decoding.
From another point of view, extracting a flow from such a misjudged impure bucket, \ie, false positive, equals to inserting a wrong flow with a negative size into \sketchname{}.
The decoding operation can automatically eliminate these false positives: in the decoding procedure, these inserted wrong flows could also be extracted and deleted from \sketchname{}, and then the impact caused by the false positives disappears.

\bbb{Space complexity:} 
Suppose \sketchname{} is large enough, and the pure bucket verification has negligible false positive rate.
The decoding operation is almost the same as that of IBLT \cite{goodrich2011invertibleIBLT}, which is exactly the procedure used to find the 2-core of a random hypergraph \cite{dietzfelbinger2010tight,molloy2004pure}.
Therefore, the memory overhead of \sketchname{} is proportional to the number of inserted flows $M$, \ie, $\Theta(M)$. 
\sketchname{} also shares similar properties with IBLT: the number of hash functions, \ie, the number of the bucket arrays $d$, is recommended to set to $3$ for the highest memory efficiency, that on average $1.23$ buckets can record a flow and its size.
%


%
%

\bbb{Time complexity of decoding operation:}
Suppose \sketchname{} is large enough and the false positive rate in pure bucket verification is negligible.
In step\noindent\circled{1}, we traverse \sketchname{} and push all non-zero buckets into the decoding queue.
The number of these buckets is at most $md$, and thus the time complexity of step\noindent\circled{1} is $O(md)$.
In the rest steps, we process all the buckets pushed into the queue, which consists of two parts: 1) the $md$ buckets pushed into in step\noindent\circled{1}, and 2) the mapped buckets except the popped pure bucket of each extracted flow.
Considering that the number of extracted flows is bounded by the number of buckets of \sketchname{}, \ie, $md$, the number of buckets of the second part is $O(m d^2)$.
Therefore, the time complexity of the rest steps is $O(m d^2)$.
Adding up the time complexities of all steps, the time complexity of decoding operation is $O(m d^2)$.

\subsection{Proof of Theorem \ref{math:sketch}}
\label{app:math}
\begin{theorem}
\label{math:sketch}
Let \sketchname{} consists of $d$ bucket arrays, each of which consists of $m$ buckets.
Let $M$ be the number of flows inserted into that \sketchname{}.
Suppose $m d > c_d M+\epsilon$ and $M= \Omega (d^{4d}log^d(md))$.
the decoding of \sketchname{} fails with probability $O(\frac{1}{M^{d-2}})$, where both $\epsilon$ and $c_d$ are small constants,
$$c_d=\left(sup\left\{\alpha \Big| \alpha \in (0,1), \forall x\in(0,1), 1-e^{-d\alpha x^{d-1} } \right\}\right)^{-1}$$
For example, $c_3=1.23, c_4=1.30, c_5=1.43.$
\end{theorem}

\begin{proof}

This is an analysis based on the theory of the 2-core in random hypergraph \cite{molloy2004pure,dietzfelbinger2010tight} and IBLT. Compared with 2-core or IBLT, we only introduce a kind of additional error, which is the false positives when we use pure bucket verification to verify the pure buckets. The IBLT assumes there is no error when verifying buckets because IBLT uses additional hashkeySum field that can be long enough. 
The results of 2-core and IBLT show that the failure probability without wrong verification is $O(\frac{1}{M^{d-2}})$.
Here we aim at proving that the consequences of our false positives are negligible when $M$ is not too small.

In the decoding procedure, the pure bucket verification runs at most $O(Md)$ times, and the false positive rate is $O(\frac{1}{m})$ with only rehashing verification.
By Chernoff bound, when $M=O(md)$ and $\delta=O(\frac{1}{M^{d-2}})$, the number of false positives will not exceed $F=O(d^3 log(md))$ in most cases (\ie, $1-O(\delta)$).
A false positive will incur a wrong flow ID with a wrong single flow deletion that influences $d$ buckets.
There is at most $Fd$ buckets can be influenced, called poisoned buckets. 
The existing study \cite{goodrich2011invertibleIBLT} of poisoned bucket shows that a small number of poisoned buckets will be automatically recovered, and the probability of failure due to poisoned bucket is $O(\left(\frac{Fd}{M}\right)^d)=O(\frac{d^{4d} log^d(md)}{M^{(d-1)}})$.
When it satisfies that $M=\Omega(d^{4d} log^d(md))$, the overall failure probability is $\delta=O(\frac{1}{M^{d-2}})$.
In practice, $M=\Omega(d^{4d} log^d(md))$ is easy to meet because $M$ is large and $d$ is a small constant.
Here, we only use rehashing verification for pure bucket verification. 
The theorem can also be easily extended if we further use fingerprint verification.
\end{proof}
\subsection{Fingerprint Verification}
\label{app:fv}
To reduce the false positive rate of pure bucket verification, we can perform an extra verification method, namely \textit{fingerprint verification}, by extending the IDsum field in each bucket by $w$ bits and using the extra $w$ bits as a fingerprint.
For each incoming packet with flow ID $f$, a new hash function $h_{fp}(\cdot)$ gives it a $w$-bit fingerprint $h_{fp}(f)$ for checking whether a bucket is pure.
For encoding operation, instead of inserting flow ID $f$, we insert an extended ID concatenated by flow ID $f$ and fingerprint $h_{fp}(f)$, and the extended IDsum field stores the result of the sum of the extended IDs modulo prime $p$.
Note that $p$ must be a prime larger than any available extended ID.
For decoding operation, obviously, we can still perform rehashing verification with the extended ID.
Our fingerprint verification works as follows.
Suppose a bucket is pure.
%
%
First, we reuse the the extended ID of the single flow calculated in rehashing verification.
Then, we divide the extended ID to get the flow ID and its fingerprint.
If the divided fingerprint equals to the fingerprint of the divided flow ID, we consider the bucket passes fingerprint verification.
Only buckets pass both rehashing and fingerprint verification will be considered as pure.
The false positive rate of only fingerprint verification is obviously $\frac{1}{2^w}$.
Considering that rehashing verification and fingerprint verification are independent, the false positive rate of pure bucket verification could be reduced to $\frac{1}{2^w m }$ with $w$-bit fingerprint.

We conduct experiments to demonstrate the effect of 8-bit fingerprint on improving the decoding success rate.
As shown in Figure \ref{fig:sketch:fp:bucket}, when the number of flows is 1K, with the same number of buckets, 8-bit fingerprint can improve the decoding success rate by at most 6.73\%.
However, when the number of flows is 10K, the improvement falls to at most 2.26\%.
This is because as the number of buckets increases, $m$ increases, and the false positive rate of pure bucket verification quickly drops, and thus further reducing the false positive rate with fingerprint yields less improvement on the decoding success rate.
As shown in Figure \ref{fig:sketch:fp:mem}, under the same memory usage, 8-bit fingerprint actually reduces the decoding success rate.
This is because fingerprint consumes additional memory, while this memory could be used as buckets to reduce the probability of 2-core of the random hypergraph and improve the decoding success rate. 
Figure~\ref{fig:sketch:fingerprint}\subref{fig:sketch:fp:bucket}-\subref{fig:sketch:fp:mem} also demonstrate that the memory overhead of \sketchname{} is proportional to the number of inserted flows.
In summary, for simplicity and accuracy, we recommend implementing \sketchname{} without fingerprints in most cases.
Only if there is some memory that can hardly be utilized due to hardware constraints unless used as fingerprints, we recommend implementing \sketchname{} with fingerprints.

\begin{figure}[t!]
\setlength{\abovecaptionskip}{-0.15cm}
\setlength{\belowcaptionskip}{-0.0cm}
\centering
    \subfigure[Same number of buckets per flow.]{
        \begin{minipage}[b]{0.225\textwidth}
            \includegraphics[width=\textwidth]{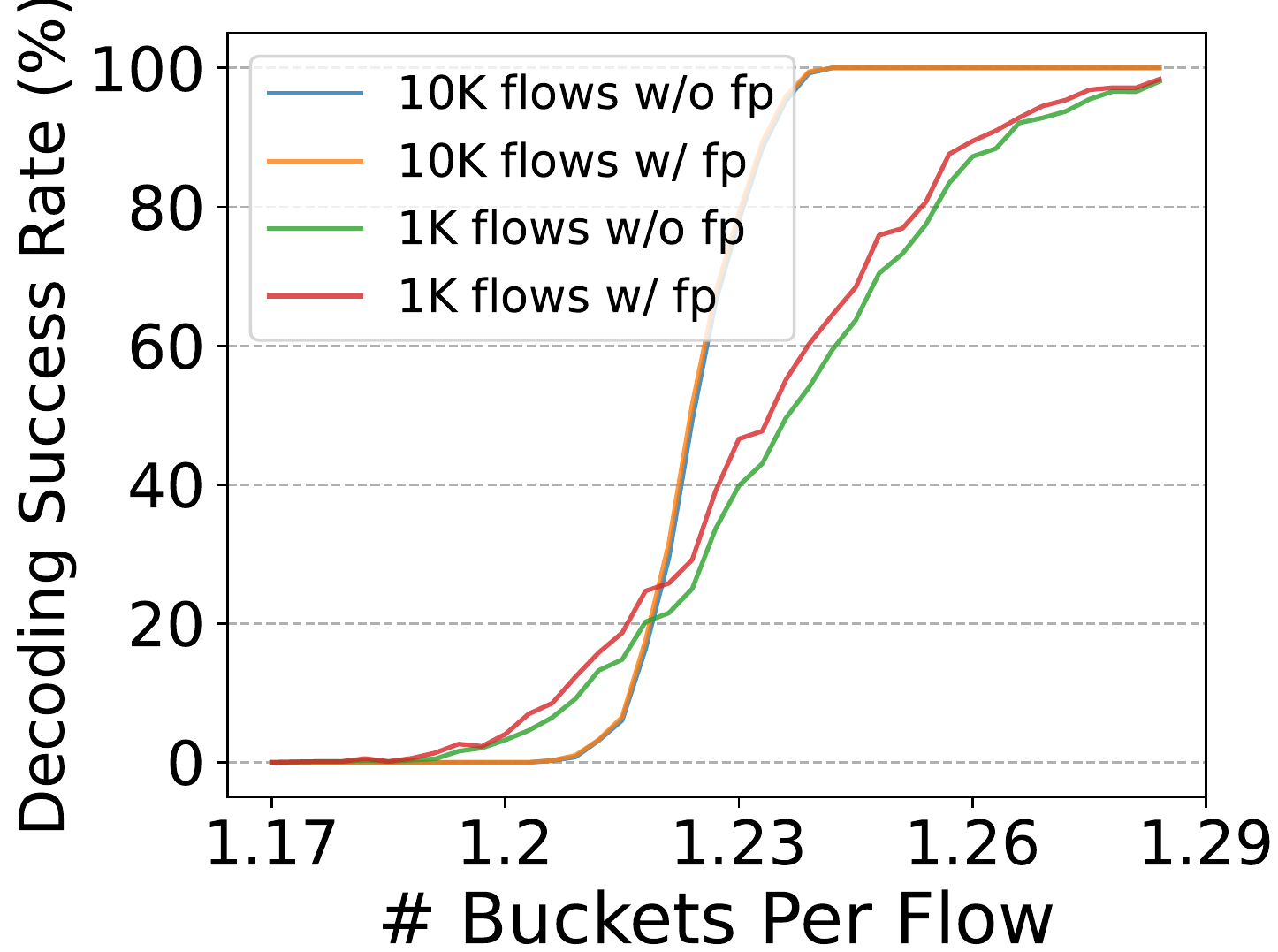}
            \label{fig:sketch:fp:bucket}
            \vspace{-0.2in}
        \end{minipage}
    }
    \subfigure[Same memory per flow.]{
        \begin{minipage}[b]{0.225\textwidth}
            \includegraphics[width=\textwidth]{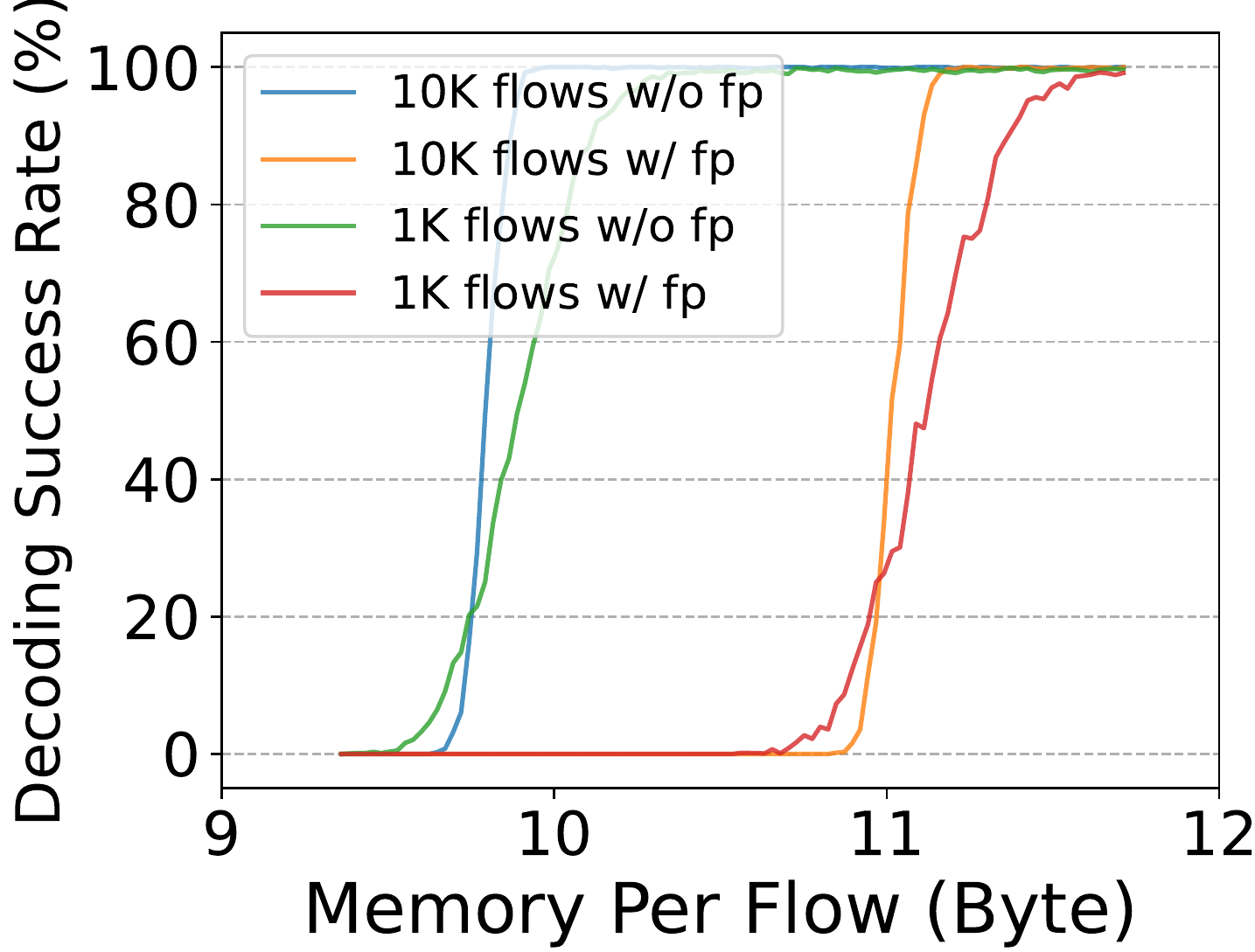}
            \label{fig:sketch:fp:mem}
            \vspace{-0.2in}
        \end{minipage}
    }
\caption{Experiments on 8-bit fingerprints. We use the anonymized IP traces collected in 2018 from CAIDA \cite{caida} as dataset.
We use the 32-bit source IP address as the flow ID, and choose the first 10K flows for experiments. Here, \textbf{fp} represents 8-bit fingerprint.}
\label{fig:sketch:fingerprint}
\end{figure}

\section{Collection from Data Plane}
\label{appendix:collection}

\bbb{Timeline split:}
For each edge switch, we maintain a 1-bit timestamp in its ingress, which is periodically flipped by the switch control plane, so as to split the timeline into consecutive fixed-length epochs with interleaved 0/1 timestamp value.
Further, we copy an additional group of sketches in the switch data plane for rotation.
Each group of sketches corresponds to a distinct timestamp value (0/1), and monitors the epochs with that timestamp value.
Specifically, at each edge switch, every packet entering the network first obtains the current timestamp value, and then is inserted into the flow classifier and upstream flow encoder corresponding to the obtained timestamp value.
When the packet exits the network, it is also inserted into the downstream flow encoder corresponding to the timestamp value it obtained when it entered the network.
To maintain the timestamp value during the packet transmission, we can use one unused bit in the original packet header as discussed above ($\S$~\ref{sec:dfe}). 

\bbb{Clock synchronization:}
Through maintaining a 1-bit timestamp and copying a group of sketches, we successfully split the timeline and insert packets of different epochs to their corresponding groups of sketches, laying a solid foundation for subsequent collection.
However, if the clocks of the control planes of edge switches are out of synchronization to some extent, we still can not find opportunities to collect the sketches without colliding with packet insertion.
Considering such an extreme situation.
There are three edge switches in a given network, and the transmission time between any two edge switches is the same.
The time offset between the control planes of two of the edge switches is exactly the size of the epoch, \ie, at any time, the flipping timestamps of the two edge switches are different (0<->1).
There are continuous packets entering the network at the above two edge switches, and exiting the network at the third switch. 
As a result, both groups of sketches of the third switch are continuously inserted, and thus can never be collected.
To address this issue, the central controller synchronizes the clocks of the control planes of all edge switches with itself, trying to keep only a group of sketches being inserted at any time, so as to make opportunities to collect the other group.
Then, we further discuss the appropriate time for the central controller to collect sketches.

\bbb{Appropriate time for collection:}
The central controller also maintains a 1-bit timestamp, trying to collect the group of sketches monitoring the previous epoch after it ends.
Before collection, the central controller should ensure that all the packets in the previous epoch have been inserted into sketches or lost in the network, so as to guarantee the correctness of measurement. 
%
%
First, we analyze an ideal situation, that the clock synchronization is zero-error.
For ingress sketches, \ie, the flow classifier and the upstream flow encoder, as soon as the locally maintained 1-bit timestamp flips, the central controller can collect the group of ingress sketches monitoring the previous epoch from each edge switch, because all the packets in that epoch have already been inserted into ingress sketches.
For egress sketches, \ie, the downstream flow encoder, every time the locally maintained timestamp flips, the central controller must first wait an additional period of time, so as to ensure that all the packets in the previous epoch have either been lost in the network, or passed through the network and been inserted into egress sketches.
Then, the central controller can collect the group of egress sketches monitoring the previous epoch from each edge switch.
Obviously, the additional period of time should be longer than the maximum time for packet transmission in the network.
Considering that the buffer sizes of DCN switches are at 10MB-level \cite{zeng2019congestion}, with 100Gb link speed, the queuing delay in a single switch is at most 1ms in most cases.
Therefore, for typical data center networks that usually have at most five hops, setting the additional time to 10ms can cope with most cases.
However, in practice, the clock synchronization can never be zero-error.
Therefore, before collecting both ingress and egress sketches, the central controller needs to wait for another additional period of time, which should be longer than the precision of synchronization, so as to guarantee the correctness of measurement.
In addition, the central controller should end the collection some time before its 1-bit timestamp flips again, which should also be longer than the precision of synchronization, in case the packets in the next epoch are inserted into the group of sketches being collected.

%% file: NSDI2022_9/Appendix/setup_exp.tex
\presec
\section{Evaluation on Packet Accumulation Tasks}
\postsec
\label{eval:2}

\bbb{Metrics:}
\begin{itemize}[leftmargin=*,parsep=0pt,itemsep=0pt,topsep=0pt,partopsep=0pt]
    \item \textit{Average Relative Error (ARE):} $\frac{1}{|\Omega|}\sum_{f_i \in \Omega}\frac{|v_i-\hat{n}_i|}{v_i}$, where $\Omega$ is the set including all flows, $v_i$ is the true size of flow $f_i$, and $\hat{v}_i$ is the estimated size of flow $f_i$. 
    
    \item \textit{$F_1$ Score:} $\frac{2\cdot PR\cdot RR}{PR+RR}$, where $PR$ (Precision Rate) refers to the ratio of the number of the correctly reported instances to the number of all reported instances, and $RR$ (Recall Rate) refers to the ratio of the number of the correctly reported instances to the number of all correct instances. 
    \item \textit{Relative Error (RE):} $\frac{|True-Est|}{True}$, where $True$ and $Est$ are the true and estimated statistics, respectively.
    \item \textit{Weighted Mean Relative Error (WMRE) \cite{huang2017sketchvisor}:} $\frac{\sum_{i=1}^{z}|n_i-\hat{n_i}|}{\sum_{i=1}^{z} \left(\frac{n_i+\hat{n_i}}{2}\right)}$, where $z$ is the maximum flow size, $n_i$ and $\hat{n_i}$ are the true and estimated numbers of the flows of size $i$, respectively.
    
\end{itemize}

\bbb{Dataset:}
We also use the IP traces from CAIDA \cite{caida} as our dataset, and use the 32-bit source IP address as the flow ID.
We use four traces for evaluation, each of which monitors the traffic in five seconds.
Each trace contains 63K flows and 2.3M packets on average.
We report the average accuracy that each algorithm achieves on each CAIDA trace.

\bbb{Setup:}
We compare the combination of TowerSketch and \sketchname{} (Tower+Fermat) with 9 algorithms: CM \cite{cmsketch}, CU\cite{cusketch}, CountHeap \cite{countsketch}, UnivMon \cite{univmon2016}, ElasticSketch \cite{elastic2018}, FCM-sketch \cite{song2020fcm}, HashPipe \cite{sivaraman2017heavy}, CocoSketch \cite{zhang2021cocosketch}, and MRAC \cite{kumar2004data}.
We do not compare with FlowRadar because FlowRadar can hardly perform successful decoding with the memory sizes we used for evaluation (200KB-600KB).
For heavy-hitter detection and heavy-change detection, we set their thresholds $\Delta_h$ and $\Delta_c$ to about 0.02\% and 0.01\% of the total packets, \ie, 500 and 250, respectively.
We configure Tower+Fermat and its competitors as follows.
Overall, the configurations of these competitors are recommended in literature. 
\begin{itemize}[leftmargin=*,parsep=0pt,itemsep=0pt,topsep=2pt,partopsep=2pt]

\item\textit{Tower+Fermat:} 
For Tower, we set it to consist of an $8$-bit counter array and a $16$-bit counter array.
For Fermat, We set its count field and ID field to 32bits, and allocate $2500$ buckets to it for 99.9\% decoding success rate.

we set the threshold $T_h$ for identifying heavy-hitter candidates to the heavy-change threshold $\Delta_c$, \ie, 250, for detecting most heavy-hitters and heavy-changes.

\begin{figure}[h!]
\setlength{\subfigcapskip}{-0.6cm}
\setlength{\abovecaptionskip}{-0.15cm}
\setlength{\belowcaptionskip}{-0.05cm}
\centering
 \hspace{0.05cm}
    \subfigure[Heavy-hitter]{
        \begin{minipage}[b]{0.2\textwidth}
            \includegraphics[width=\textwidth]{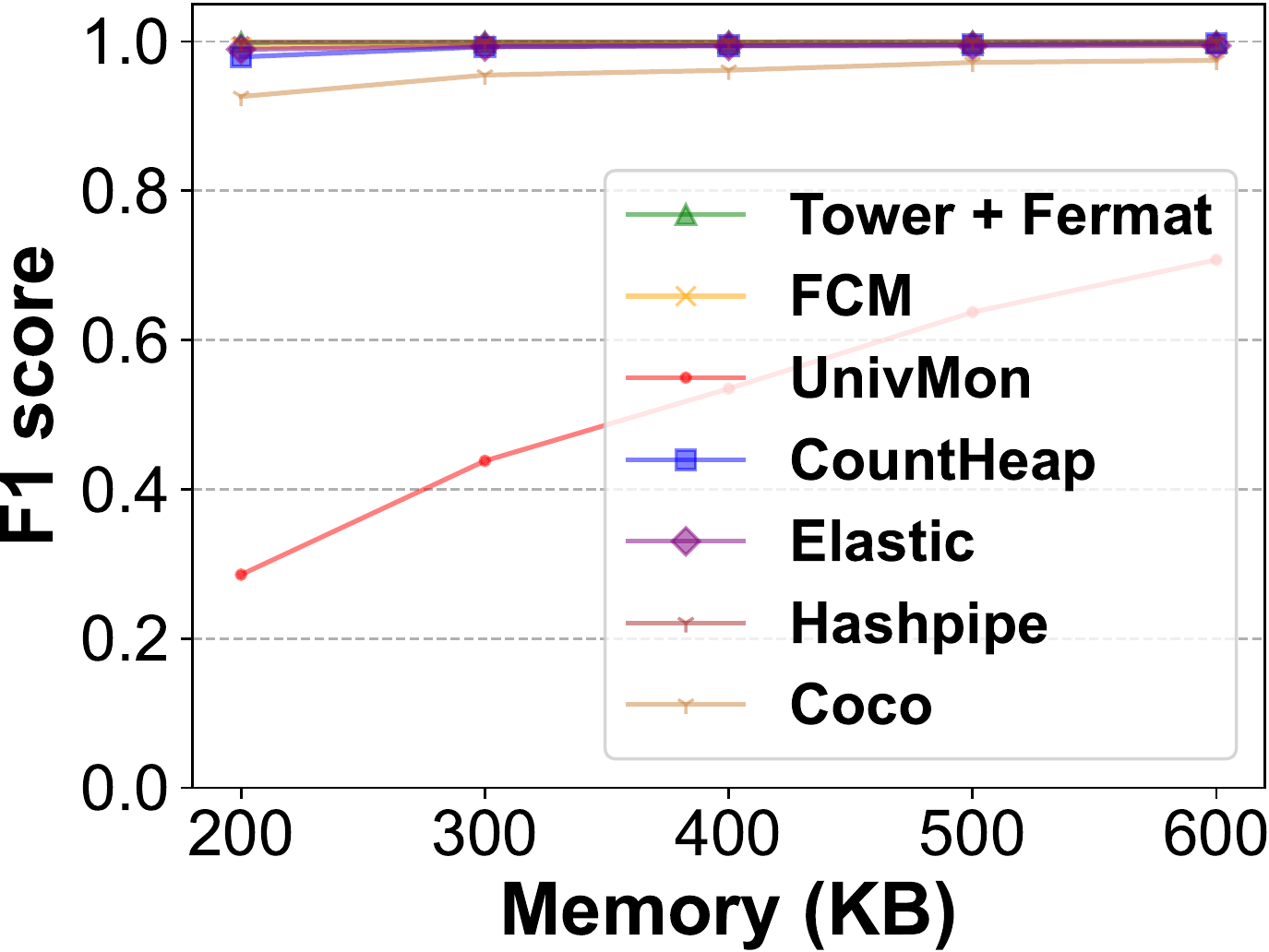}
            \label{fig:accuracy:HH F1 score}
        \end{minipage}
    }
    \subfigure[Flow size]{
        \begin{minipage}[b]{0.19\textwidth}
            \includegraphics[width=\textwidth]{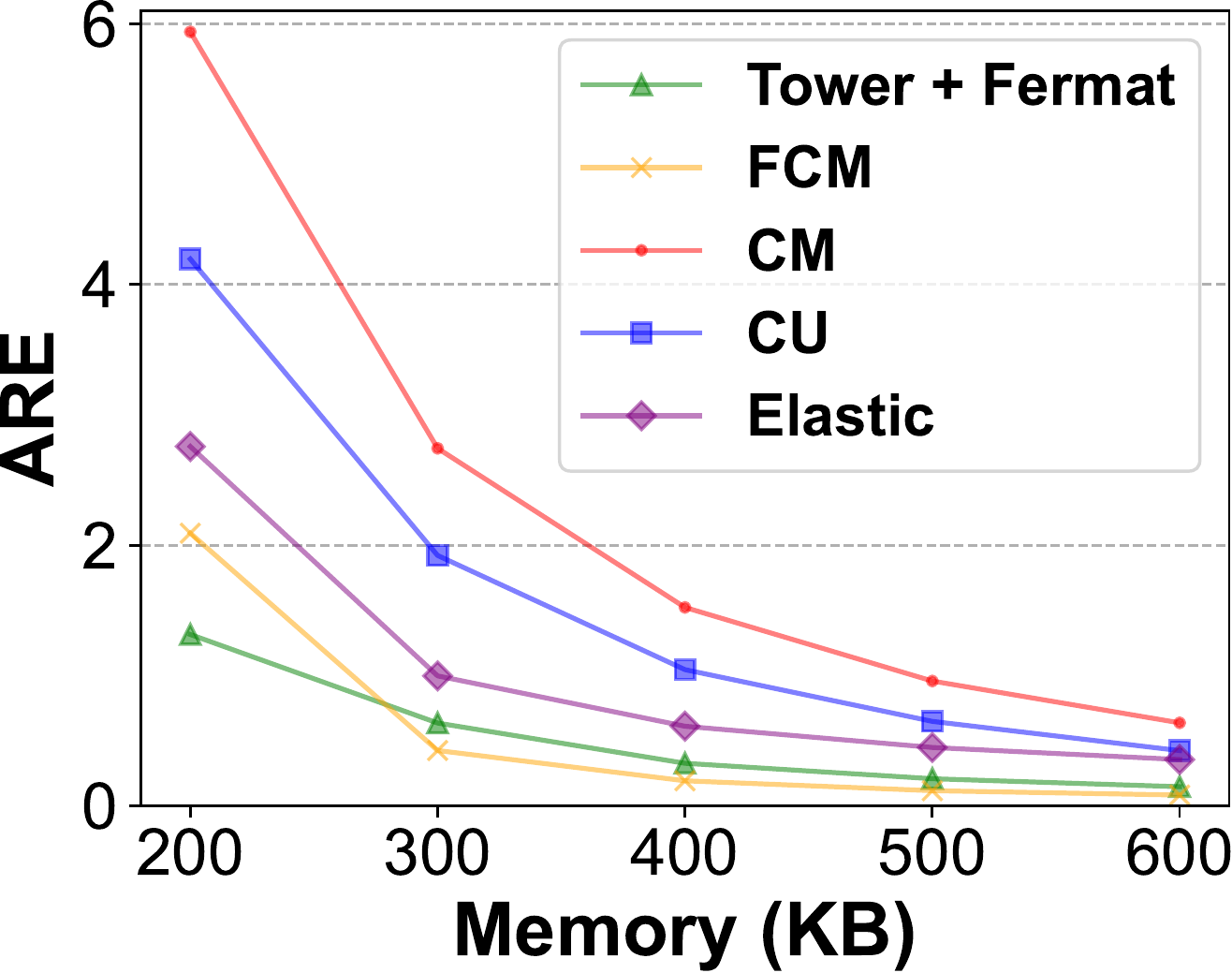}
            \label{fig:accuracy:ARE}
        \end{minipage}
    }
    
    \vspace{-0.4cm}
   
    \subfigure[Heavy-change.]{
        \begin{minipage}[b]{0.19\textwidth}
            \includegraphics[width=\textwidth]{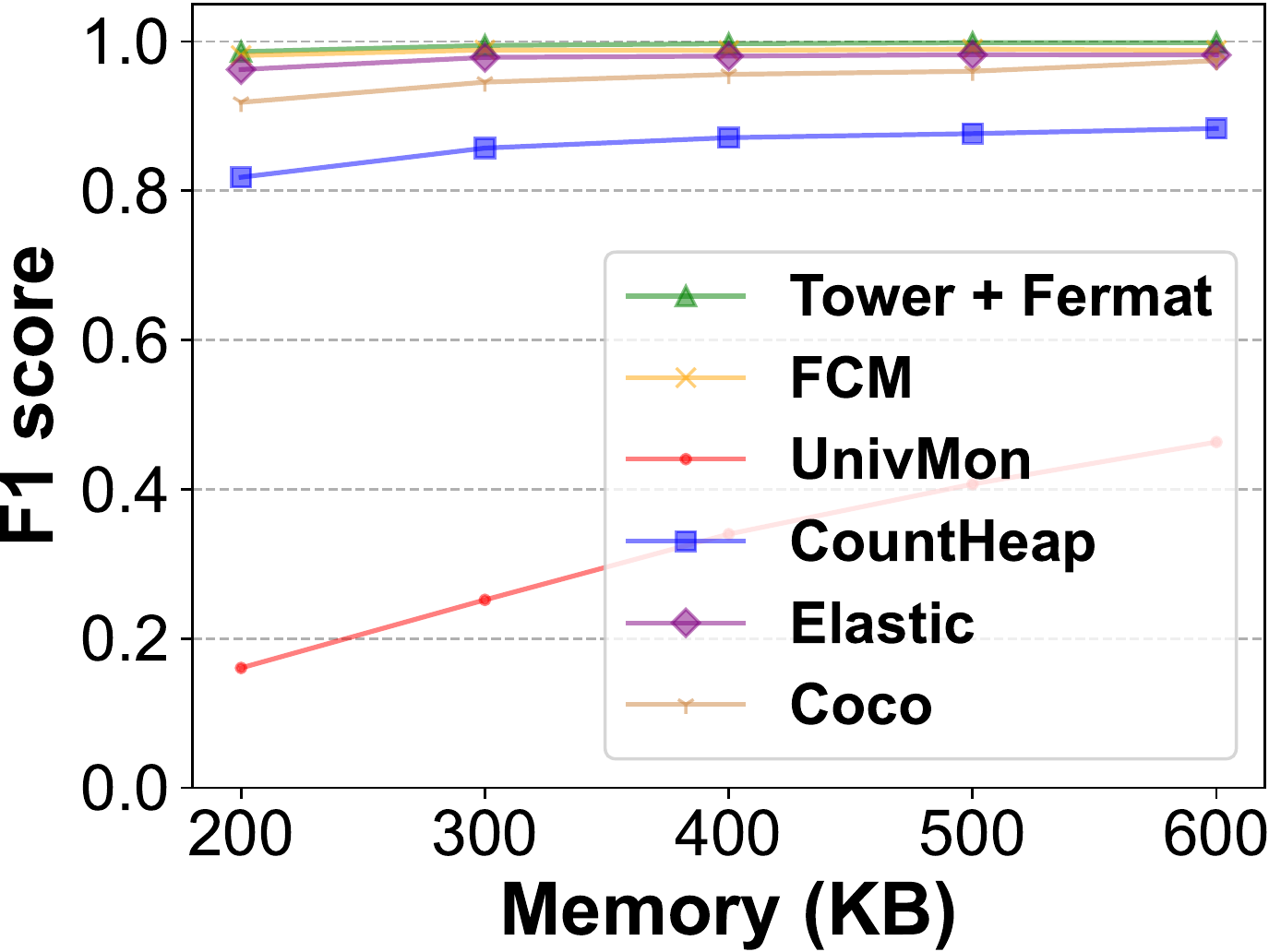}
            \label{fig:accuracy:HC F1 score}
        \end{minipage}
    }
    \subfigure[Flow size distribution.]{
        \begin{minipage}[b]{0.2\textwidth}
            \includegraphics[width=\textwidth]{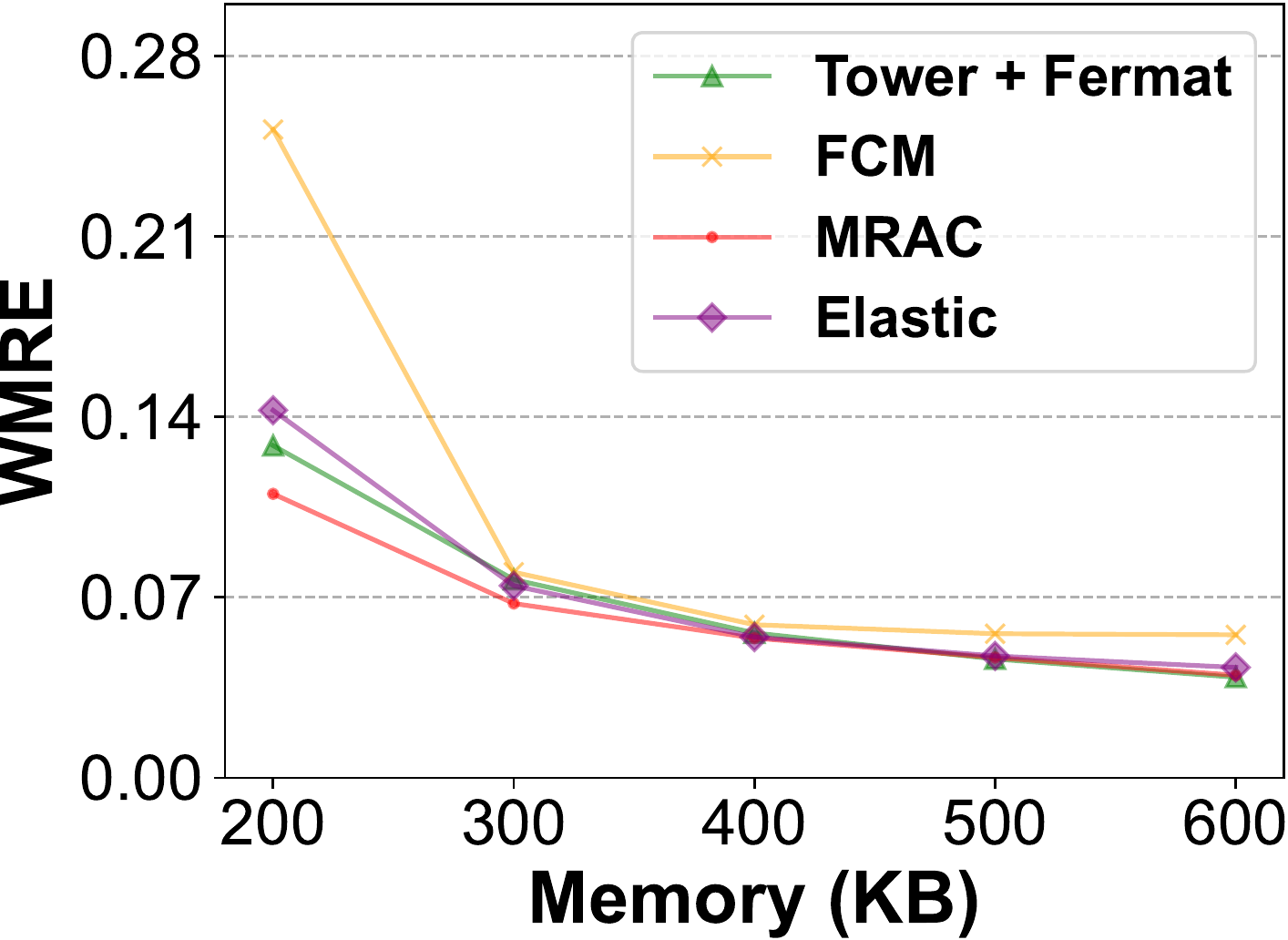}
            \label{fig:accuracy:FSD WMRE}
        \end{minipage}
    }
    
    \vspace{-0.4cm}
    
    \subfigure[Entropy.]{
        \begin{minipage}[b]{0.2\textwidth}
            \includegraphics[width=\textwidth]{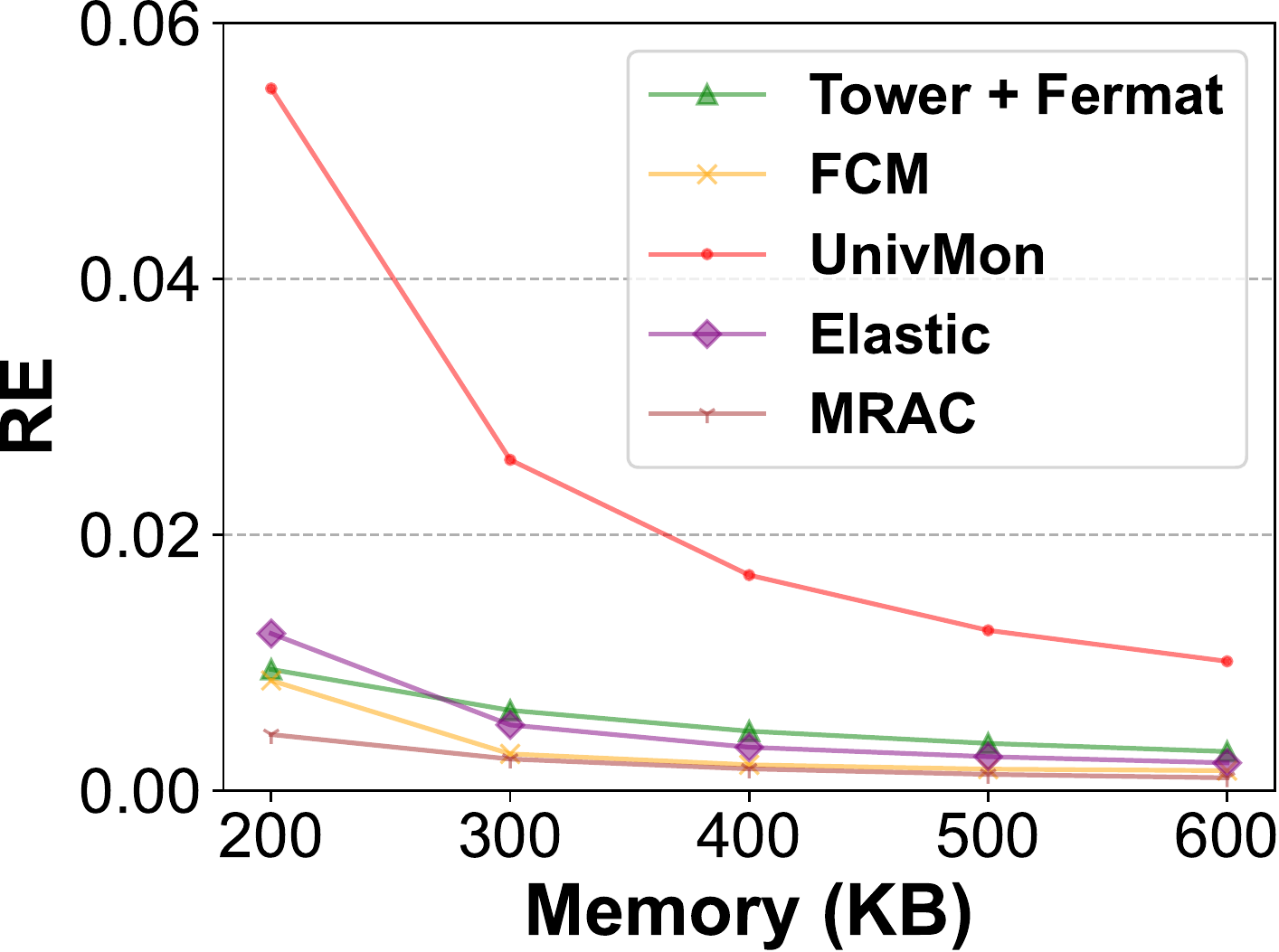}
            \label{fig:accuracy:ENTR RE}
        \end{minipage}
    }
    \subfigure[Cardinality.]{
        \begin{minipage}[b]{0.2\textwidth}
            \includegraphics[width=\textwidth]{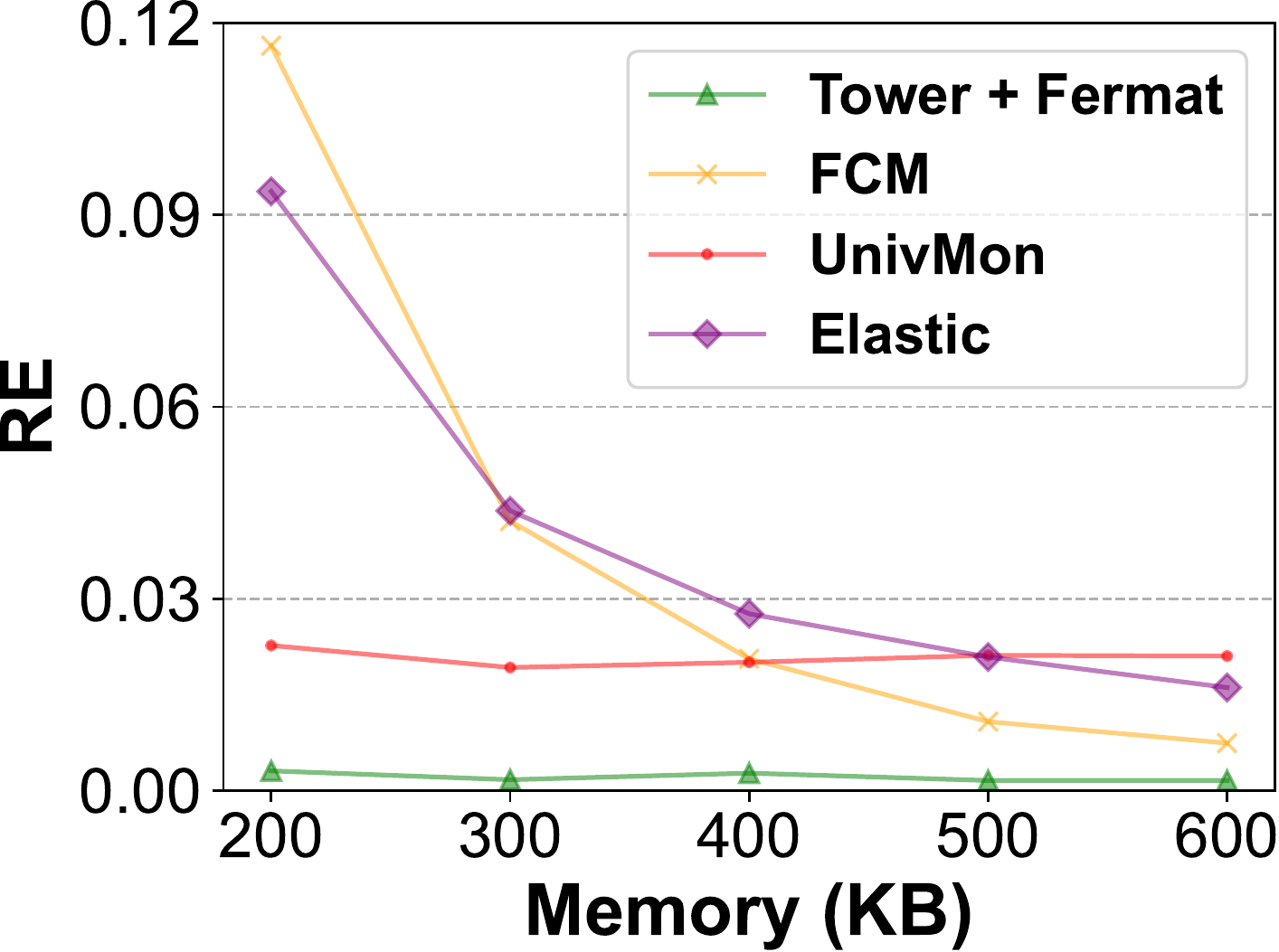}
            \label{fig:accuracy:CARD RE}
        \end{minipage}
    }
\caption{Accuracy comparison for six tasks.}
\label{fig:accuracy}
\end{figure}

\item\textit{CM/CU/CountHeap:}
We use 3 hash functions as recommended in \cite{goyal2012sketch}.
We set the counter size to 32bits.
For CountHeap, we additionally set its heap capacity to 4096 for heavy-hitter detection.

\item\textit{UnivMon:}
We use 14 levels and each level can record 1000 heavy hitters. 

\item\textit{Elastic:} 
We use the hardware version of Elastic.
For the heavy part, we use 4 stages and each stage has 3072 buckets. 
For the light part, we use a one-layer CM with 8-bit counters.
\item\textit{FCM:}
We use the top-$k$ version of FCM. 
It is almost the same as Elastic except the light part is substituted by a 16-ary FCM whose depth is set to $2$.
\item\textit{Hashpipe:}
We set the number of stages to $6$. 
\item\textit{Coco:}
We use the hardware version of Coco that only uses one hash function.

\end{itemize}

\bbb{Heavy-hitter detection (Figure \ref{fig:accuracy:HH F1 score}):}
Experimental results show that Tower+Fermat achieves comparable accuracy with HashPipe, and higher accuracy than other algorithms.
When using only 200KB memory, the F1 score of Tower+Fermat is $99.8\%$, while that of Elastic and FCM is lower than $99\%$.

\bbb{Flow size estimation (Figure~\ref{fig:accuracy:ARE}):}
Experimental results show that Tower+Fermat achieves comparable accuracy with FCM, and higher accuracy than other algorithms.
%
When using only 200KB memory, the ARE of Tower+Fermat is $4.51$ times, $3.19$ times, $2.09$ times, and $1.59$ times smaller than that of CM, CU, Elastic, and FCM, respectively.

\bbb{Heavy-change detection (Figure \ref{fig:accuracy:HC F1 score}):}
Experimental results show that the Tower+Fermat achieves higher accuracy than other algorithms.
Tower+Fermat achieves $99.6\%$ F1 score when using only $400$KB memory, while that of the other algorithms is below $99.0\%$.

\bbb{Flow size distribution estimation (Figure \ref{fig:accuracy:FSD WMRE}):}
Experimental results show that Tower+Fermat achieves higher accuracy than Elastic and FCM, and comparable accuracy with MRAC.
When using 600KB memory, the WMRE of Tower+Fermat is $0.039$, $1.09$ and $1.42$ times smaller than that of Elastic and FCM, respectively.

\bbb{Entropy estimation (Figure \ref{fig:accuracy:ENTR RE}):}
Experimental results show that Tower+Fermat achieves higher accuracy than UnivMon, and comparable accuracy with Elastic and FCM.
When using 600KB memory, the ARE of Tower+Fermat is $0.003$, $3.3$ times smaller than that of UnivMon.

\bbb{Cardinality estimation (Figure \ref{fig:accuracy:CARD RE}):}
Experimental results show that the Tower+Fermat achieves higher accuracy than other algorithms.
When using 600KB memory, the RE of Tower+Fermat is $0.0016$, $13.125$ times, $10.08$ times, and $4.57$ times smaller than that of UnivMon, Elastic, and FCM, respectively.


%% file: NSDI2023/8Implementation.tex
\presec
\section{Prototype Implementation}
\postsec
\label{sec:protoimplementaion}
In this section, we present the important details of \systemname{} prototype. We lay out important implementation details of the \systemname{} data plane and control plane in sequence.

\begin{figure}[t!]
    \setlength{\abovecaptionskip}{0.05cm}
    \setlength{\belowcaptionskip}{0cm}
    \prefig
    \centering  
    \includegraphics[width=1\linewidth]{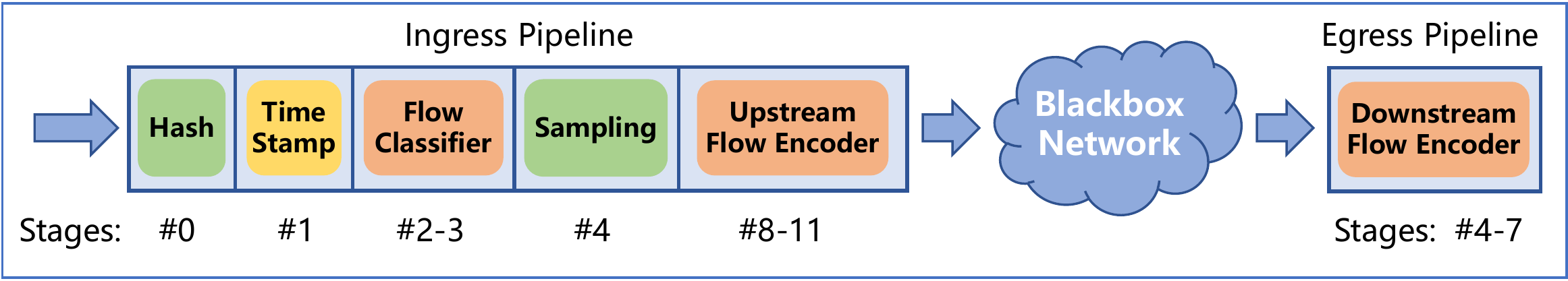}

    \caption{Implementation logic of \systemname{}.}
    \label{algopic:dataimpl}
    \postfig
\end{figure}

\presec
\subsection{Data Plane Implementation}
\postsec

\label{sec:dataimpl}
We have fully implemented the \systemname{} data plane on the switch data planes of four edge Tofino switches in P4 \cite{bosshart2014p4}.
In this section, we detail the implementation logic of data plane along the workflow (Figure \ref{algopic:dataimpl}).

\bbb{Hash:}
First, a packet with flow ID $f$ enters the network at an edge switch. 
With its flow ID (5-tuple) as input, the packet is hashed to multiple indexes through pairwise-independent hash functions generated from different CRC polynomials, which are deployed at stage 0 in ingress.
These hash indexes are either used as base indexes for locating the mapped counters/buckets in the subsequent insertions, or used for sampling LL candidates, or used as fingerprints for improving decoding success rate of \sketchname{}.
Note that due to the limitation of Tofino switches, each hash index is uniformly distributed on [$0$, $2^t-1$], where $t$ is an arbitrary positive integer.
%

\bbb{1-bit flipping timestamp:}
Second, the packet reads the current 1-bit flipping timestamp and from a match-action table, which is deployed at stage 1 in ingress.
The 1-bit timestamp is used to indicate the corresponding group of sketches for the subsequent insertions.
%

\bbb{Flow classifier:}
Third, the packet is inserted into the flow classifier, which is deployed at stage 2-3 in ingress.
The flow classifier is a TowerSketch consisting of an 8-bit counter array and a 16-bit counter array.
The 8-bit and 16-bit counter arrays consist of $w_1$ 8-bit and $w_2$ 16-bit counters, respectively.
Each counter array is built on a register and accessed by a stateful arithmetic logic unit (SALU).
To save SALU resources, we simulate the two flow classifiers by doubling the number of counters of the 8-bit and 16-bit counter arrays instead of building additional registers.
The left/right $w_1$ 8-bit and $w_2$ 16-bit counters form the flow classifier corresponding to timestamp $0$/$1$, respectively.
We use the base indexes calculated by hash functions as the relative positions of the mapped counters in the flow classifier, and add offsets corresponding to the 1-bit timestamp to the base indexes, so as to locate the mapped counters.
Specifically, when the timestamp is $0$, the offset is 0; when the timestamp is $1$, the offset is $w_1$ for 8-bit counter array or $w_2$ for 16-bit counter array.
During insertion, the SALU adopts saturated addition operation for each mapped counter, which can increment the counter to its maximum value but never overflow it, and reports the value recorded in the counter, so as to simulate the behavior of TowerSketch.
After insertion, we take the minimum value among the reported values as the size of flow $f$, and then input the flow size to a match-action Table that uses range matching on the flow size, so as to obtain the hierarchy of flow $f$.

\bbb{Sampling:}
Fourth, if flow $f$ is classified as a LL candidate, the packet reads a value from a match-action table, which is deployed at stage 4 in ingress.
We then compare the read value with a 16-bit value, which is calculated by a hash function with the 5-tuple of the flow and a random seed as input.
If the read value is equal to or larger than the 16-bit value, flow $f$ is classified as a sampled LL candidate.
Otherwise, the flow is classified as a non-sampled LL candidate.
Obviously, to simulate a sample rate $R$, the value should be set to $\lceil 65536 \times R \rceil$.

\begin{figure}[t!]
    \setlength{\abovecaptionskip}{0.05cm}
    \setlength{\belowcaptionskip}{0cm}
    \prefig
    \centering  
    \includegraphics[width=0.96\linewidth]{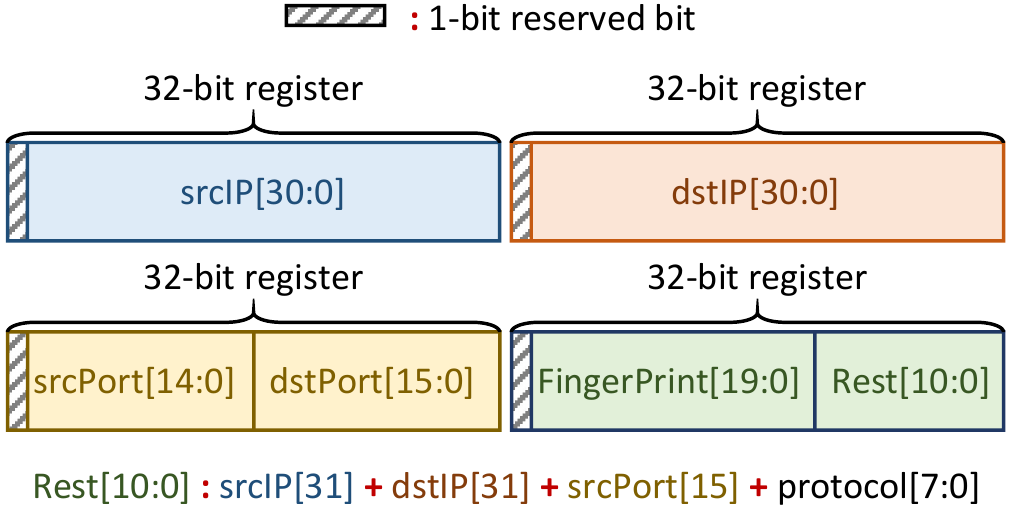}
    \caption{Division of the 5-tuple.}
    \label{ID:division}
    \postfig
\end{figure}

\begin{algorithm}[t!]
\renewcommand\baselinestretch{0.9}\selectfont
    \SetAlgoNoLine
    \caption{Simulated modular addition.}
    \label{algo:code:moduloadd}
    \KwIn{An ID fragment $f$, a counter $reg$ for encoding the ID fragment, and a prime $p$.}
    
    $inv \gets p - f$ {\Comment{Get the additive inverse of $f$ in $Z_p$}}
    
    
    \eIf{$reg + f < p $} 
    {$reg \gets reg + f$}
    {$reg \gets reg - inv$}
    
\end{algorithm}

\bbb{\sketchname{}:}
Before detailing the implementation of upstream and downstream flow encoders, we present the implementation of \sketchname{} that they are based on.
To encode the 104-bit flow ID (5-tuple) of each packet, an ideal bucket in \sketchname{} should contain a 105-bit IDsum field and a 32-bit count field.
%
%
However, because each SALU can access up to a pair of 32-bit counters, the IDsum field cannot be directly built in Tofino.  
%
%
To address this issue, we divide the IDsum field into multiple counters.
Rather than encoding complete flow IDs, each counter only encodes specific ID fragments. 
Considering that a 32-bit counter can support at most 32-bit primes, and thus can encode at most 31-bit ID fragment, we need four 32-bit counters to simulate the IDsum field.
Specifically, the division of the IDsum field is shown in Figure \ref{ID:division}.
The first three 32-bit counters encode the lower 31-bits of the source IP address, the destination IP address, and the concatenation of the source port and destination port, respectively.
The last one 32-bit counter encodes the rest 11-bit ID fragment (1-bit source IP address + 1-bit destination IP address + 1-bit source port + 8-bit protocol), and the other unused 20 bits are used to encode a fingerprint to improve decoding success rate.
In summary, each bucket of \sketchname{} consists of five 32-bit counters: four counters to encode the IDsum field and the fingerprint field, and a counter to encode the count field.
Considering that there is no dependency between the five counters in any bucket of \sketchname{}, a bucket array of \sketchname{} can be built with five 32-bit counter arrays, each of which is built on a register and accessed by a SALU.
During insertion, for any of the four counter arrays encoding the IDsum field and fingerprint field, the SALU needs to insert the specific ID fragment into its counter through modular addition.
As shown in Algorithm \ref{algo:code:moduloadd}, the SALU simulates the modular addition with logic consisting of a conditional judgement and two branches.
Such logic is naturally supported by SALUs.
For the other counter array encoding the count field, the SALU simply increments its counter by one.
In this way, the SALUs simulate the behavior of \sketchname{}.
By duplicating these five registers and SALUs $d$ times, we can easily build a $d$-array \sketchname{}.
Note that we use registers consisting of 32-bit counters, but not registers consisting of pairs of 32-bit counters that can further save SALU resources, to simulate the buckets of \sketchname{}.
This is because the logic used to simulate the modular addition requires two 32-bit metadata ($f$ and $inv$) as input, which is just the maximum number that a SALU can support.
However, encoding two ID fragments with a SALU requires four 32-bit metadata as input, which is beyond the capabilities of SALU.

\bbb{Upstream flow encoder:}
Fifth, unless flow $f$ is a non-sampled LL candidate, the packet is inserted into the upstream flow encoder, which is deployed at stage 8-11 in ingress.
The upstream flow encoder consists of three bucket arrays for the highest memory efficiency.
Each bucket array is built as described above, and consists of $m_{uf}$ buckets.
The left $m_{ll}$ buckets, the right $m_{hh}$ buckets, and the middle $m_{hl}$ buckets in each array form the upstream LL encoder, HH encoder, and HL encoder, respectively.
Similarly, we simulate the two upstream flow encoders by doubling the number of buckets in each array.
Based on the hierarchy of flow $f$, we can easily determine the encoder that the packet should be inserted into.
We denote the number of buckets of a bucket array of that encoder by $m^{\prime}$.
Different from the flow classifier, the base indexes calculate by hash functions cannot be directly used to locate the relative positions of the mapped buckets in the encoder.
This is because a base index is uniformly distributed on [$0$,$2^t-1$], while $m^{\prime}$, which could be any of $m_{ll}$, $m_{hl}$, and $m_{hh}$, may not be powers of two, as they are required to vary for supporting dynamic memory allocation. 
To address this issue, we decide to use the results of base indexes modulo $m^{\prime}$ as the relative positions of the mapped buckets.
To simulate modulo operation in data plane, we input the hierarchy of flow $f$ and a base index $h_b$ to a match-action table that uses exact matching on flow hierarchy and range matching on index.
The table first determines $m^{\prime}$ based on the input flow hierarchy, then outputs the largest number that is divisible by $m^{\prime}$ and less than $h_b$, and finally subtracts that number from $h_b$.
Obviously, the result equals to $h_b$ modulo $m^{\prime}$.
In this way, we locate the relative positions of the mapped buckets at the cost of TCAM resources, and can finally locate the mapped buckets by adding offsets corresponding to the 1-bit timestamp and the flow hierarchy to the relative positions.
%
%
Considering that the width of base index is fixed at run-time, if its width is too long compared to the width of $m^{\prime}$, the match-action table will need a lot of entries to support range matching, and thus consumes lots of TCAM resources; if its width is just a bit longer than the width of $m^{\prime}$, the uniformity of the calculated relative positions will be quite poor, leading to reduction of the decoding success rate of \sketchname{}.
To address this issue, before we input the base index to the match-action table, we bitwise-AND the base index with a mask to guarantee that the value range of the index is between $4m^{\prime}$ and $8m^{\prime}$, so as to make great trade-off between the uniformity of relative positions and the consumption of TCAM resources.
Note that due to the inherent features of TCAM, when TCAM is used for range matching, different value range would require different number of TCAM entries for supporting modulo operation.

\bbb{Downstream flow encoder:}
Sixth, unless flow $f$ is a non-sampled LL candidate, the packet is inserted into the upstream flow encoder, which is deployed at stage 4-7 in egress.
The implementation of downstream flow encoder is almost the same as that of upstream flow encoder, except it omits the heavy-hitter encoder.
Note that the flow hierarchy and 1-bit timestamp are obtained from the edge switch where the packet enters the network, and carried by recording them in three bits of the ToS field of the IPv4 protocol.

\bbb{Resources Usage:}
As shown in Table \ref{tab:overhead}, under the parameter settings in Section \ref{eval:testbed}, the \systemname{} data plane consumes SALUs most, achieving 66.7\%. 
This is because the flow classifier, the upstream flow encoder, and the downstream flow encoder all need SALUs to access memory. 
For resources other than SALUs, \systemname{} consumes no more than 25\%.
Overall, the resource usage of \systemname{} is moderate.
Although \systemname{} indeed consumes a lot of SALUs, the consumption of SALUs will not increase when we further enlarge the above sketches.
With the advent of Tofino 2 switches and even Tofino 3 switches, we believe the resource usage will be much more acceptable on these more advanced programmable switches.

\begin{table} [h!]
\setlength{\abovecaptionskip}{-0.0cm}
\setlength{\belowcaptionskip}{-0.0cm}
\centering
{\centering
\small
\begin{tabular}{|l|r| r|c|}
\hline
\textbf{Resource} & \textbf{Usage} & \textbf{Percentage}\\ \hline
Exact Match Input xbar & 353 & 22.98\% \\ 
Ternary Match Input xbar & 33 & 4.17\% \\ 
VLIW Instructions & 43 & 11.20\% \\ 
Map RAM & 102 & 17.71\% \\
SRAM & 130 & 13.54\% \\ 
TCAM & 8 & 2.78\% \\  
Hash Bits & 809 & 16.21\% \\ 
Stateful ALU & 32 & 66.67\% \\ \hline 
\end{tabular}
}
\caption{Resources used by \systemname{} in Tofino.
}
\label{tab:overhead}
\end{table}

\presec
\subsection{Control Plane Implementation}
\label{sec:cpi}
\bbb{Central controller:}
The central controller integrates three modules into a DPDK \cite{dpdk} program: 1) a packet receiver module responsible for collecting sketches; 2) an analyzer module for decoding sketches, monitoring real-time network state, and generating reconfiguration packets for reconfiguring the \systemname{} data plane; 3) a packet sender module responsible for sending reconfiguration packets to the control plane of each edge switch.

\bbb{Switch control plane:}
The control plane of each edge switch runs a C++ program to load the P4 program to the Tofino ASIC.
Every time the switch control plane receives a reconfiguration packet, it first extracts the packet to obtain the reconfiguration.
Then, based on the reconfiguration, it generates corresponding table entries and update them to the corresponding match-action tables in the data plane to reconfigure the switch data plane.
The time consumption in this step is shown in Figure \ref{fig:testbed:cdf} in Appendix \ref{app:timeoverhead}.
To avoid the updated entries to interfere with the monitoring of the current epoch, those corresponding match-action tables further use exact matching on the 1-bit timestamp.
Those newly updated entries match the 1-bit timestamp in the next epoch, so as to function in the next epoch.

\bbb{Epoch length:}
On our testbed, we set the epoch length to 50ms by default, and the additional time for all traffic passing through the network is set to 10ms (described in Appendix \ref{appendix:collection}).

\bbb{Clock synchronization:} 
On our testbed, we use the well-known software time synchronization protocol NTP \cite{mills1991internet} to synchronize the clocks between the control planes of edge switches and the central controller.
Every 10s, every edge switch synchronizes its clock with the central controller.
The precision of synchronization is around 0.3ms$\sim$0.5ms, and thus NTP can already satisfy the precision requirement for epochs of 50ms.
We can also use more advanced solutions \cite{geng2018exploiting, kannan2019precise} for us-level or even ns-level precision.



\bbb{Data plane collection:}
To collect sketches from data planes of edge switches, a naive solution is to directly use the C++ control plane APIs provided by the Tofino SDK \cite{SDK}.
Currently, the most efficient strategy for this solution is to first use bulk DMA transfer to read data plane counter arrays into control plane buffer, and then read the counter arrays from control plane buffer \cite{namkung2021telemetry}.
However, on our testbed, such strategy takes about 338ms to collect only the upstream flow encoder, which seriously limits the setting of epoch length, and thus degrades the accuracy and timeliness of measurement.
To address this issue, we fully exploit the capabilities and features of programmable data plane, including SALUs, mirror, and recirculate ports.
Specifically, the switch control plane just needs to send several tailored packets to data plane for collecting sketches.
The tailored packet is forwarded to the recirculate port, so as to access the counters of each sketches in turn.
Every time a tailored packet accesses a counter, leveraging the SALU, it reads the value and inserts the value into its payload.
Every time a tailored packet reaches the maximum transmission unit (MTU, \eg, 1514 Bytes), the switch data plane forwards it to the central controller, and mirrors a new truncated packet (\eg, 64 Bytes) to read the remaining counters.
In this way, collecting the upstream flow encoder from the switch data plane only takes 0.44ms, which is 775 times faster than the straightforward solution.
%
%
%
%
To ensure that the tailored packets will not be lost during the transmission, we reserve idle ports in their forwarding paths.
Overall, the central controller takes 11.33ms to collect sketches from the data plane each edge switch, which consists five parts: 1) every time the timestamp flips, the central controller first sleeps 1ms to eliminate the impact caused by the error in clock synchronization, ensuring that all the edge switches have started the current epoch; 2) the central controller takes 2.68ms to collect the flow classifier; 3) the central controller takes 0.44ms to collect the upstream flow encoder; 4) the central controller sleeps 6.88ms to ensure that all the packets in the previous epoch have passed through or lost in the network; 5) the central controller takes 0.33ms to collect the downstream flow encoder.

\section{Evaluation on Different Workloads}
\label{sec:workload}

In this section, we show that on workloads other than DCTCP, how \systemname{} shifts measurement attention with the change of the number of flows or ratio of victim flows.
For the measurement attention under different number of flows, we vary the number of flows in the network from $10$K to $100$K, and fix the ratio of victim flows to $10\%$. 
For the measurement attention under different ratios of victim flows, we vary the ratio of victim flows from $2.5\%$ to $25\%$, and fix the number of flows to $50$K.

\subsection{CACHE Workload}

\bbb{Measurement attention \textit{vs.} number of flows (Figure~\ref{fig:cache:flownum}):}
As the number of flows increases from $10$K to $20$K, \systemname{} can record all flows and victim flows, and therefore sets both $T_h$ and $T_l$ to $1$.
As the number of flows increases from $30$K to $70$K, \systemname{} allocates more memory to HL encoders and raises $T_h$ higher than $1$.
As the number of flows increases from $80$K to $100$K, the healthy network state transitions to the ill network state.
\systemname{} allocates memory to LL encoders, increases $T_l$ and decreases the sample rate, so as to control the number of HLs and sampled LLs.
Meanwhile, \systemname{} raises $T_h$ to control the number of HH candidates.
The relatively low load factor when the number of flows is between $80$K and $100$K is because of the high skewness of CACHE workload: 
lower thresholds will lead to a huge increase of the number of recorded flows, thus causing decoding failure.
In fact, when the number of flows is between $80$K and $100$K, the $T_h$ is set to $3$, and the $T_l$ is set to $2$.
\systemname{} has tried its best to select thresholds to maximizes the load factor.

\bbb{Measurement attention \textit{vs.} ratio of victim flows (Figure~\ref{fig:cache:lossrate}):}
As the ratio of victim flows increases from $2.5\%$ to $12.5\%$, \systemname{} records all victim flows by allocating more and more memory to HL encoders.
$T_h$ is not adjusted because of the high skewness of CACHE workload: setting $T_h$ to $2$ already makes a fairly small portion of flows as HH candidates, and lower $T_h$ leads to decoding failure.
As the ratio of victim flows increases from $15\%$ to $25\%$, the healthy network state transitions to the ill network state.
\systemname{} allocates memory to LL encoders, increases $T_l$ to $2$ and decreases the sample rate so as to  control the number of HLs and sampled LLs.
Meanwhile, because the memory of upstream heavy-hitter encoder and the number of flows remain unchanged, $T_h$ also remains unchanged.
%
The reason why \systemname{} suffers low load factor when the ratio of victim flows is between $15\%$ to $25\%$ is also due to high skewness of CACHE workload.
Both $T_h$ and $T_l$ are set to $2$, and decrease of thresholds will lead to decoding failure.
\systemname{} has tried its best to select thresholds to maximize the load factor.

\begin{figure*}[ht!]
    \centering
    \subfigure[Memory division.]{
    \includegraphics[width=0.23\textwidth]{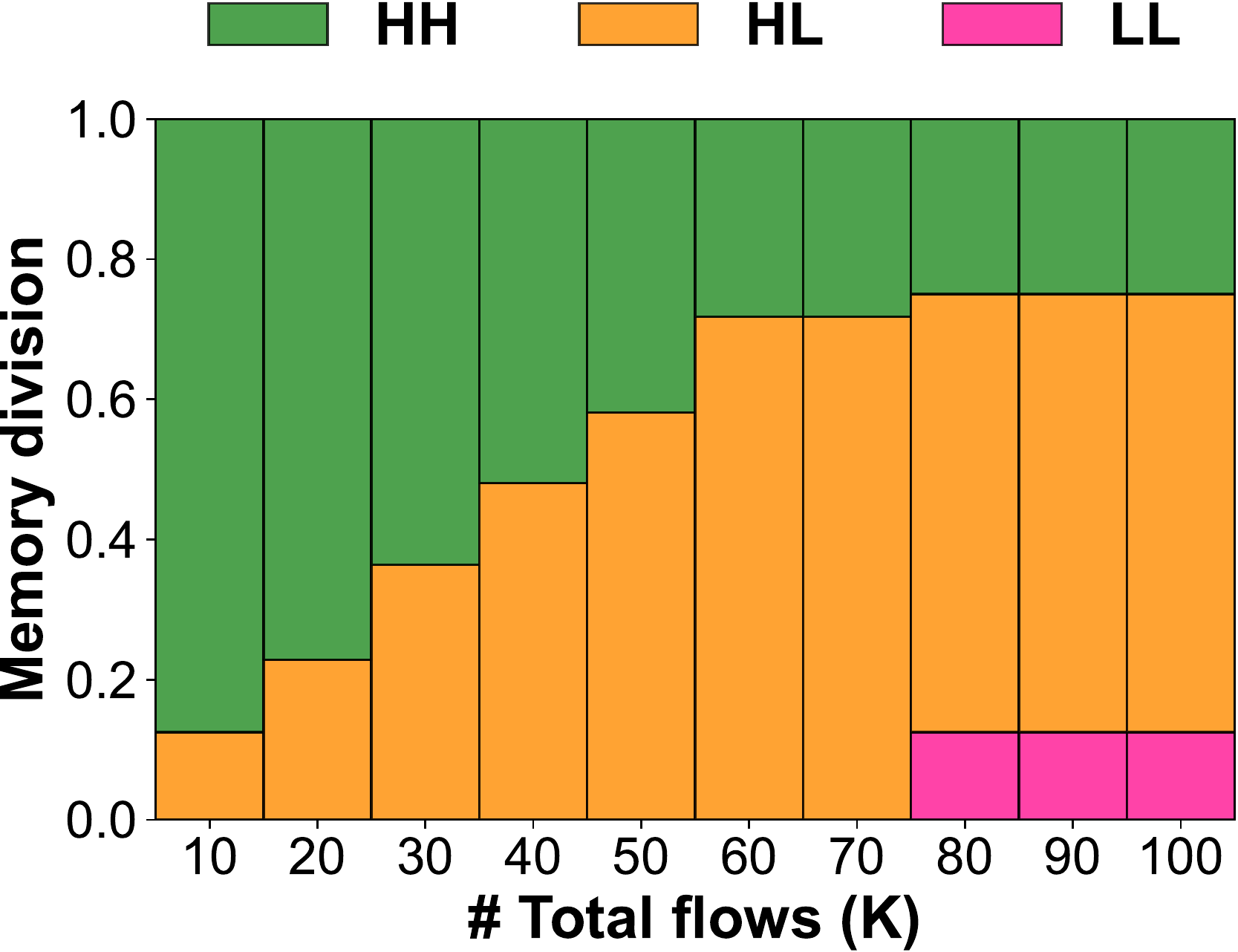}
    \label{fig:cache:flownum:mem}
    }
    \subfigure[Number of decoded flows.]{
    \includegraphics[width=0.23\textwidth]{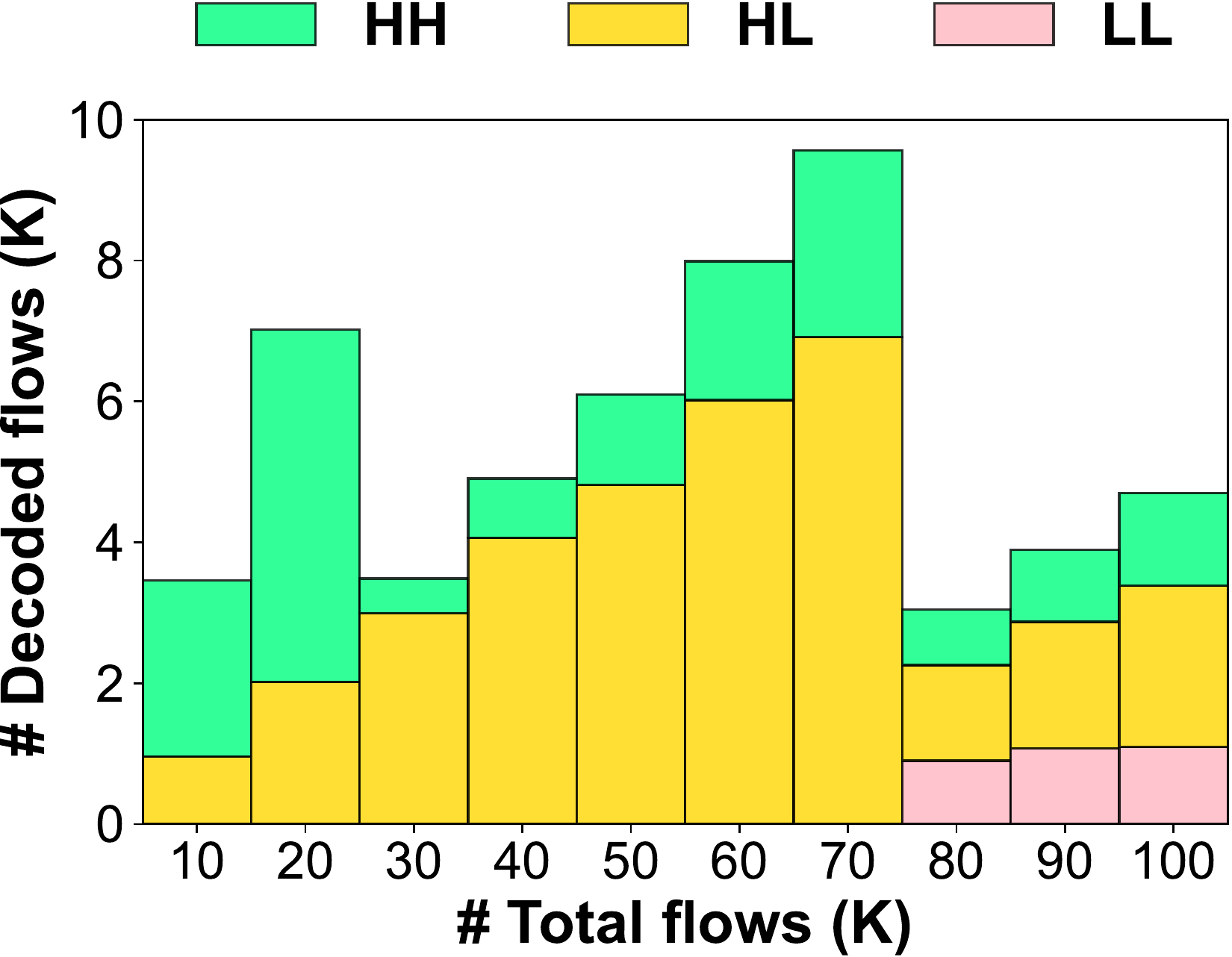}
    \label{fig:cache:flownum:num}
    }
    \subfigure[Threshold.]{ \includegraphics[width=0.23\textwidth]{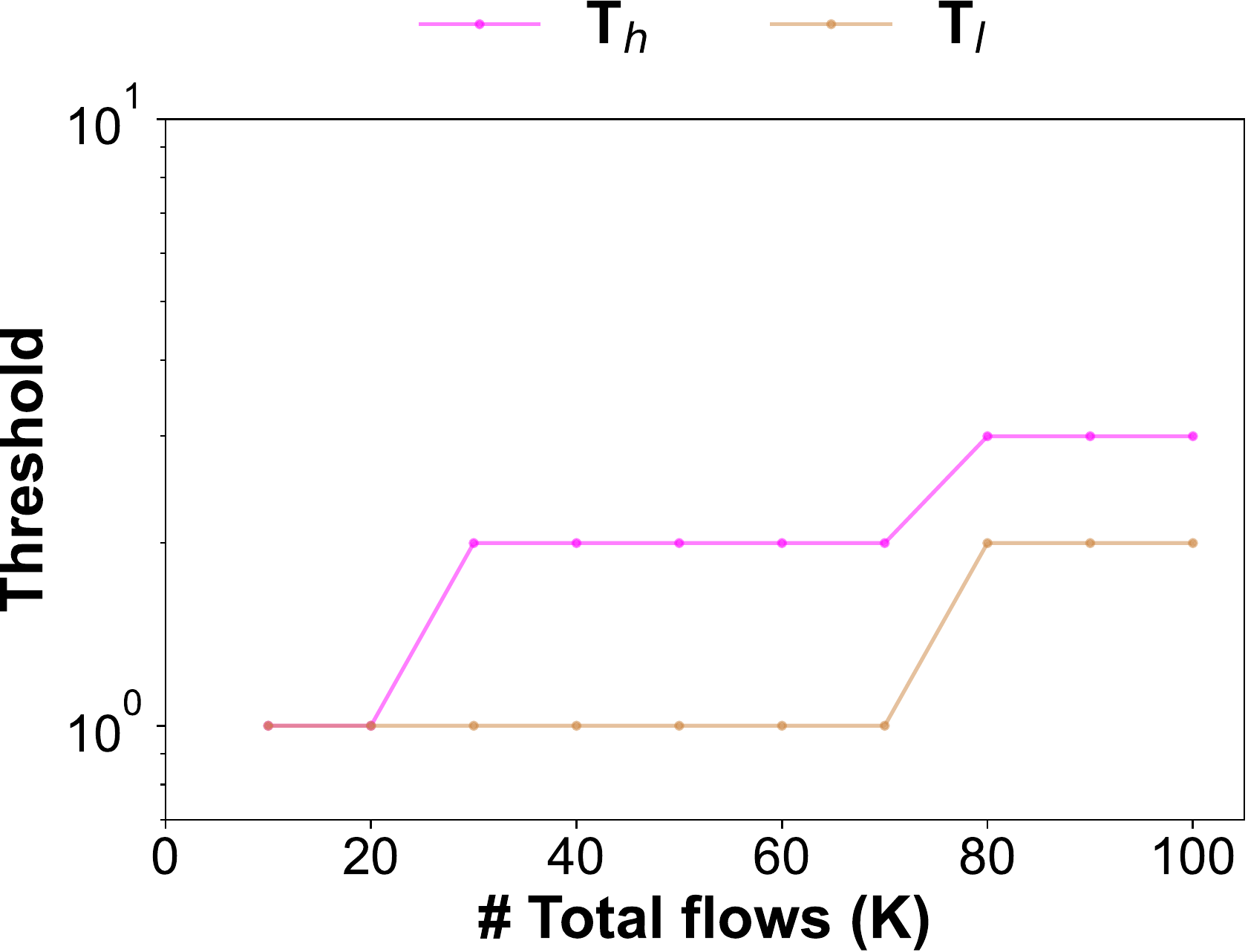}
    \label{fig:cache:flownum:thresh}
    }
    \subfigure[Sample rate.]{ \includegraphics[width=0.23\textwidth]{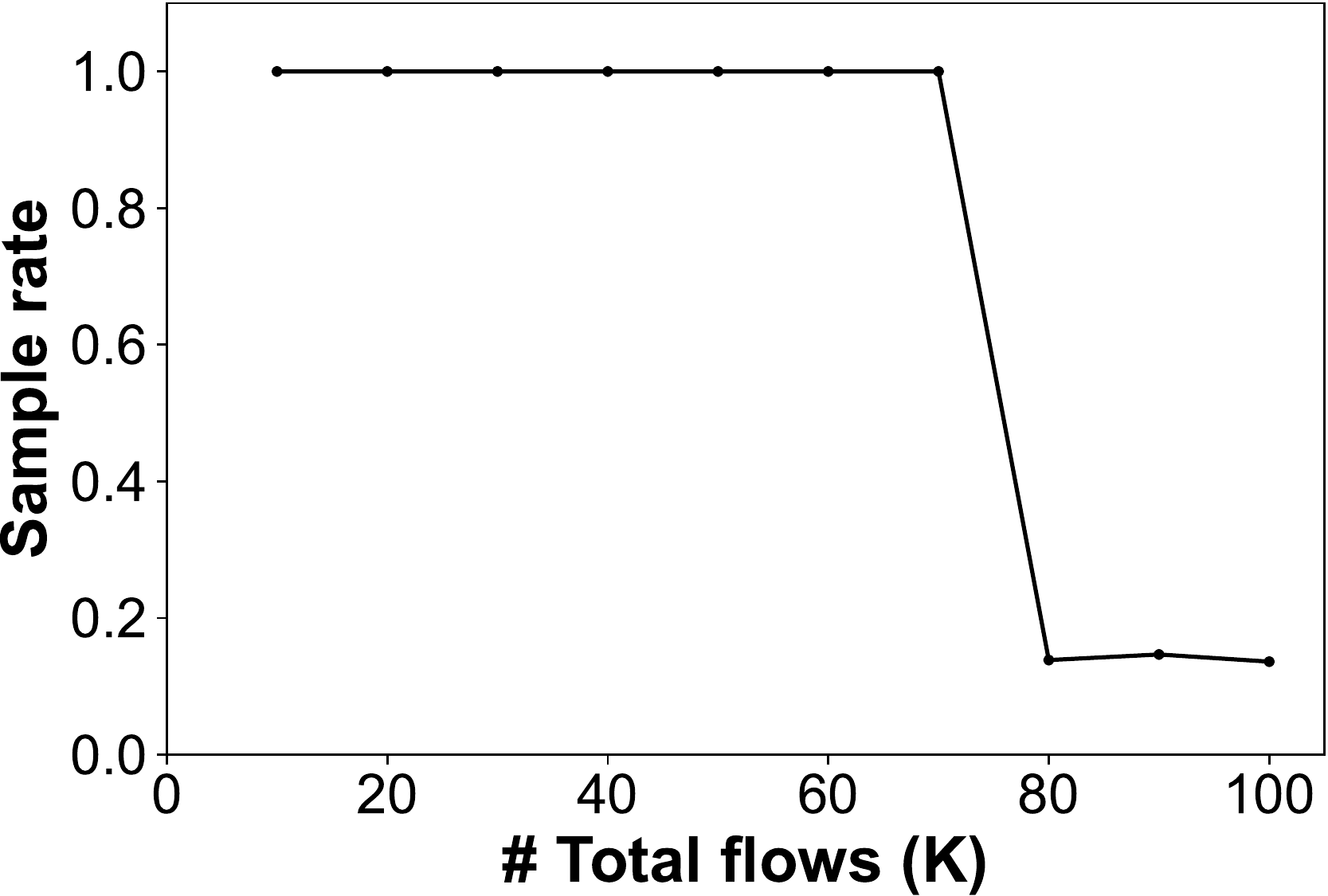}
    \label{fig:cache:flownum:sample}
    }
\caption{Measurement attention \textit{vs.} number of flows on CACHE workload.}
\label{fig:cache:flownum}
\end{figure*}

\begin{figure*}[ht!]
    \centering
    \subfigure[Memory division.]{
    \includegraphics[width=0.23\textwidth]{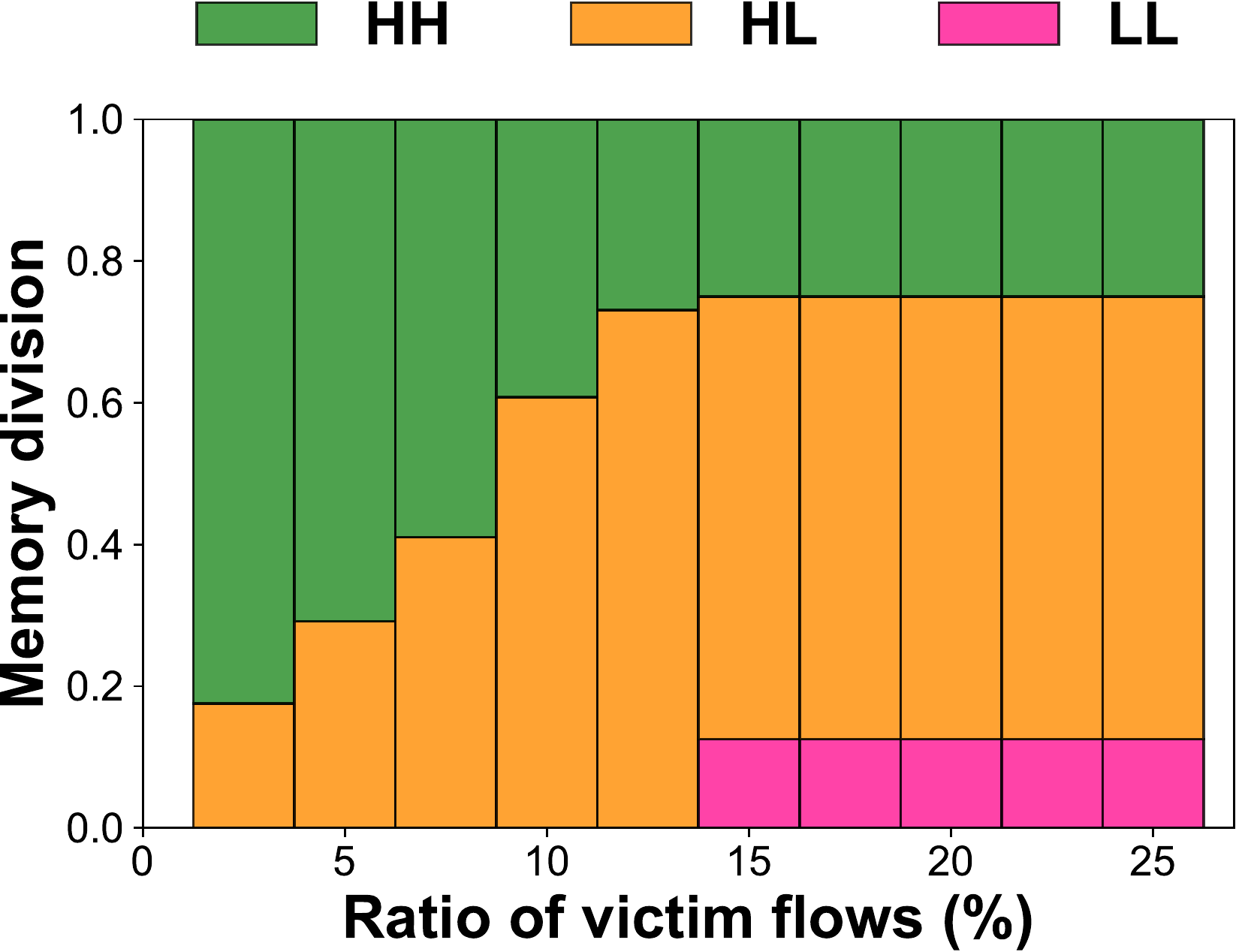}
    \label{fig:cache:lossrate:mem}
    }
    \subfigure[Number of decoded flows.]{
    \includegraphics[width=0.23\textwidth]{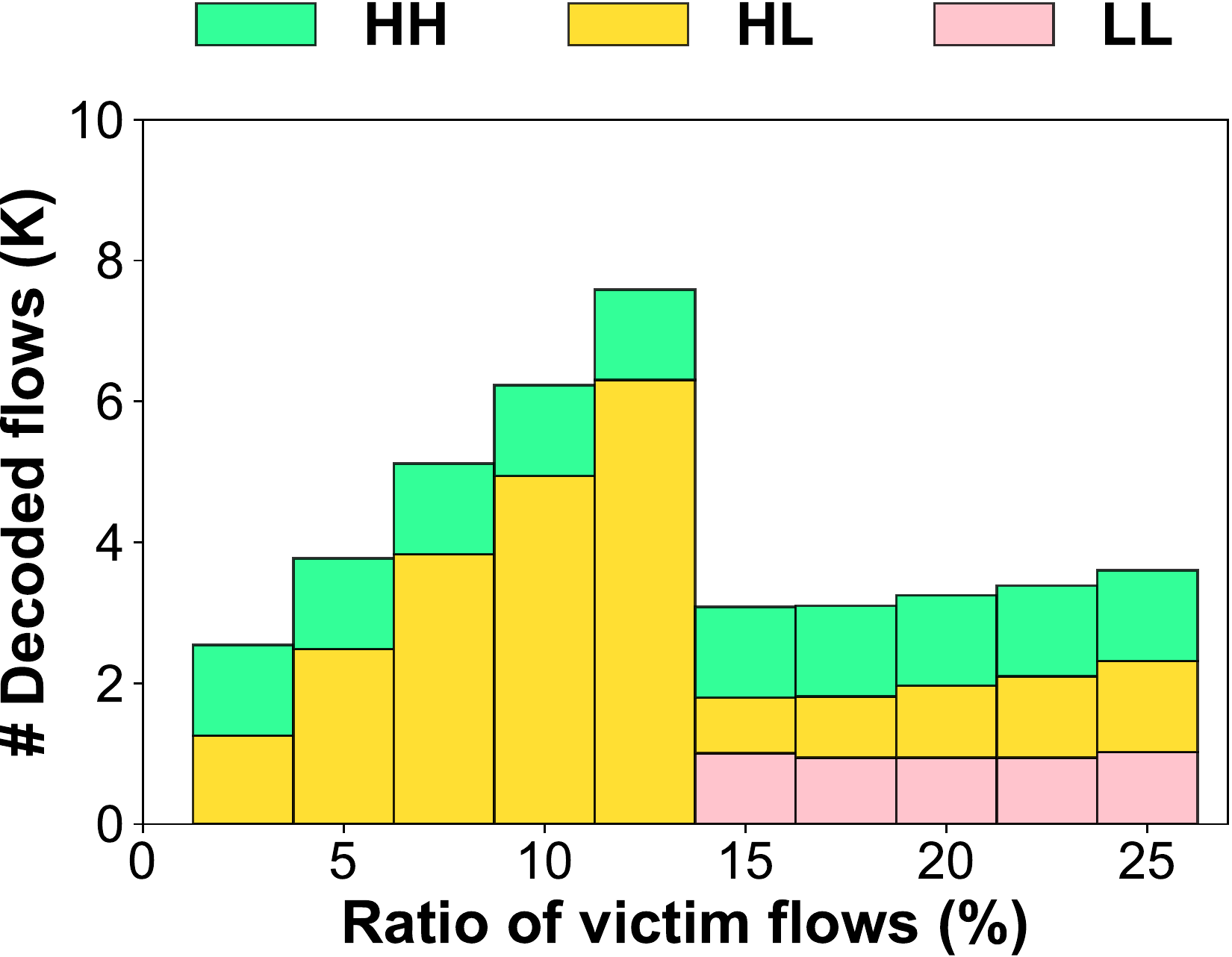}
    \label{fig:cache:lossrate:num}
    }
    \subfigure[Threshold.]{ \includegraphics[width=0.23\textwidth]{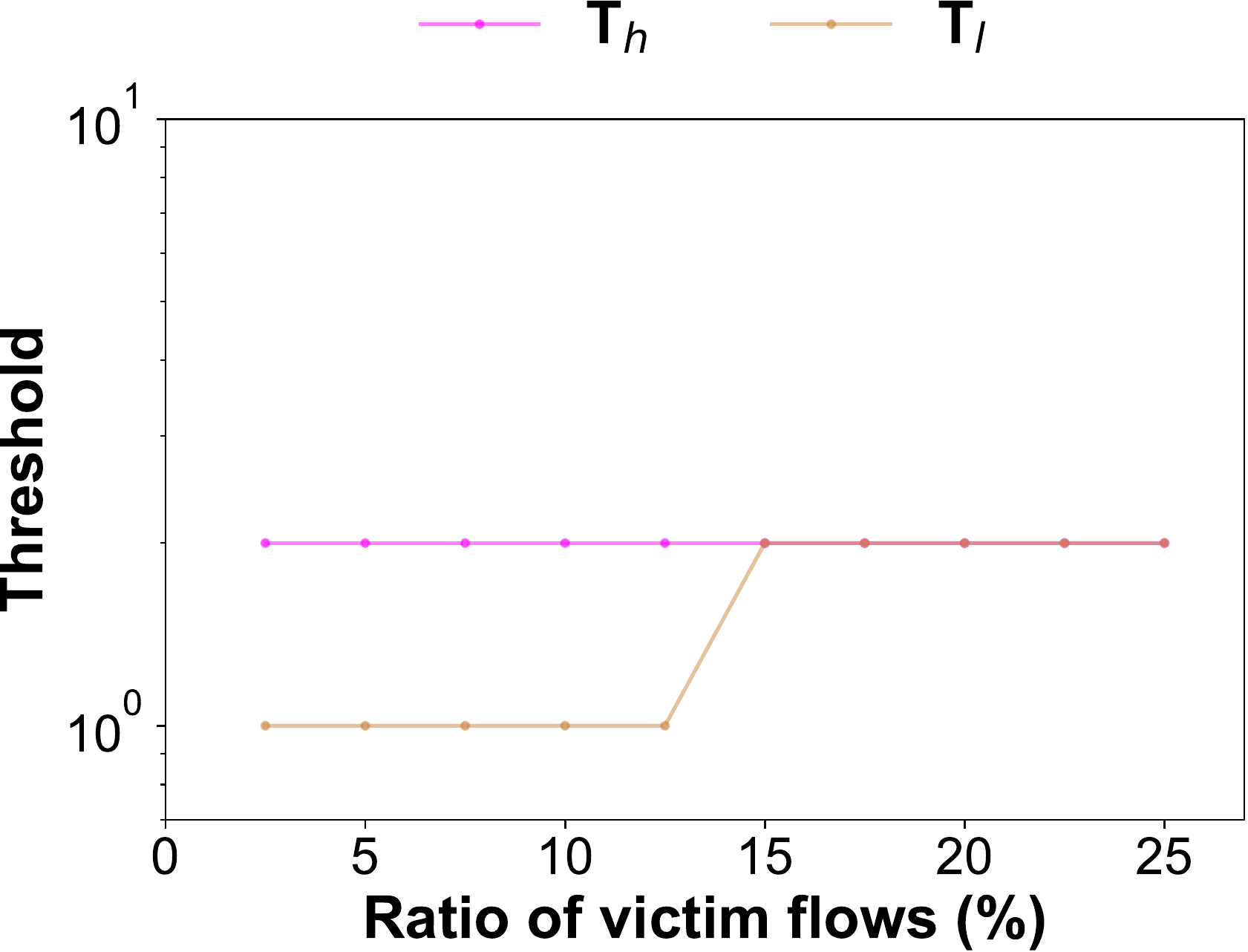}
    \label{fig:cache:lossrate:thresh}
    }
    \subfigure[Sample rate.]{ \includegraphics[width=0.23\textwidth]{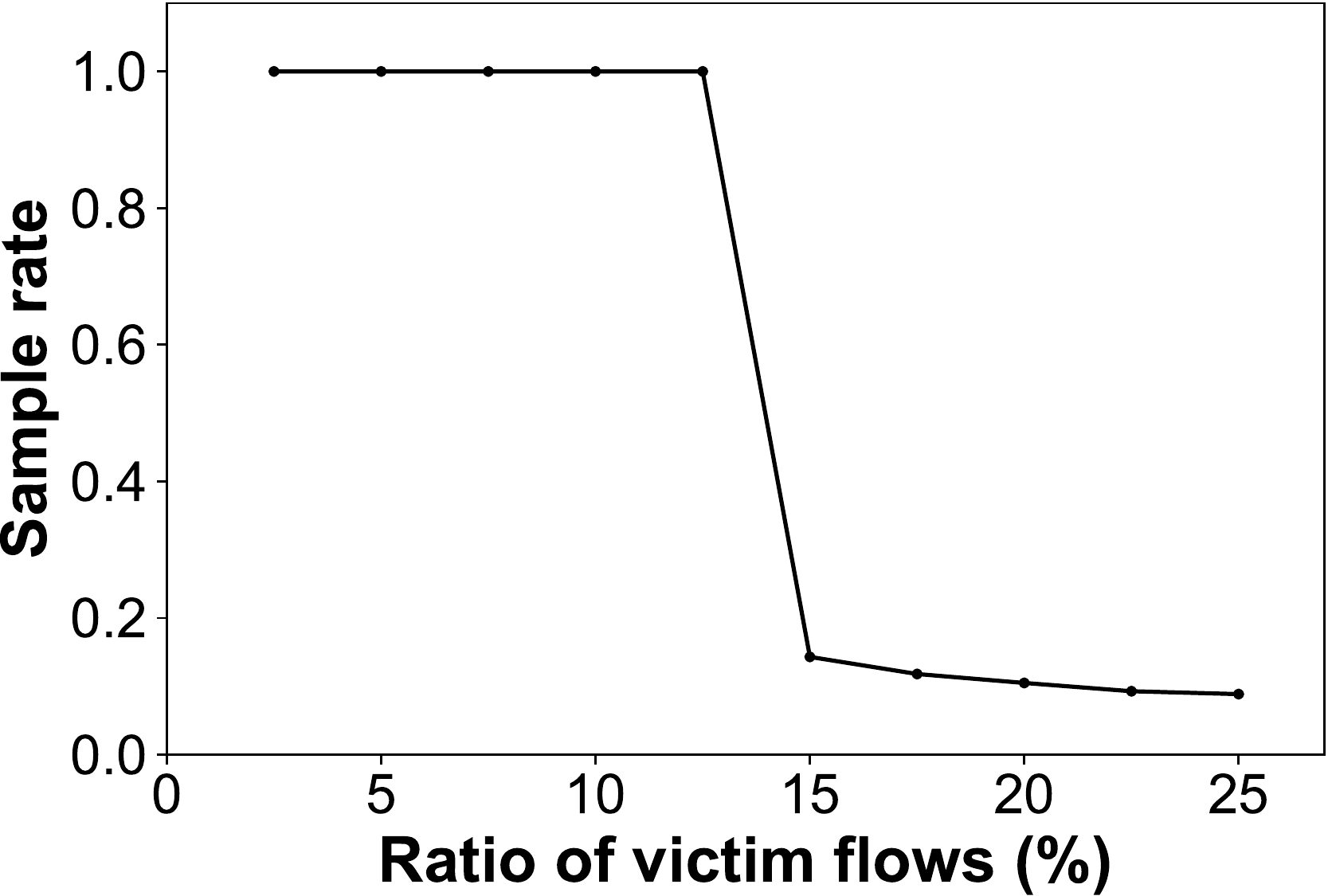}
    \label{fig:cache:lossrate:sample}
    }
\caption{Measurement attention \textit{vs.} ratio of victim flows on CACHE workload.}
\label{fig:cache:lossrate}
\end{figure*}

\begin{figure*}[ht!]
    \centering
    \subfigure[Memory division.]{
    \includegraphics[width=0.23\textwidth]{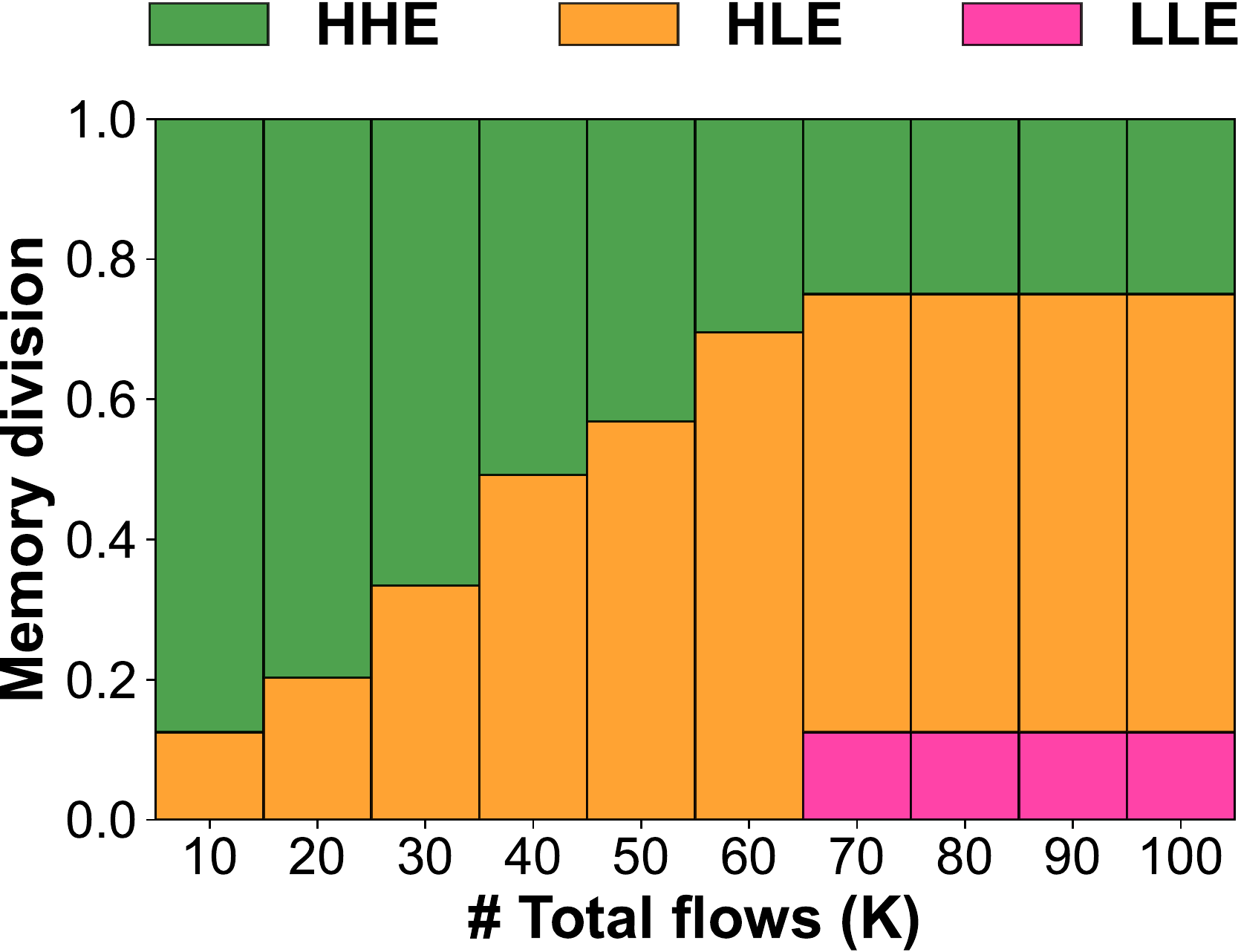}
    \label{fig:VL2:flownum:mem}
    }
    \subfigure[Number of decoded flows.]{
    \includegraphics[width=0.23\textwidth]{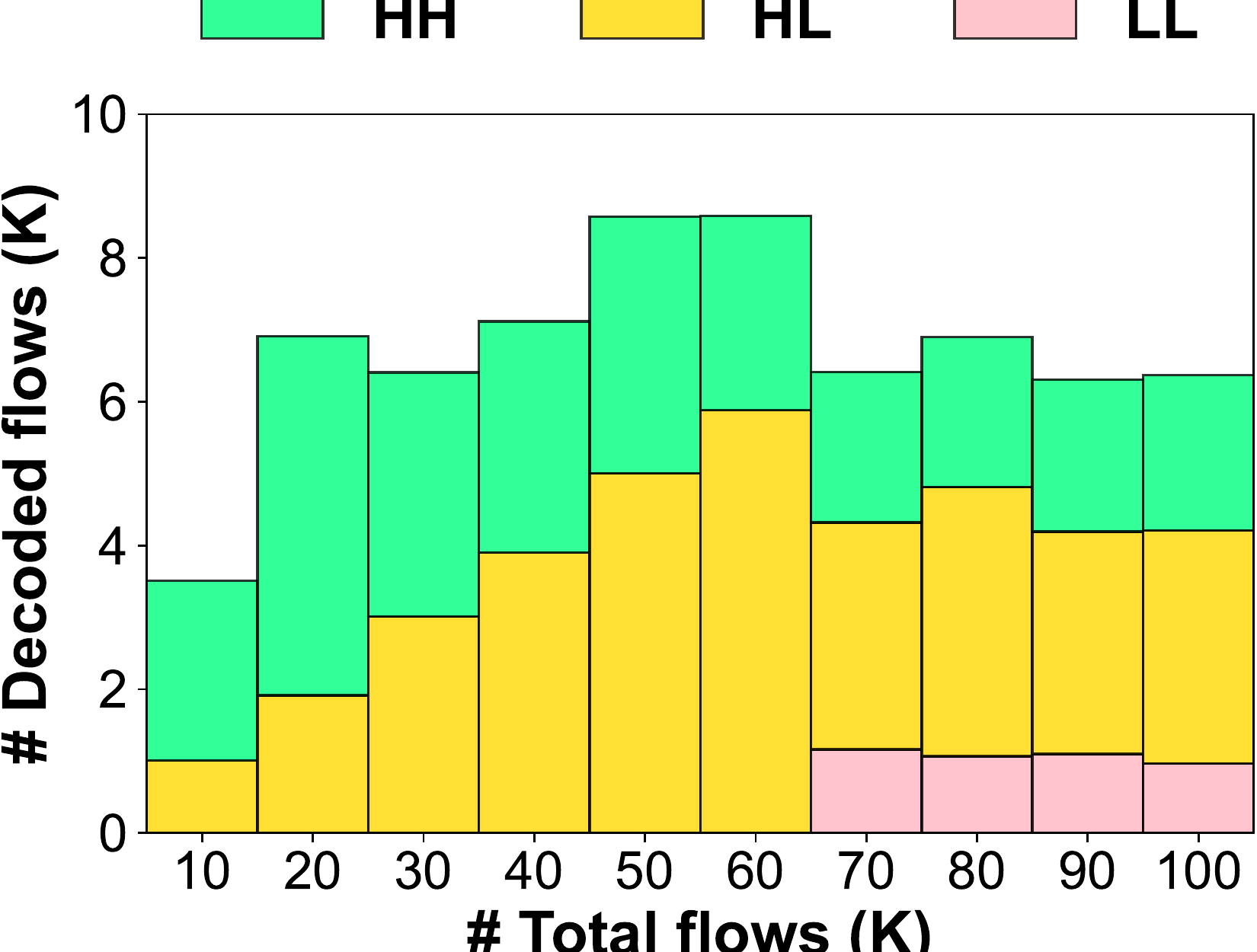}
    \label{fig:VL2:flownum:num}
    }
    \subfigure[Threshold.]{ \includegraphics[width=0.23\textwidth]{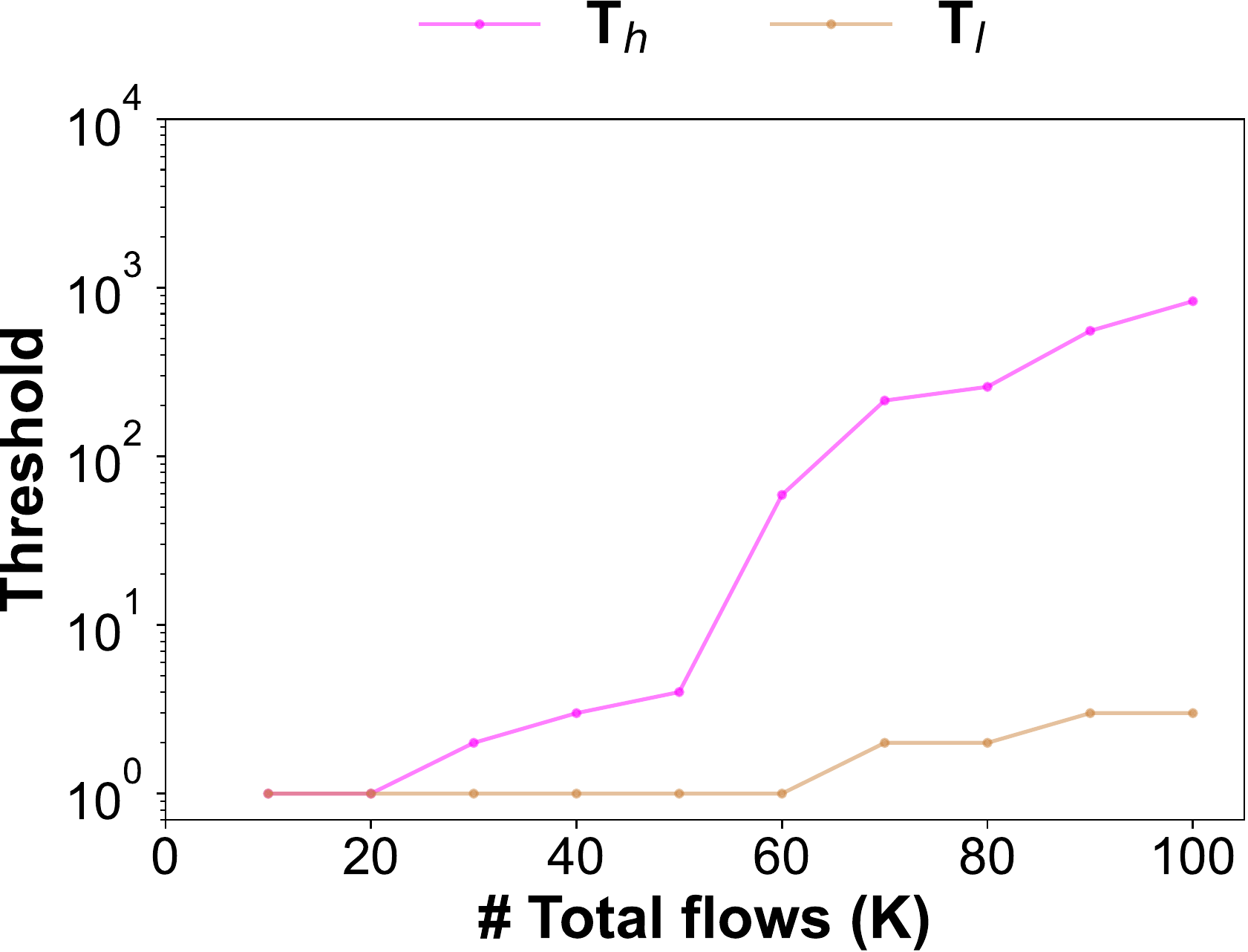}
    \label{fig:VL2:flownum:thresh}
    }
    \subfigure[Sample rate.]{ \includegraphics[width=0.23\textwidth]{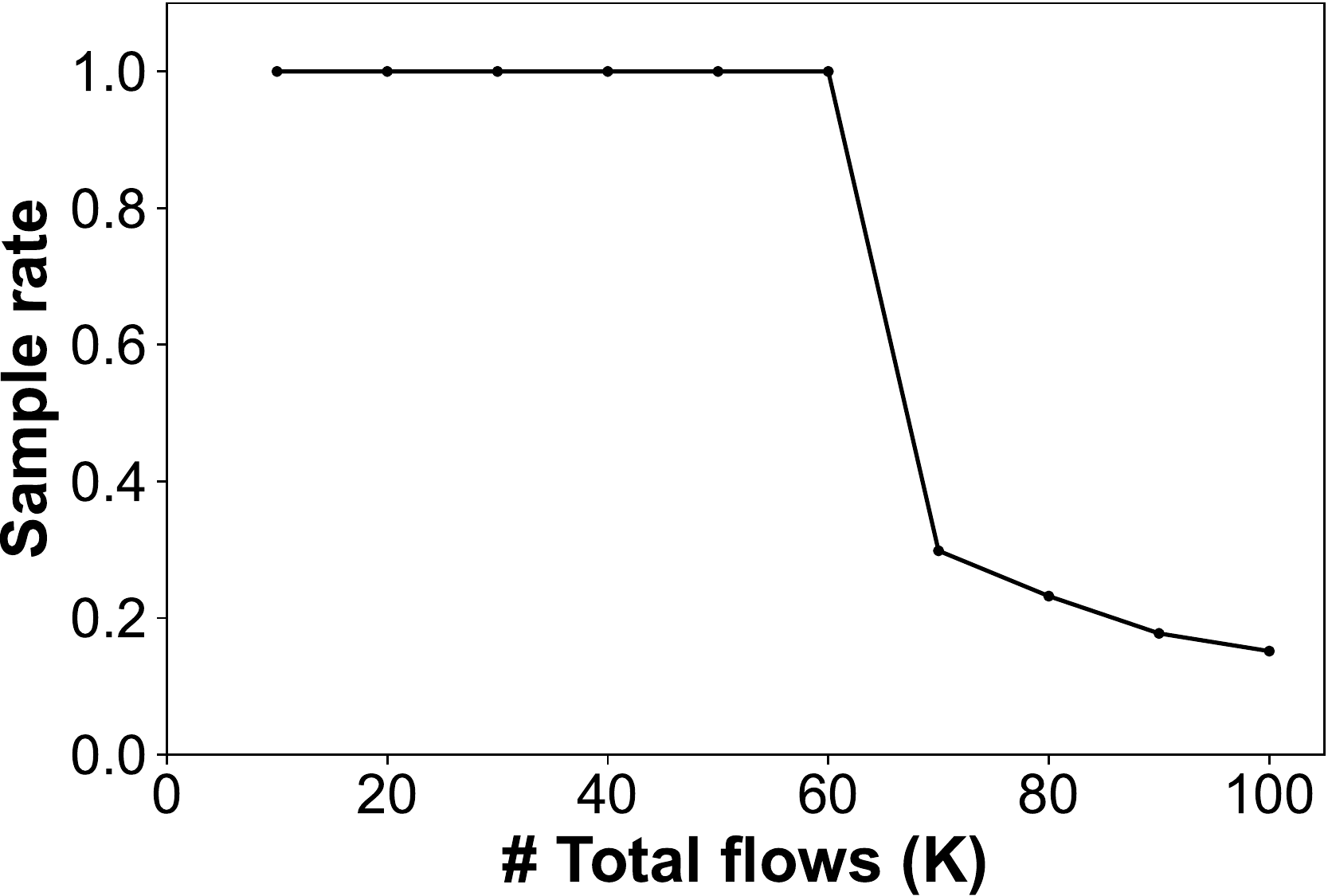}
    \label{fig:VL2:flownum:sample}
    }
\caption{Measurement attention \textit{vs.} number of flows on VL2 workload.}
\label{fig:VL2:flownum}
\end{figure*}

\begin{figure*}[ht!]
    \centering
    \subfigure[Memory division.]{
    \includegraphics[width=0.23\textwidth]{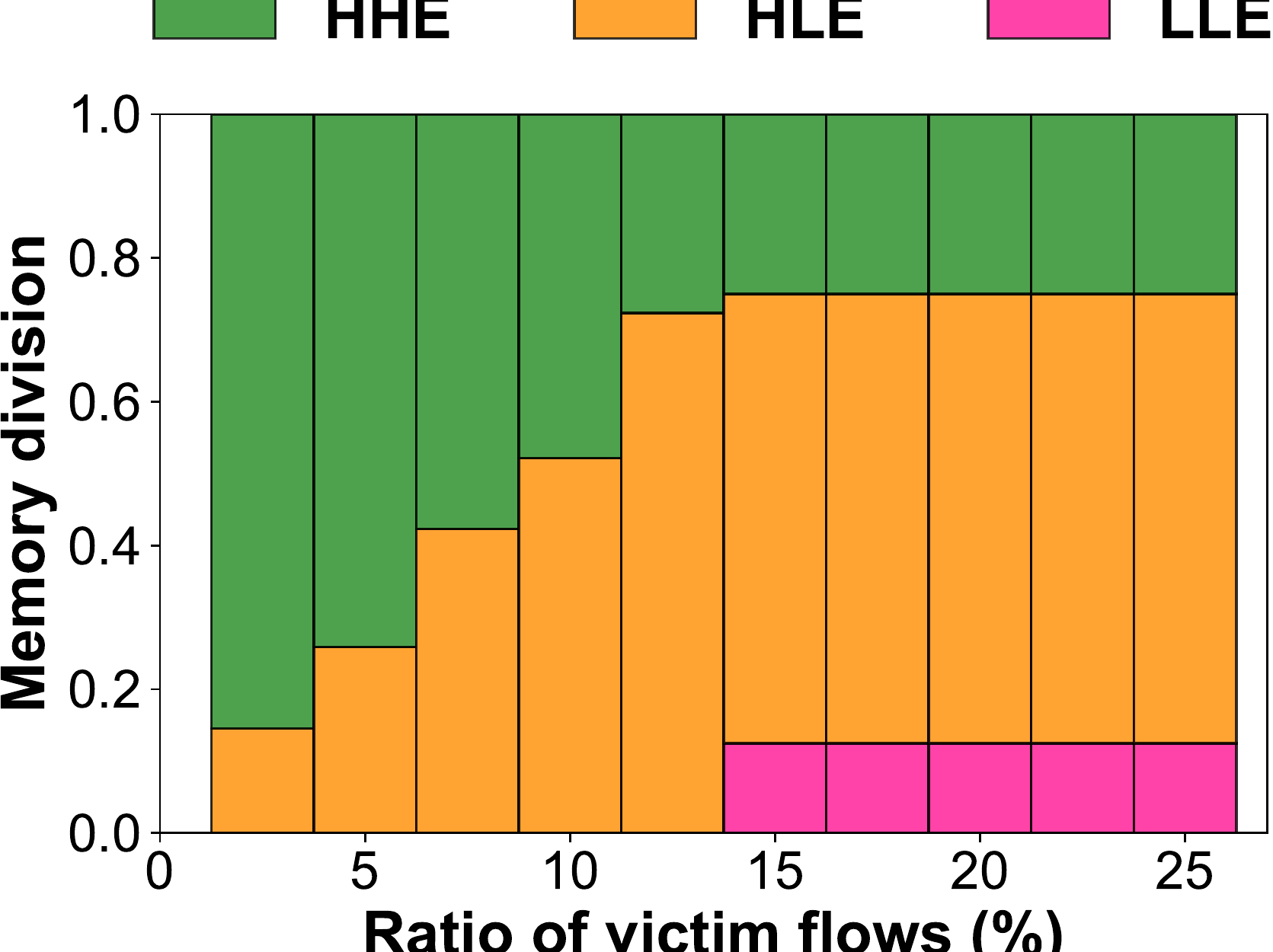}
    \label{fig:VL2:lossrate:mem}
    }
    \subfigure[Number of decoded flows.]{
    \includegraphics[width=0.23\textwidth]{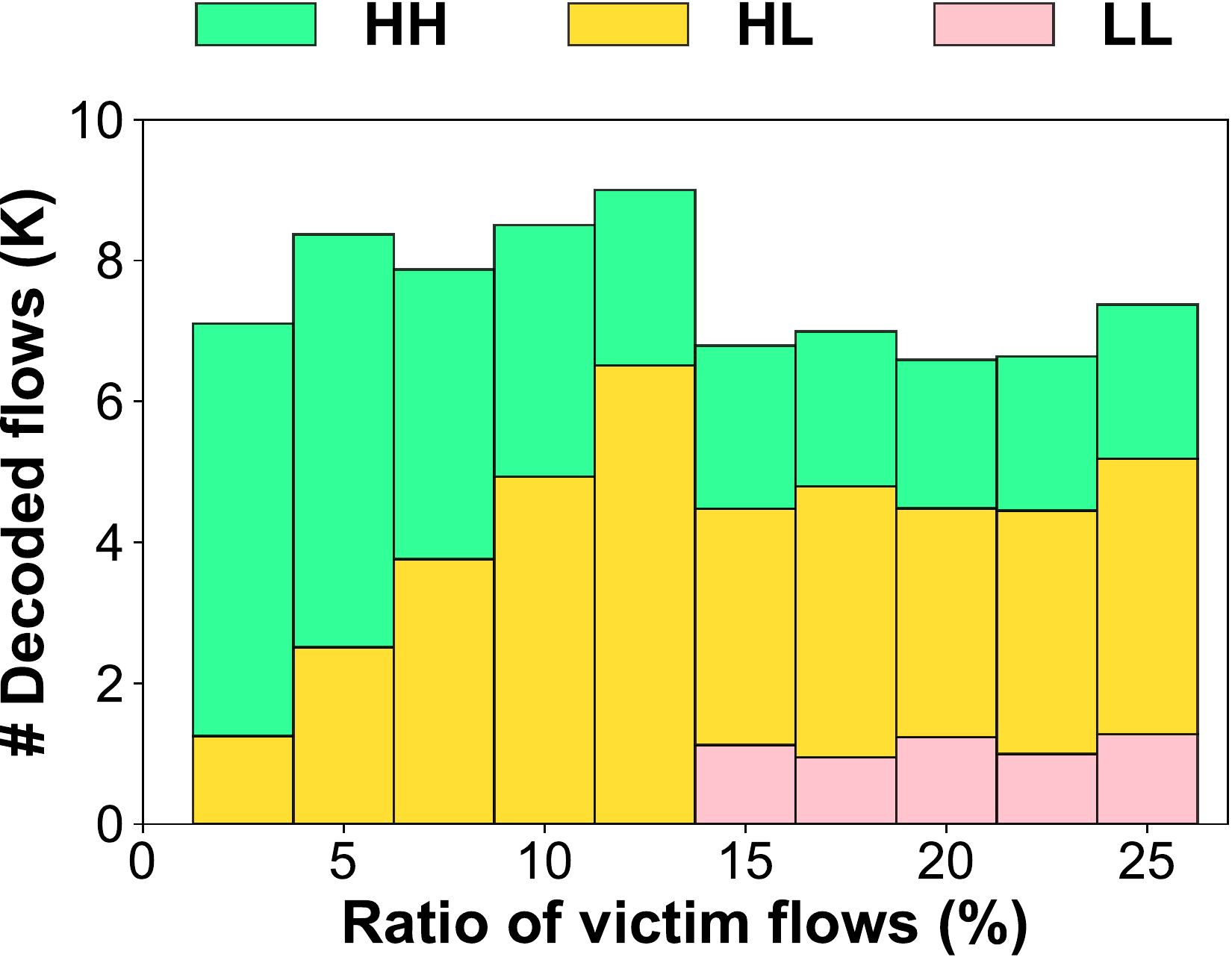}
    \label{fig:VL2:lossrate:num}
    }
    \subfigure[Threshold.]{ \includegraphics[width=0.23\textwidth]{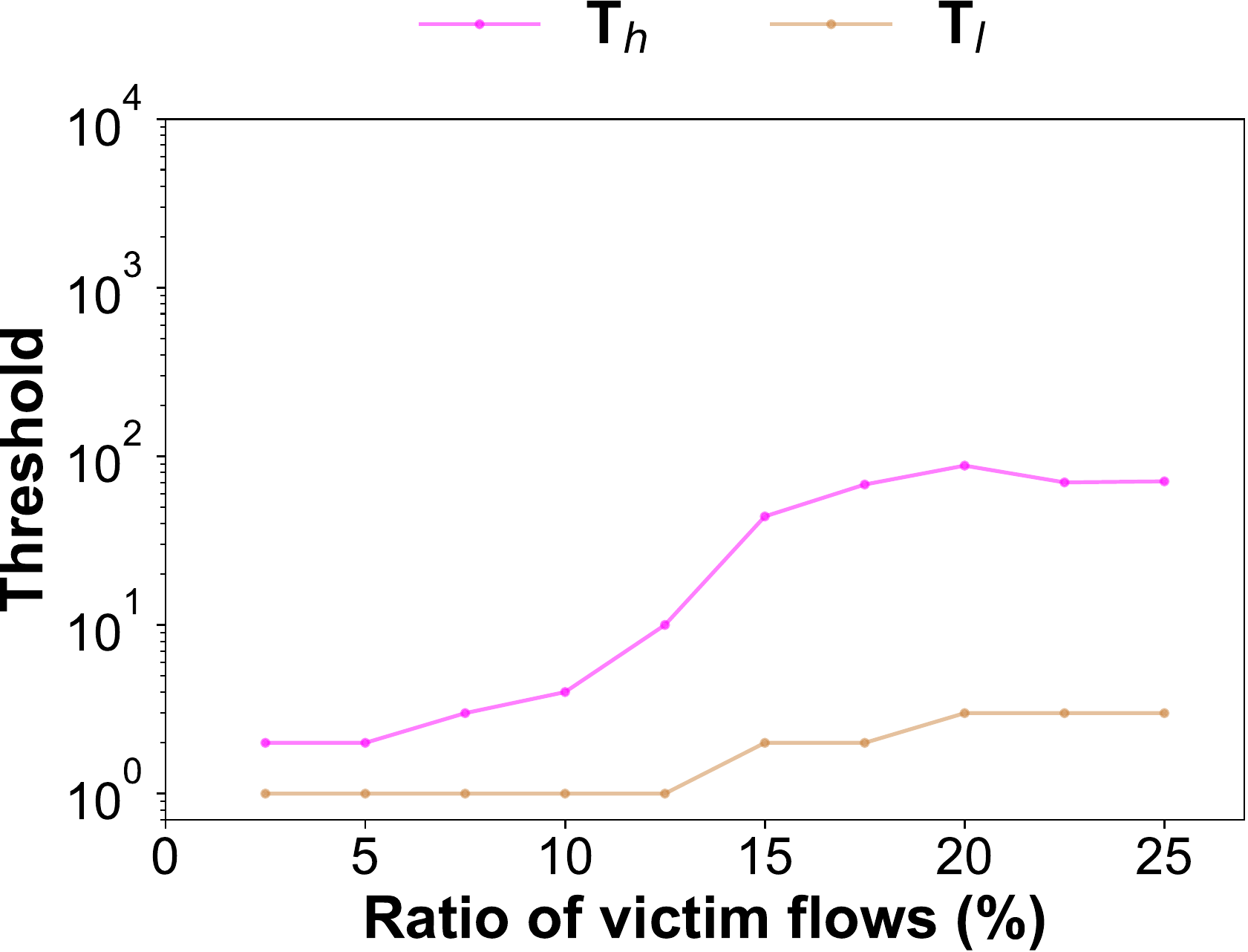}
    \label{fig:VL2:lossrate:thresh}
    }
    \subfigure[Sample rate.]{ \includegraphics[width=0.23\textwidth]{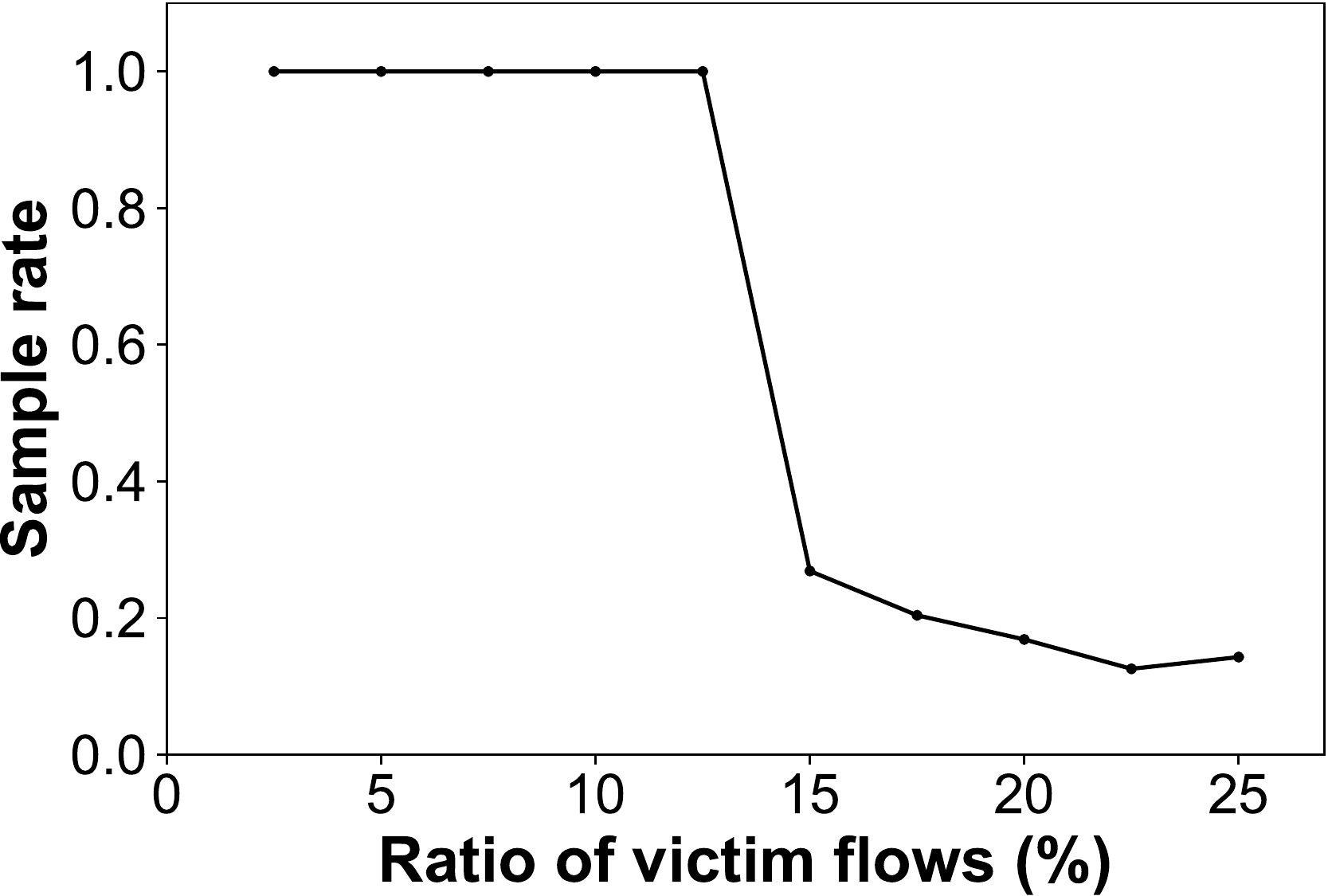}
    \label{fig:VL2:lossrate:sample}
    }
\caption{Measurement attention \textit{vs.} ratio of victim flows on VL2 workload.}
\label{fig:VL2:lossrate}
\end{figure*}

\begin{figure*}[ht!]
    \centering
    \subfigure[Memory division.]{
    \includegraphics[width=0.23\textwidth]{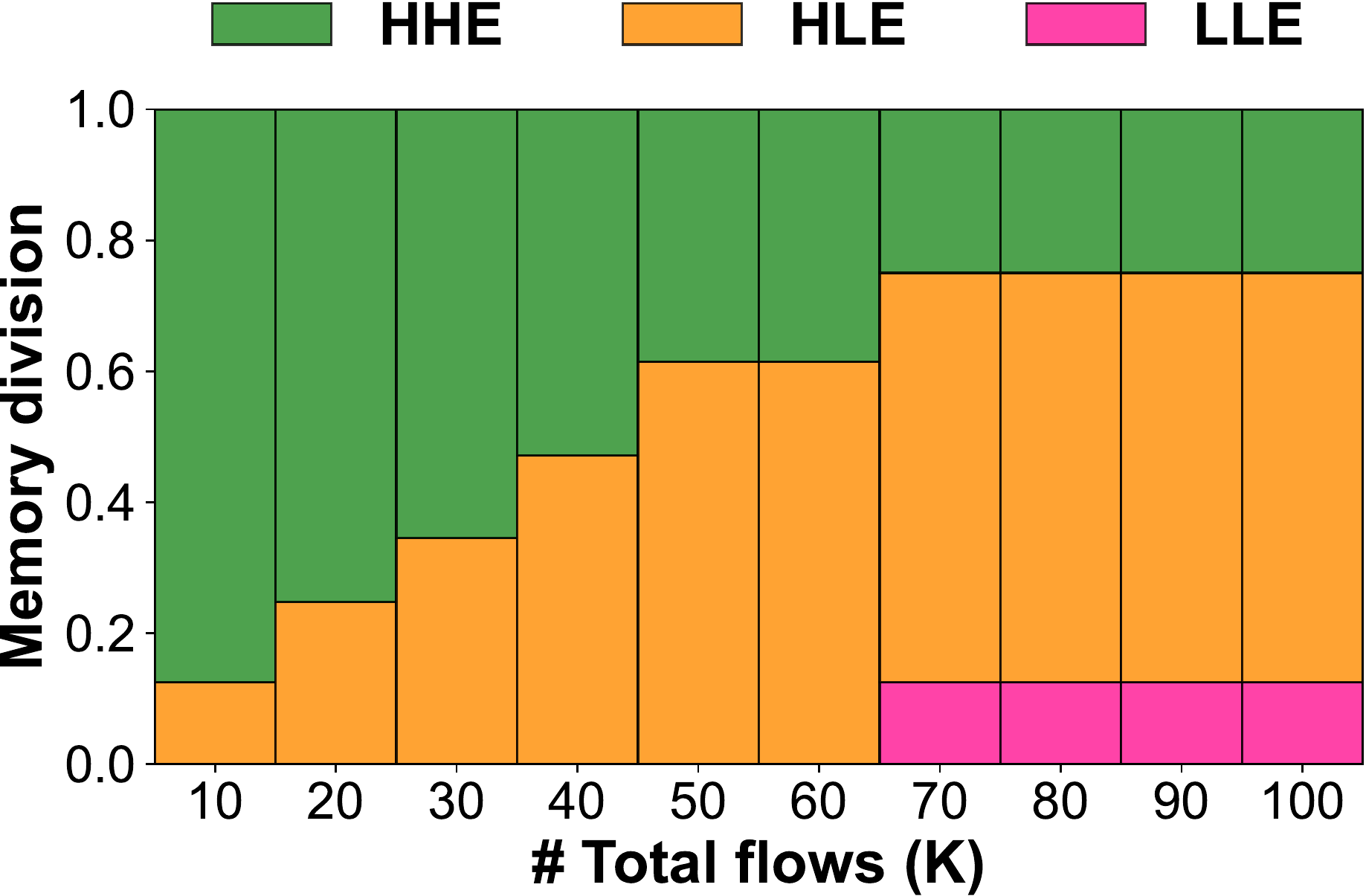}
    \label{fig:HADOOP:flownum:mem}
    }
    \subfigure[Number of decoded flows.]{
    \includegraphics[width=0.23\textwidth]{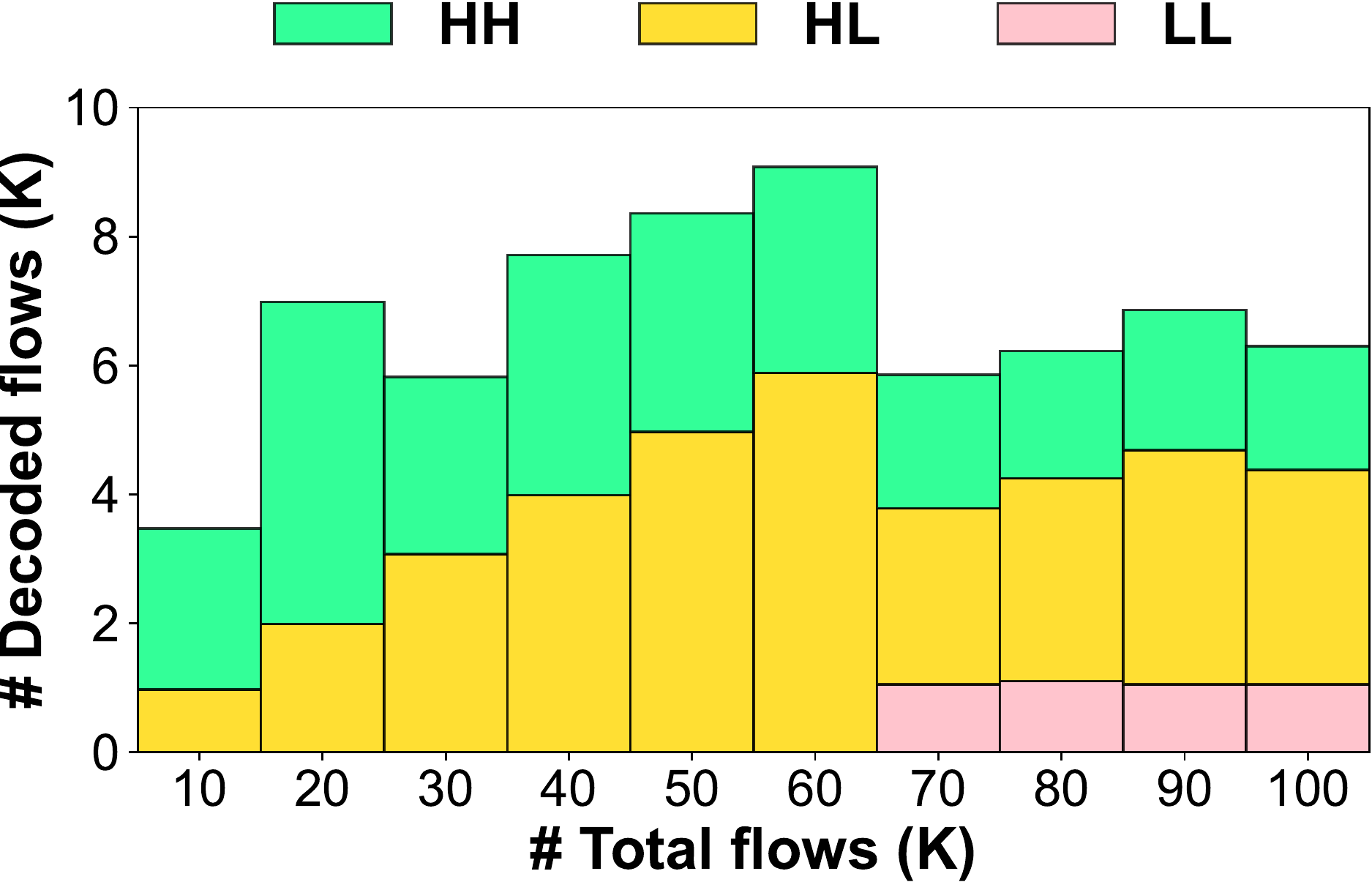}
    \label{fig:HADOOP:flownum:num}
    }
    \subfigure[Threshold.]{ \includegraphics[width=0.23\textwidth]{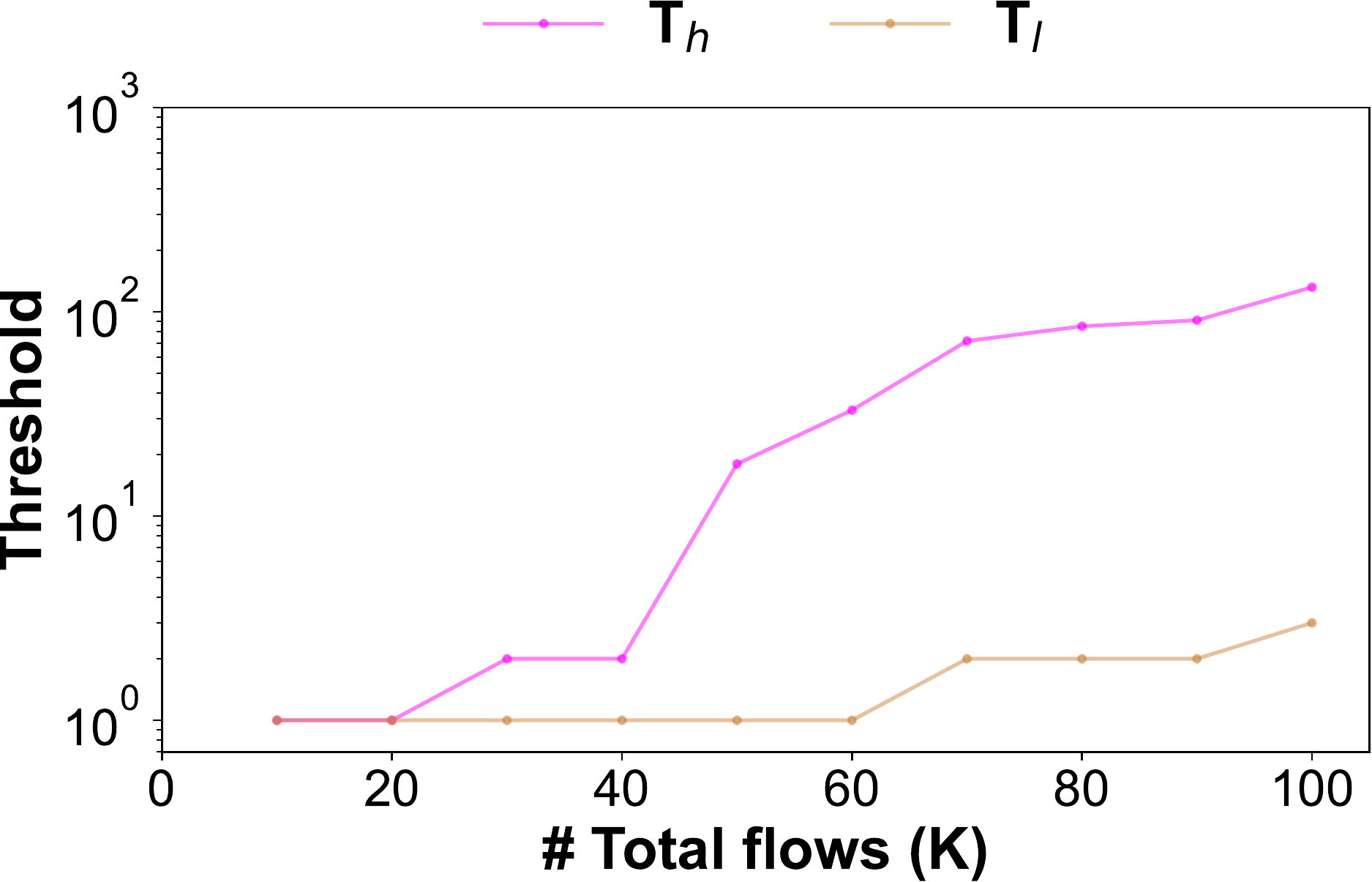}
    \label{fig:HADOOP:flownum:thresh}
    }
    \subfigure[Sample rate.]{ \includegraphics[width=0.23\textwidth]{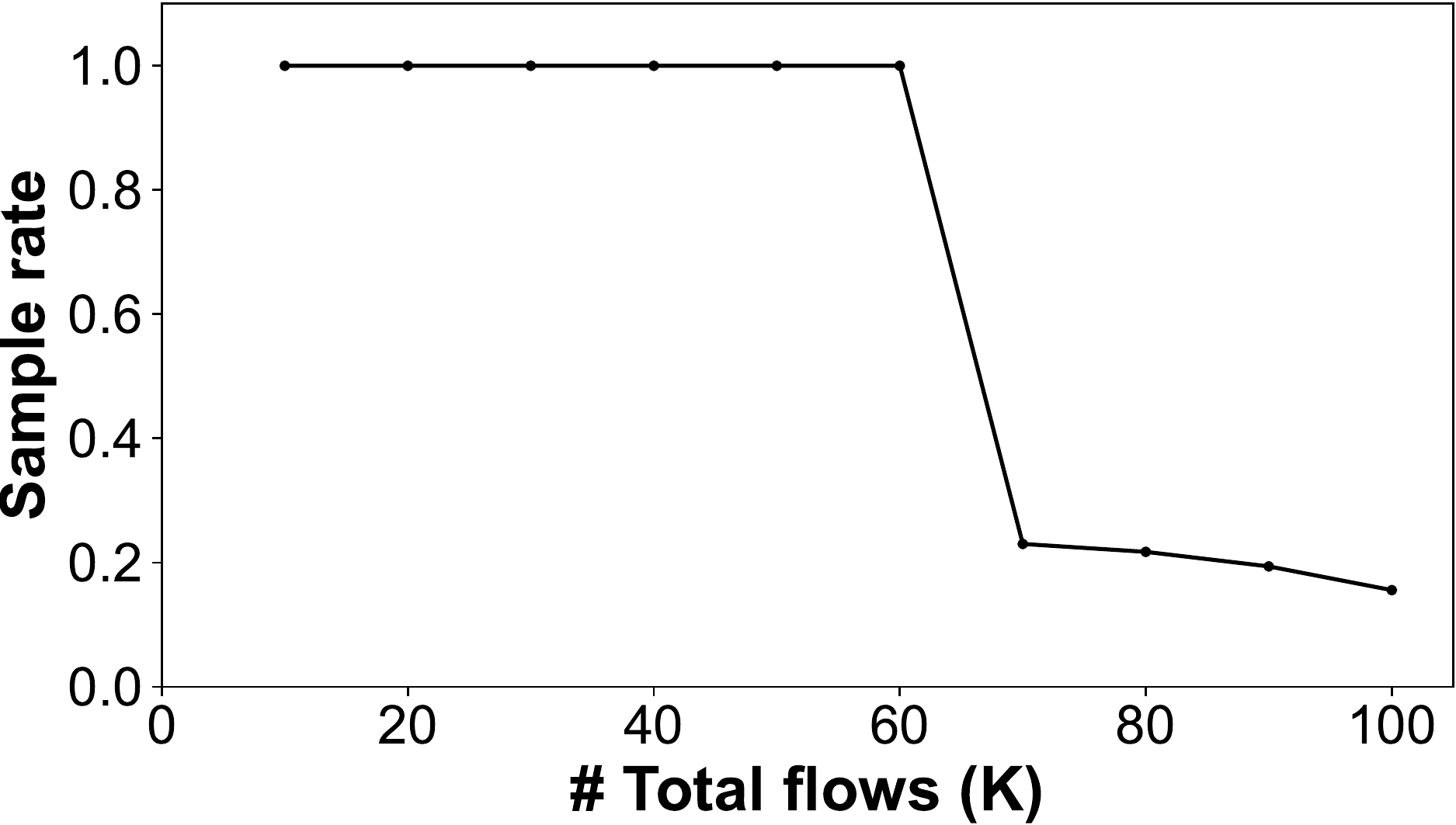}
    \label{fig:HADOOP:flownum:sample}
    }
\caption{Measurement attention \textit{vs.} number of flows on HADOOP workload.}
\label{fig:HADOOP:flownum}
\end{figure*}

\begin{figure*}[ht!]
    \centering
    \subfigure[Memory division.]{
    \includegraphics[width=0.23\textwidth]{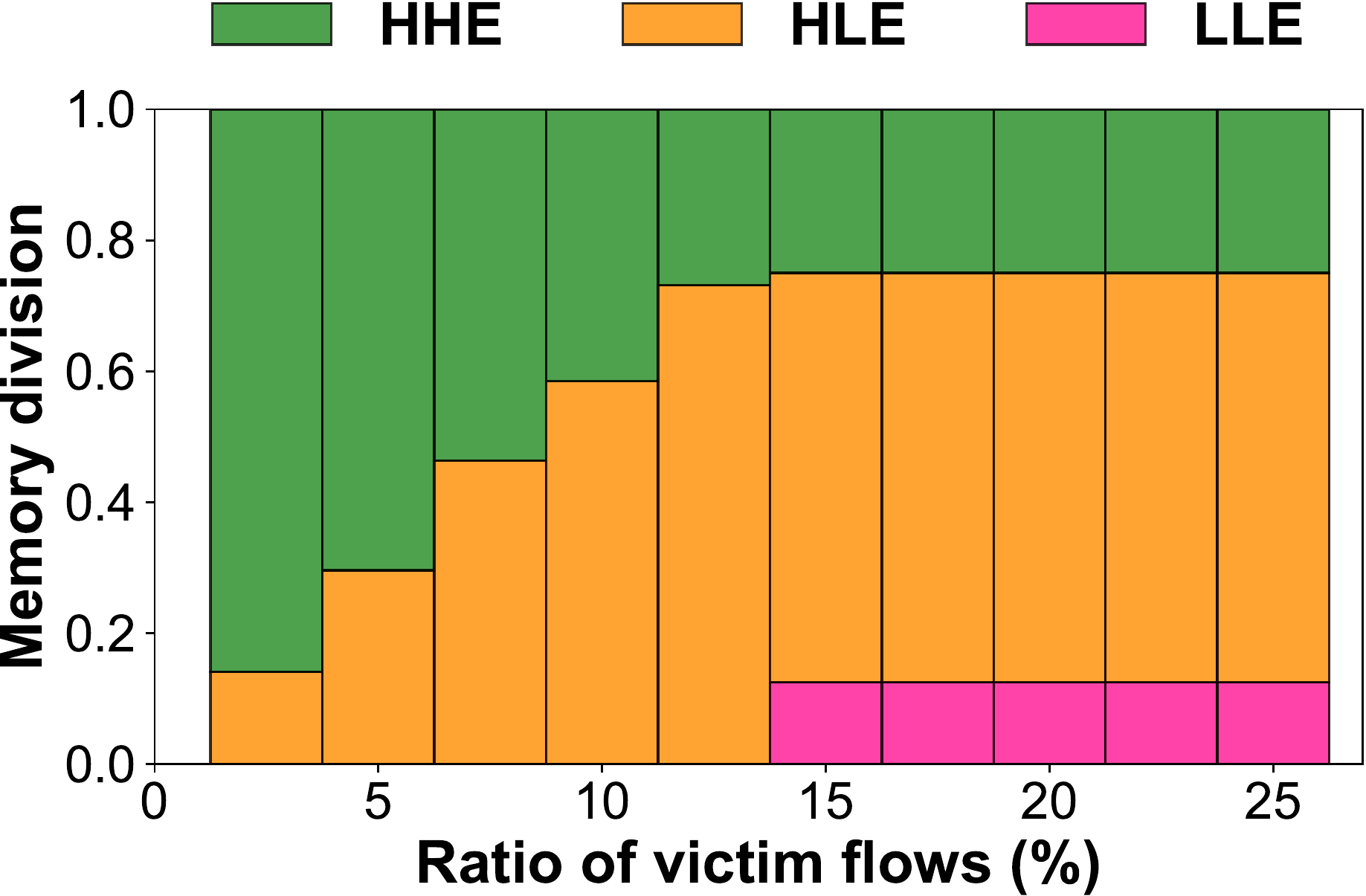}
    \label{fig:HADOOP:lossrate:mem}
    }
    \subfigure[Number of decoded flows.]{
    \includegraphics[width=0.23\textwidth]{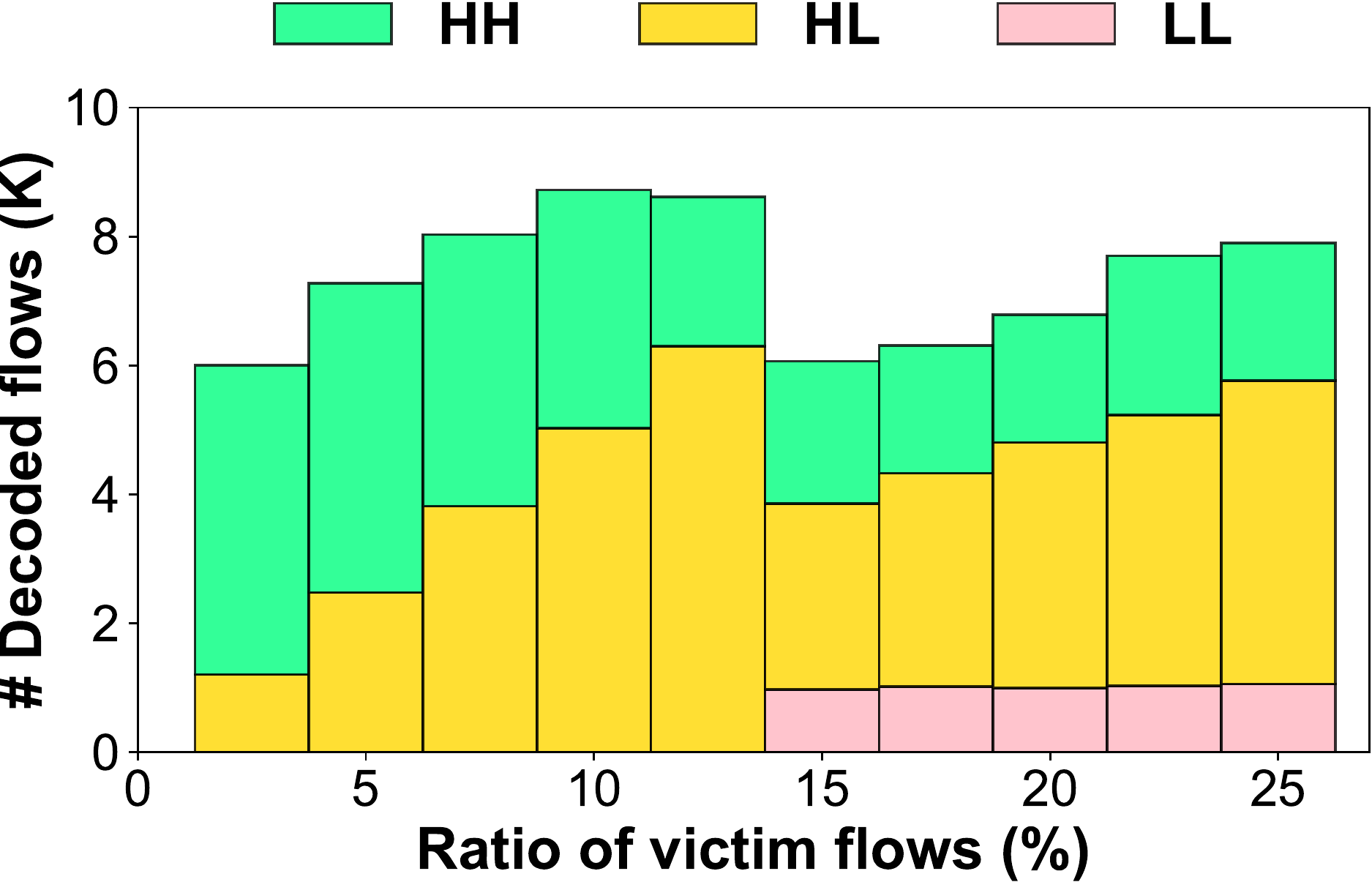}
    \label{fig:HADOOP:lossrate:num}
    }
    \subfigure[Threshold.]{ \includegraphics[width=0.23\textwidth]{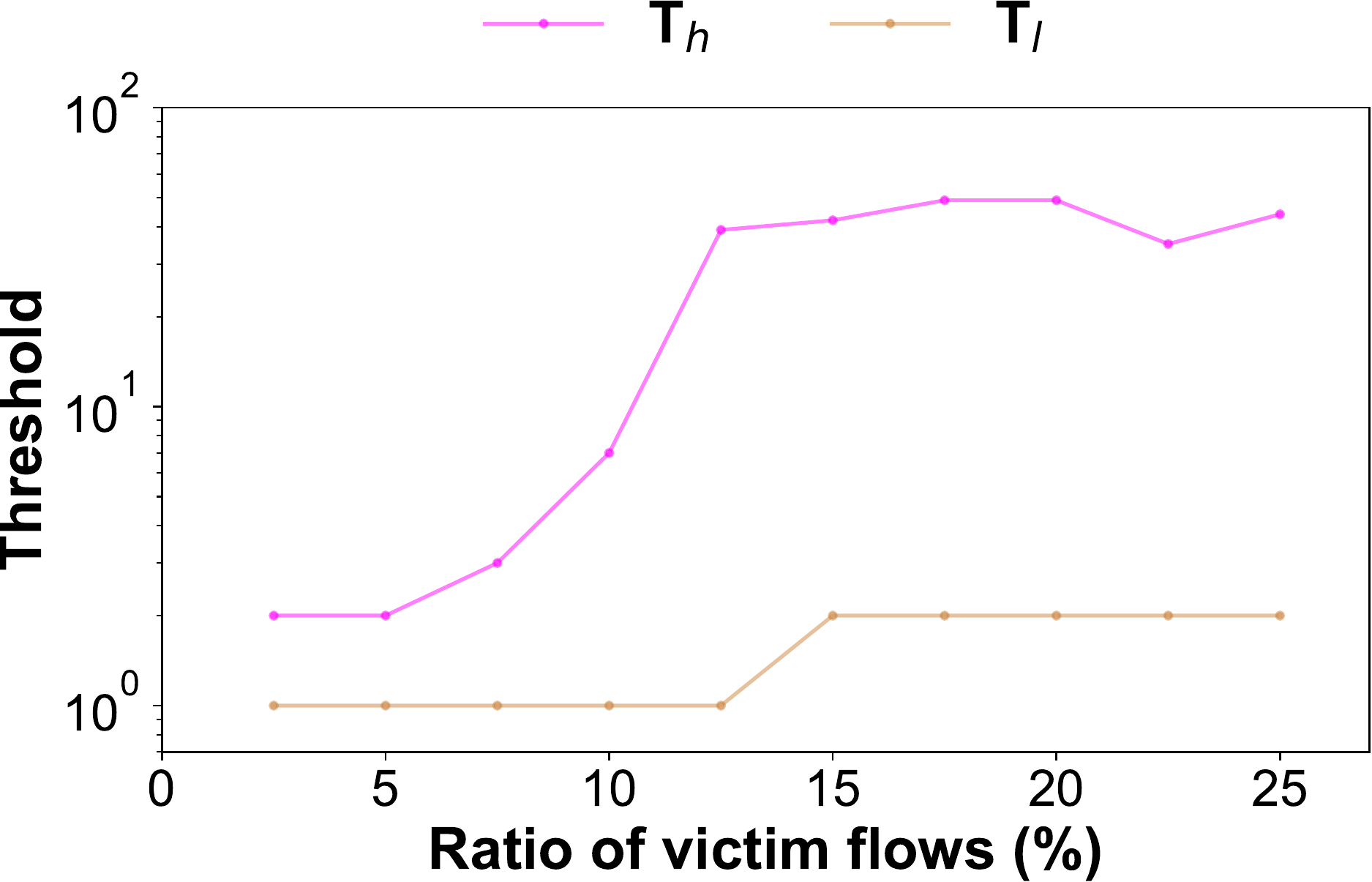}
    \label{fig:HADOOP:lossrate:thresh}
    }
    \subfigure[Sample rate.]{ \includegraphics[width=0.23\textwidth]{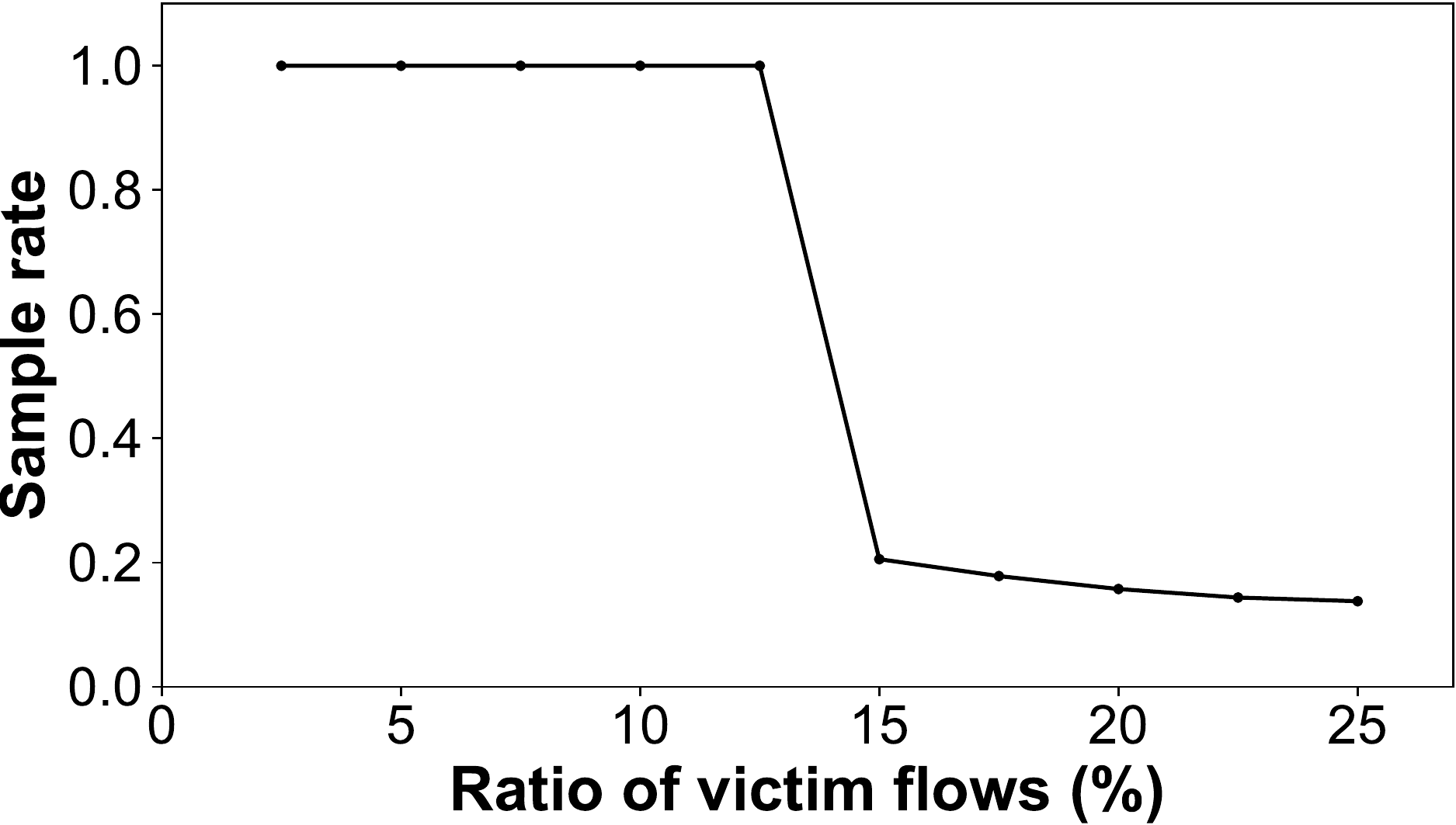}
    \label{fig:HADOOP:lossrate:sample}
    }
\caption{Measurement attention \textit{vs.} ratio of victim flows on HADOOP workload.}
\label{fig:HADOOP:lossrate}
\end{figure*}

\subsection{VL2 Workload}

\bbb{Measurement attention \textit{vs.} number of flows (Figure~\ref{fig:VL2:flownum}):}
As the number of flows increases from $10$K to $20$K, \systemname{} can record all flows and victim flows, and therefore sets both $T_h$ and $T_l$ to $1$.
As the number of flows increases from $30$K to $60$K, \systemname{} allocates more and more memory to HL encoders, and increases $T_h$ to
decrease the number of HH candidates to avoid decoding failure.
As the number of flows increases from $70$K to $100$K, the healthy network state transitions to the ill network state.
\systemname{} allocates memory to LL encoders, increases $T_l$, and decreases the sample rate, so as to to control the number of HLs and sampled LLs.
Meanwhile, \systemname{} keeps increasing $T_h$ to control the number of HH candidates.
Throughout the experiment, \systemname{} maintains the load factor higher than $51\%$.
The load factor is sightly lower, and it is because the distribution of VL2 is highly skewed.
Decreasing the thresholds by $1$ will lead to huge increase in the number of recorded flows, and thus causing decoding failure.

\bbb{Measurement attention \textit{vs.} ratio of victim flows (Figure~\ref{fig:VL2:lossrate}):}
As the ratio of victim flows increases from $2.5\%$ to $12.5\%$, \systemname{} records all victim flows by allocating more and more memory to HL encoders, and increases $T_h$ to decrease the number of HH candidates.
As the ratio of victim flows increases from $15\%$ to $25\%$, \systemname{} cannot record all victim flows and thus the healthy network state transitions to the ill network state.
\systemname{} allocates memory to LL encoders, increases $T_l$, and decreases the sample rate so as to control the number of HLs and sampled LLs.
Meanwhile, because the memory of upstream HH encoders and the number of flows remain unchanged, $T_h$ also remains unchanged.
Throughout the experiment, \systemname{} maintains the load factor higher than $53\%$.
The load factor is sightly lower, and the reason is the same as the former experiment of the the number of flow.

\subsection{HADOOP Workload}

\bbb{Measurement attention \textit{vs.} number of flows (Figure~\ref{fig:HADOOP:flownum}):}
As the number of flows increases from $10$K to $20$K, \systemname{} can record all flows and victim flows, and therefore sets both $T_h$ and $T_l$ to $1$.
As the number of flows increases from $30$K to $60$K, \systemname{} allocates more and more memory to HL encoders, and increases $T_h$ to
decrease the number of HH candidates to avoid decoding failure.
As the number of flows increases from $70$K to $100$K, the healthy network state transitions to the ill network state. 
\systemname{} allocates memory to LL encoders, increases $T_l$, and decreases the sample rate to control the number of HLs and sampled LLs.
Meanwhile, \systemname{} keeps increasing $T_h$ to control the number of HH candidates.
Throughout the experiment, \systemname{} maintains the load factor higher than $47\%$.
The load factor is sightly lower, and it is because the distribution of HADOOP is highly skewed.
Decreasing the thresholds by $1$ will lead to huge increase in the number of recorded flows, and thus causing decoding failure.

\bbb{Measurement attention \textit{vs.} ratio of victim flows (Figure~\ref{fig:HADOOP:lossrate}):}
As the ratio of victim flows increases from $2.5\%$ to $12.5\%$, \systemname{} records all victim flows by allocating more and more memory to HL encoders, and increases $T_h$ to decrease the number of HH candidates.
As the ratio of victim flows increases from $15\%$ to $25\%$, \systemname{} cannot record all victim flows and thus transitions to ill network state.
\systemname{} allocates memory to LL encoders, increase $T_l$, and decreases sample rate, so as to control the number of HHs and HLs.
Throughout the experiment, \systemname{} maintains the load factor higher than $48\%$.
The load factor is sightly lower, and the reason is the same as the former experiment of the the number of flows.

\section{Evaluation on Time/Bandwidth Overhead}

\label{app:timeoverhead}

To evaluate how fast can \systemname{} monitor the network, we evaluate various factors that could affect the setting of epoch length: 1) the time and bandwidth required to collect sketches from edge switches, 2) the time required to respond to different network states, and 3) the time required to reconfigure the \systemname{} data plane.
The central controller only uses one CPU core in evaluation.

\bbb{Time/Bandwidth overhead for collection (Figure~\ref{fig:testbed:bandwidth}):}
Experimental results show that \systemname{} consumes only a small amount of time and bandwidth in collecting all the data structures deployed on edge switches.
\systemname{} takes a total of 11.33ms to collect sketches (refer to Appendix \ref{sec:cpi} for details).
As for bandwidth, when the epoch length is set to $50$ms, the bandwidth overhead for collection is 317Mbps, consuming only 0.8\% bandwidth for the central controller equipped with a 40Gb NIC.

\bbb{Response time to different network states  (Figure~\ref{fig:testbed:response-time}):}
Experimental results show that \systemname{} can always respond to different network states within 30ms.
We count the response time of \systemname{} to each network state previously appeared in Figure \ref{fig:testbed:flownum}-\ref{fig:testbed:lossrate}, where the response time refers to the time interval between the central controller finishing the collection of sketches and the central controller generating the reconfiguration packet\footnote{The central controller sends the reconfiguration packets to edge switches to reconfigure their data planes. Please refer to Appendix \ref{sec:cpi} for details.} for the \systemname{} data plane.
Although the response time does not seem to show a clear trend with the network state, it is mainly determined by the number of HH candidates, because the central controller needs to first extract them from the upstream HH encoders and then reinsert them to the upstream HL encoders.
As shown in Figure \ref{fig:testbed:time:lossrate}, as the ratio of victim flows increases, the response time on all the four workloads decreases because the number of HH candidates decreases.
The response time finally stabilizes because the fixed memory allocation in the ill network state always decodes a similar number of flows.

\begin{figure}[t]
\setlength{\subfigcapskip}{-0.2cm}
\setlength{\abovecaptionskip}{-0.1cm}
\setlength{\belowcaptionskip}{-0.4cm}
    \centering
    \subfigure[Varying number of flows.]{ \includegraphics[width=0.2\textwidth]{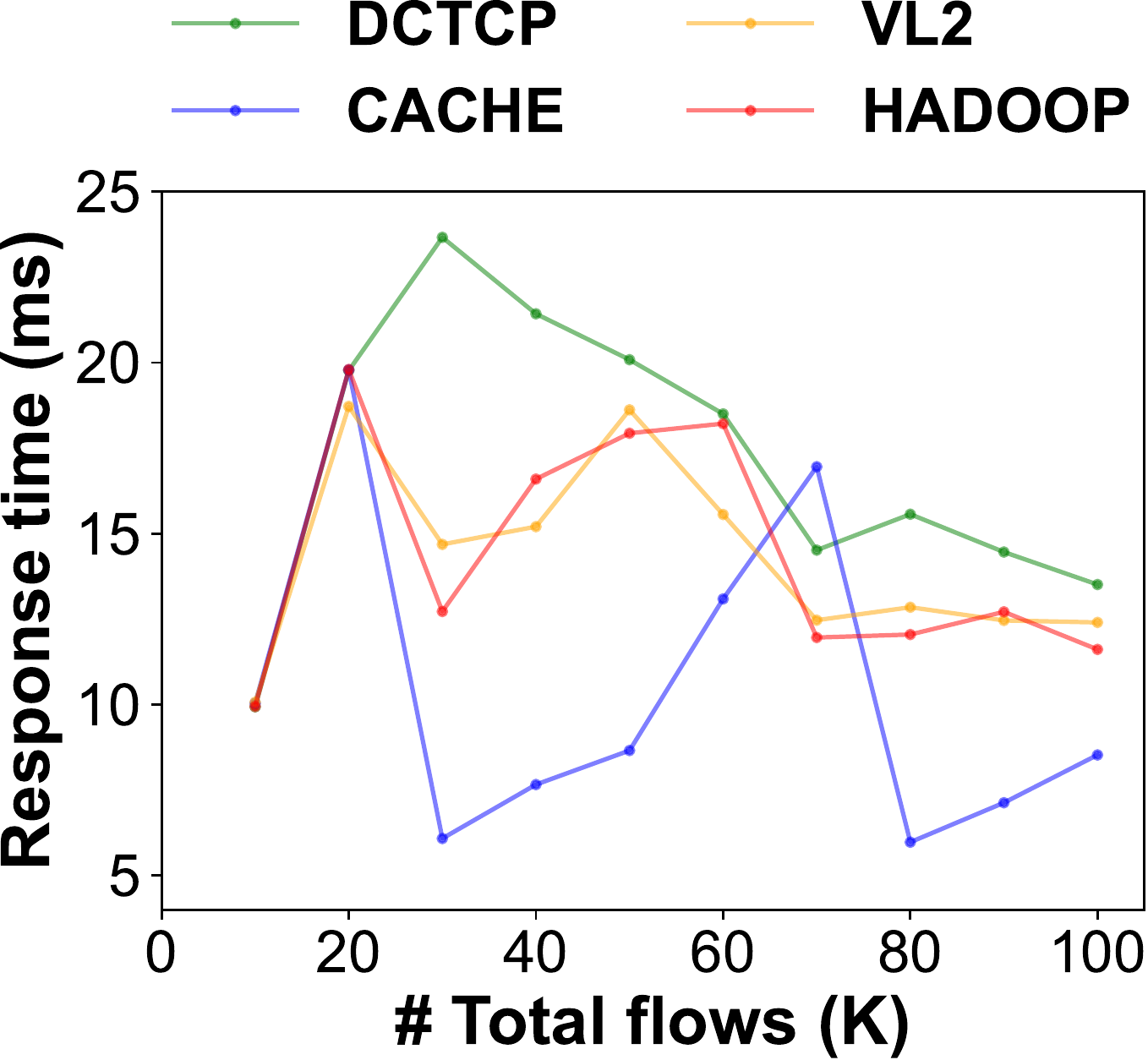}
    \label{fig:testbed:time:flownum}
    }
    \subfigure[Varying ratio of victim flows.]{ \includegraphics[width=0.2\textwidth]{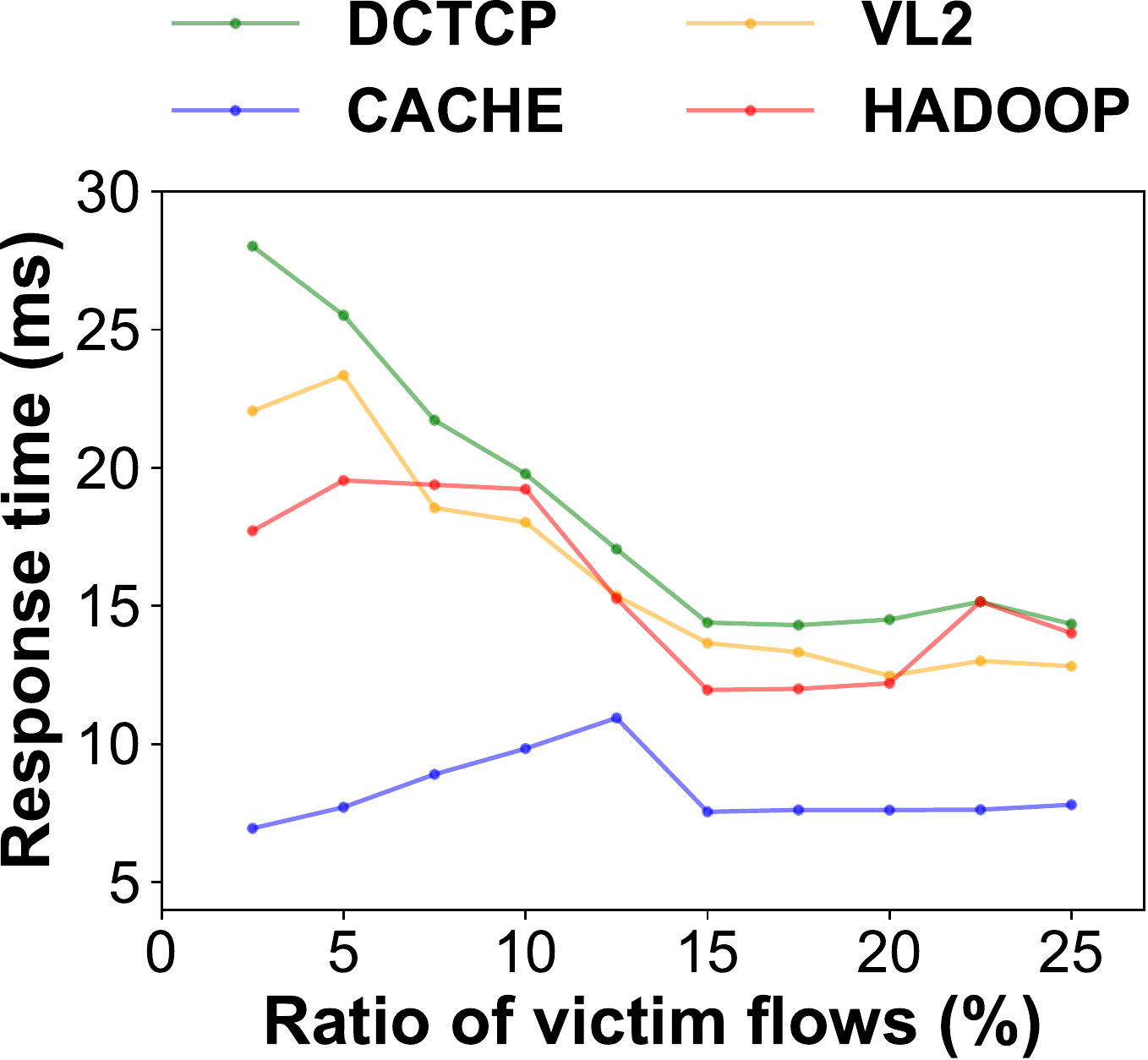}
    \label{fig:testbed:time:lossrate}
    }
    
    \caption{Response time to different network states.}
    \label{fig:testbed:response-time}
\end{figure}

\begin{figure}[t]
\setlength{\subfigcapskip}{-0.0cm}
\setlength{\abovecaptionskip}{-0.0cm}
\setlength{\belowcaptionskip}{-0cm}
\centering
    \begin{minipage}[t]{0.21\textwidth} 
    \centering
    \includegraphics[width=\textwidth,]{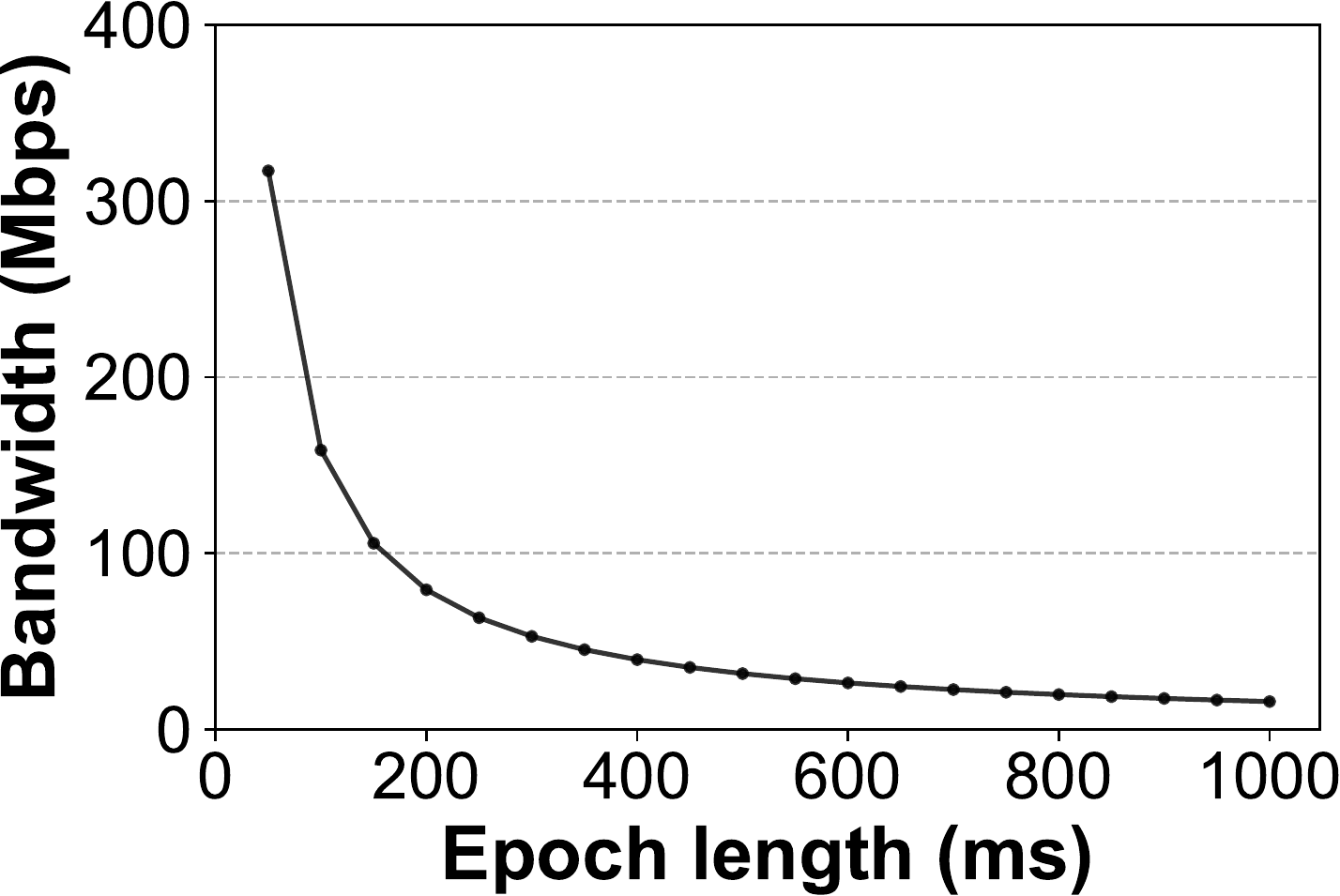}
    
    \caption{Bandwidth.}
    \label{fig:testbed:bandwidth}
    \end{minipage}
    \begin{minipage}[t]{0.22\textwidth}
    \centering
    \includegraphics[width=0.9\textwidth,]{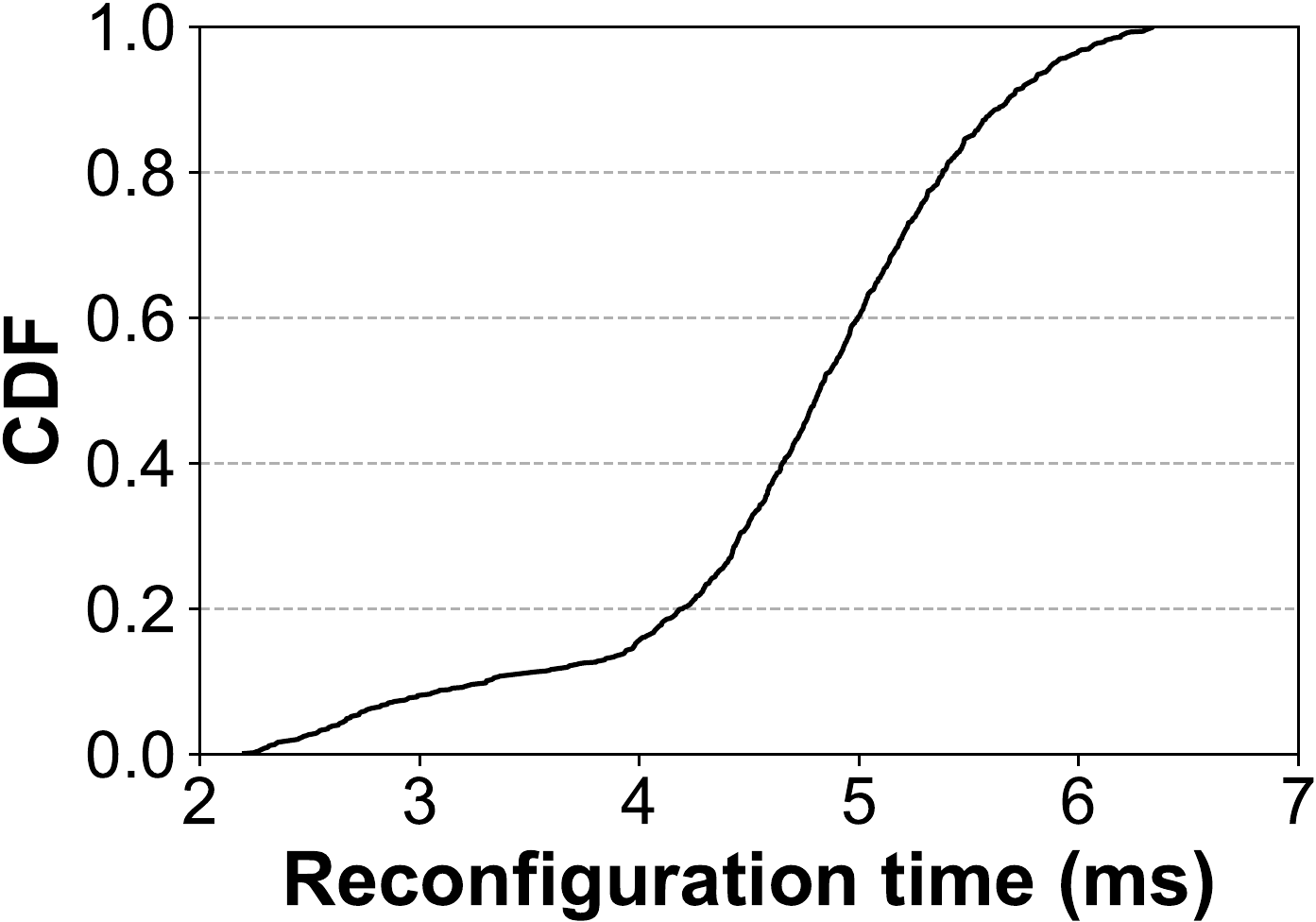}
    
    \caption{Reconfiguration time.}
    \label{fig:testbed:cdf}
    \end{minipage}
\end{figure}

\bbb{CDF of reconfiguration time  (Figure~\ref{fig:testbed:cdf}):}
Experimental results show that it takes 2$\sim$7ms to reconfigure the \systemname{} data plane. 
The central controller sends 10K reconfiguration packets with random configuration of the \systemname{} data plane to each edge switch, and we count the time for an edge switch to execute the reconfiguration.
We find 60\% of reconfigurations take less than 5ms.
The difference in time consumption is mainly because different reconfigurations require updating different numbers of TCAM entries to the switch data plane for supporting dynamic memory allocation (refer to Appendix \ref{sec:dataimpl} for details).

Adding up the above all time consumption, we find that the overall time consumption is less than 50ms.
This verifies that \systemname{} can monitor the network every 50ms on our testbed.
Considering that 1) the central controller only uses one CPU core in experiments and 2) monitoring the network every 50ms only consumes 0.8\% bandwidth of a 40Gb NIC, we believe \systemname{} can easily scale to monitor a much larger network with a shorter epoch length, requiring only one server as the central controller.

%% file: NSDI2022_9/Appendix/Related_work.tex
\presec \section{Other Related Work} \postsec
\label{app:related}
Besides the solutions in Section \ref{sec:relatedworks}, there are still three kinds of solutions relevant to network measurement.
\begin{itemize}[leftmargin=*,parsep=0pt,itemsep=0pt,topsep=2pt,partopsep=2pt]

    \item \textit{Sampling-based solutions:}
    These solutions collect desired statistics from a subset of network traffic through packet sampling, including Csamp \cite{csamp2008}, NetFlow \cite{netflow2004}, sFlow \cite{sflow2001}, EverFlow \cite{everflow2015}, and more \cite{snmp,duffield2001trajectory,sekar2010revisiting,opensample2014,li2019large,roy2015inside,planck2014,nikolopoulos2019retroactive,yu2019dshark}.
    %
    %
    While sampling solutions significantly reduce the bandwidth overhead through sampling, they cannot well handle packet loss tasks as only sampled packets are measured, and thus fail to meet efficiency requirement.
    
    \item \textit{Programmable-switch-assisted solutions:}
    Beside packet loss detection, some solutions leverage the advanced features and capabilities of programmable switches to monitor micro-bursts \cite{joshi2018burstradar}, perform queue measurement \cite{chen2019fine,lei2022printqueue,sonchack2018turboflow}, and more \cite{laraba2021mitigating,wang2022closed,molero2022fast,holterbach2019blink}.
    
    \item \textit{Host-based solutions:}
    Due to the flexibility, abundant resources, and high visibility to flow-level statistics of end-hosts, these solutions are typically designed for inferring the existences, locations, and root causes of specific network events or network failures.
Typical solutions either send tailored probes into the network \cite{pingmesh2015,dhamdhere2018inferring,netbouncer2019,adams2016netnorad,aubry2016scmon,dhamdhere2007netdiagnoser,peng2017detector,simon2019} or analyze the performance of protocol stack or other I/O \cite{roy2017passive,choffnes2010crowdsourcing,mysore2014gestalt,trumpet2016,arzani2016taking,confluo2019,zeno2019,yuan2017quantitative,007-2018,gong2020microscope}.
Besides, some solutions further leverage switches to perform measurement \cite{jeyakumar2014millions,liu2016mozart} or record forwarding paths \cite{tammana2016simplifying,tammana2018distributed}.
\systemname{} can complement these solutions as \systemname{} provides flow-level statistics with high accuracy.
Take 007 \cite{007-2018} as an instance.
Network operators can replace the TCP monitoring agent that detects TCP retransmissions in 007 with \systemname{}.
After the replacement, 007 can monitor packet losses of TCP flows as well as packet losses of flows of other protocols.
Such extra visibility can help 007 better locate the link failures.

\end{itemize}